\documentclass[10pt,a4paper]{article}

\usepackage{amsmath,amsgen,latexsym}
\usepackage{amstext,amssymb,amsfonts,latexsym}
\usepackage{theorem}
\usepackage{pifont}
\usepackage[dvips]{graphics,epsfig}

 \newcommand{\bs}{\bigskip}
 \newcommand{\ms}{\medskip}
 \newcommand{\n}{\noindent}
 \newcommand{\s}{\smallskip}
 \newcommand{\hs}[1]{\hspace*{ #1 mm}}
 \newcommand{\vs}[1]{\vspace*{ #1 mm}}

 \newcommand{\setempty}{\mathrm{\O}}
 \newcommand{\real}{\mathbb{R}}

 \newcommand{\nat}{\mathbb{N}}
 
 \newcommand{\integer}{\mathbb{Z}}

 \newcommand{\complex}{\mathbb{C}}

 \newcommand{\prob}{{\mathrm{Prob}}}

 \newcommand{\ie}{\textrm{i.e.},\hspace*{2mm}}
 \newcommand{\eg}{\textrm{e.g.},\hspace*{2mm}}

 \newcommand{\HH}{{\cal H}}

 \newcommand{\YY}{{\cal Y}}

 \newcommand{\p}{\mathrm{P}}

 \newcommand{\bpp}{\mathrm{BPP}}

 \newcommand{\bqp}{\mathrm{BQP}}

 \newcommand{\onebqlin}{1\mbox{-}\mathrm{BQLIN}}

 \newcommand{\onebplin}{1\mbox{-}\mathrm{BPLIN}}

 \newcommand{\onedlin}{1\mbox{-}\mathrm{DLIN}}

 \newcommand{\oneqfa}{1\mathrm{QFA}}
 \newcommand{\onerfa}{1\mathrm{RFA}}
 \newcommand{\oneqfastar}{1\mathrm{QFA}^{\!\diamond}\!}

 \newcommand{\reg}{\mathrm{REG}}
 \newcommand{\cfl}{\mathrm{CFL}}
 \newcommand{\dcfl}{\mathrm{DCFL}}

 \newcommand{\matrices}[4]{\left( \begin{array}{cc} #1 & #2 \\%
      #3 & #4   \end{array}\right)}

 \newcommand{\ninematrices}[9]{\left( \begin{array}{ccc}
      #1 & #2 & #3 \\%
      #4 & #5 & #6 \\%
      #7 & #8 & #9    \end{array}\right)}

 \newcommand{\smallcomb}[2]{\left(\:\begin{subarray}{c} #1 \\%
      #2 \end{subarray} \right)}

\theoremstyle{plain}
\theoremheaderfont{\bfseries}
 \newtheorem{theorem}{Theorem}[section]
 \newtheorem{lemma}[theorem]{Lemma}
 \newtheorem{proposition}[theorem]{Proposition}
 \newtheorem{corollary}[theorem]{Corollary}

 {\theorembodyfont{\rmfamily}
  }
 {\theorembodyfont{\rmfamily} }
 {\theorembodyfont{\rmfamily} }

 \newtheorem{claim}{Claim}

 \newenvironment{proof}{\par \noindent
            {\bf Proof. \hs{2}}}{\hfill$\Box$ \vspace*{3mm}}

 \newenvironment{proofof}[1]{\vspace*{5mm} \par \noindent
         {\bf Proof of #1.\hs{2}}}{\hfill$\Box$ \vspace*{3mm}}

 \newcommand{\qed}{\hfill$\Box$  \vspace*{3mm}}

 \newcommand{\ceilings}[1]{\lceil #1 \rceil}
 \newcommand{\floors}[1]{\lfloor #1 \rfloor}

 \newcommand{\qubit}[1]{| #1 \rangle}
 \newcommand{\bra}[1]{\langle #1 |}
 \newcommand{\ket}[1]{| #1 \rangle}
 \newcommand{\braket}[2]{\langle #1 | #2 \rangle}

\newcommand{\ignore}[1]{}

\newcommand{\track}[2]{[\:\begin{subarray}{c} #1 \\%
      #2 \end{subarray} ]}
\newcommand{\cent}{{|}\!\!\mathrm{c}}
\newcommand{\dollar}{\$}

\begin{document}
\pagestyle{plain}
\setcounter{page}{1}

\begin{center}
{\Large {\bf One-Way Reversible and Quantum Finite Automata with Advice}}\footnote{An extended abstract appeared in the Proceedings of
the 6th International Conference on Language and Automata Theory and Applications (LATA 2012), March 5--9, 2012, A Coru\~{n}a, Spain,
 Lecture Notes in Computer Science, Springer-Verlag, Vol.7183, pp.526--537, 2012.  This work was partly supported by the Mazda Foundation and the Japanese Ministry of Education, Science, Sports, and Culture.} \bs\ms\\

{\sc Tomoyuki Yamakami}\footnote{Affiliation: Department of
Information Science, University of Fukui, 3-9-1 Bunkyo, Fukui 910-8507,  Japan} \bs\\
\end{center}

\begin{quote}
\n{\bf Abstract:}
We examine the characteristic features of reversible and quantum computations in the presence of supplementary external information, known as advice.
In particular, we present a simple, algebraic characterization of languages recognized by one-way reversible finite automata augmented with deterministic advice. With a further  elaborate argument, we prove
a similar but slightly weaker result for bounded-error one-way quantum finite automata with advice.
Immediate applications of those properties lead to
containments and separations among various language families when they are assisted by appropriately chosen advice. We further demonstrate the power and limitation of randomized advice and quantum advice when they are given to one-way quantum finite automata.

\s

\n{\bf Keywords:} reversible finite automaton, quantum finite automaton,
regular language, context-free language, randomized advice,
quantum advice, rewritable tape
\end{quote}

\section{Background, Motivation, and Challenge}\label{sec:introduction}

In a wide range of the past literature, various notions of supplemental external  information have been sought to empower automated computing devices and the  power and limitation of such extra information have been  studied extensively. In the early 1980s, Karp and Lipton  \cite{KL82} investigated a role of simple external information, known as {\em (deterministic) advice}, which encodes  useful  data, given in parallel with a standard input, into a single string (called an {\em advice string}) depending only on the size of the input.
Such advice has been since then widely used for polynomial-time Turing machines, particularly, in connection to non-uniform circuit families.
When {\em one-way deterministic finite automata} (or 1dfa's, in short) are  concerned,  Damm and Holzer \cite{DH95}
first studied such advice whose advice string is given ``next to'' an ordinary input string written on a single input tape. By contrast, Tadaki, Yamakami, and Li \cite{TYL10} provided 1dfa's with advice ``in dextroposition with''
an input string, simply by splitting an input tape into two tracks,  in which the upper track
carries a given input string and the lower track holds an advice string.
Using the latter model of advice, a series of recent studies \cite{TYL10,Yam08,Yam09,Yam10,Yam11} concentrating on the strengths and weaknesses of the advice have unearthed advice's delicate roles for various types of underlying one-way finite automata.
Notice that these ``advised'' automaton models have immediate connections to other important fields, including
one-way communication, random access coding, two-player zero-sum games, and pseudorandom generator.
Two central questions concerning the advice are: how can we encode necessary information into a piece of advice before a computation starts and, as a computation proceeds step by step,  how can we decode and utilize such information stored inside the advice?
Whereas  there is rich literature on the power and limitation of advice  for a model of polynomial-time quantum Turing machine (see, for instance, \cite{Aar05,NY04b,Raz09}), disappointingly, except for the aforementioned studies, little has been known to date for the roles of advice when it is given to finite automata.
To promote our understandings of the advice, we intend to expand a scope of our study from 1dfa's to one-way reversible and quantum finite automata.

{}From theoretical as well as practical interests, we wish to examine two machine models realizing reversible and quantum computations, known as (deterministic) reversible finite automata and quantum finite  automata.
Since our objective is to analyze the roles of various forms of advice, we
want to choose simpler models for reversible and quantum computations in order to make our analysis easier.
Of various types of such automata, we intend to initiate our study by limiting our focal point within one of the simplest automaton models:  {\em one-way (deterministic) reversible finite automata} (or 1rfa's, in short) and
{\em one-way measure-many quantum finite automata}
(or 1qfa's, thereafter).
Although these particular models are known to be strictly weaker in computational power than even regular languages, they still embody an essence of reversible and quantum mechanical computations for which the advice can play a significantly important role.
Our 1qfa scans each cell of a read-only input tape by moving a single tape head only in one direction (without stopping) and performs a {\em (projective) measurement} immediately after every head move, until the tape head eventually scans the right endmarker.
{}From a theoretical perspective, the 1qfa's having more than $7/9$ success probability are essentially as powerful as 1rfa's \cite{AF98}, and therefore 1rfa's are important part of 1qfa's.
As this fact indicates, for bounded-error 1qfa's, it is not always possible to make a sufficient amplification of success probability. This is merely one of many intriguing features that make an analysis of the 1qfa's distinct from that of polynomial-time quantum Turing machines, and it is such remarkable features that have kept stimulating our research since their
introduction in late 1990s.
Let us recall some of the numerous unconventional features that
have been revealed in an early period of intensive study of the 1qfa's.
As Ambainis and Freivalds \cite{AF98} demonstrated, certain quantum finite automata can be built more state-efficiently than deterministic finite automata.  However, as Kondacs and Watrous \cite{KW97} proved, not all  regular languages are  recognized  with bounded-error probability by 1qfa's.  Moreover,
by Brodsky and Pippenger \cite{BP02}, no bounded-error 1qfa recognizes  languages accepted by minimal finite automata that lack a so-called {\em partial order condition}. The latter two facts suggest that
the language-recognition power
of 1qfa's is hampered by their own inability to generate useful quantum states from input information alone.

We wish to understand how advice can change the nature of 1rfa's and 1qfa's.
For a bounded-error 1qfa, for instance, an immediate advantage of taking  advice is the elimination of the {\em both} endmarkers placed on the 1qfa's  read-only input tape.
Beyond such a clear advantage, however, there are numerous challenges lying in the study of the roles of the advice.
To analyze the behaviors of ``advised'' 1qfa's as well as ``advised''
1rfa's, we must face those challenges and eventually overcome them.
Generally speaking, the presence of advice tends to make an analysis of underlying computations quite difficult and it often demands quite different kinds of proof techniques. As a quick example, a standard {\em pumping lemma}---a typical proof technique that showcases the non-regularity of a given language---is not quite serviceable to advised computations; therefore, we have already developed other useful tools (e.g.,  a swapping lemma \cite{Yam08}) for them.
In similar light, certain advised 1qfa's fail to meet the aforementioned partial order condition (Lemma \ref{partial-order-cond}) and, unfortunately, this fact makes a proof technique of Kondacs and Watrous \cite{KW97} inapplicable to, for example,  a class separation between advised regular languages and languages accepted by bounded-error advised 1qfa's.

To overcome foreseen difficulties in out study, our first task must be to lay out a necessary ground work in order to (1) capture fundamental features of those automata when advice is given to boost their language-recognition power and (2) develop methodology necessary to lead to  collapses and separations of advised language families.
It is the difficulties surrounding the advice for 1qfa's that motivate us to seek different kinds of proof techniques.


In Sections \ref{sec:QFA-theorem} and \ref{sec:N-S-condition-1RFA}, we will prove two main theorems. In the first main theorem (Theorem \ref{oneqfa-character}), with an elaborate argument using a new metric vector space called $\YY_{\HH}$, we will show a machine-independent, algebraic necessary condition for languages to be recognized by bounded-error 1qfa's that take appropriate deterministic advice.  In the second theorem (Theorem \ref{onerfa-characterization}) for 1rfa's augmented with deterministic advice, we will give a completely machine-independent, algebraic necessary and sufficient condition. These two conditions exhibit certain behavioral characteristics of 1rfa's and 1qfa's when appropriate advice is prepared. Our proof techniques for 1qfa's, for instance,  are quite different from the previous work \cite{AF98,ANTV02,BP02,KW97,MC00}.
Applying those theorems further, we can prove several class separations among advised language families.
These separations indicate, to some extent, inherent strengths and weaknesses of reversible and quantum computations even in the presence of advice.

Another important revelation throughout our study  in the field of reversible and quantum computation is the excessive power of {\em randomized advice} over deterministic advice. In randomized advice \cite{Yam10}, advice strings of a fixed length are generated at random according to a pre-determined  probability distribution so that a finite automaton looks like ``probabilistically''  processing those generated advice strings together with a standard input.
{\em Quantum advice} further extends randomized advice; however,
our current model of 1qfa with ``read-only'' advice strings inherently has a structural limitation, which prevents quantum advice from being more resourceful than randomized advice.
Hence, we will engage in another challenging task of seeking a ``simple''  modulation of the existing
1qfa's in order to  utilize  more effectively quantum information stored in quantum advice. We will discuss in Section \ref{sec:rewritable-1qfa} how to remedy the deficiency of the current 1qfa model and which direct implications such a remedy leads to. Similar treatments were already made for various types of one-way quantum finite automata in, e.g., \cite{Pas00,YFS+12}. The model of 1qfa itself has been also extended in various directions, including {\em interactive proof systems}  \cite{NY04a,NY09,NY14,Yam14}.


\begin{figure}[t]
\begin{center}
\includegraphics*[width=8.0cm]{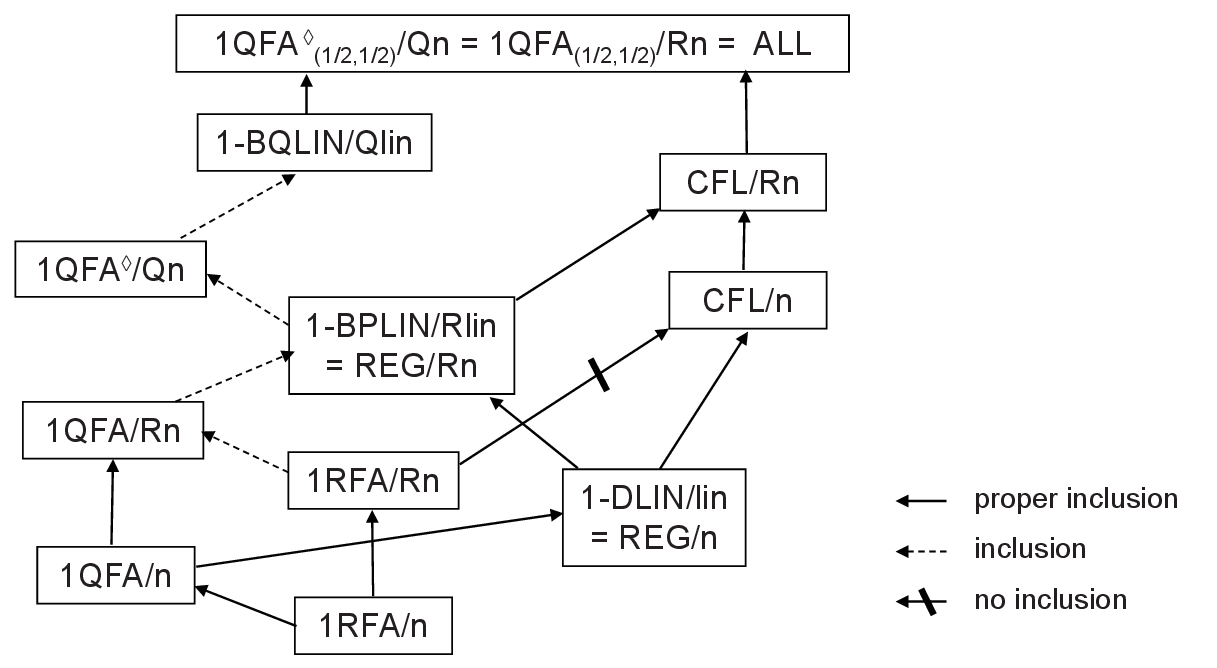}
\caption{{\small A hierarchy of advised language families. All containments and separations associated with quantum finite automata, reversible automata, and quantum Turing machines are newly proven in this paper. All dotted arrows indicate that the associated containments are not known to be proper.}}\label{fig:hierarchy}
\end{center}
\end{figure}


\paragraph{A Quick Overview of Relations among Advised Language Families}
As summarized in Fig.~\ref{fig:hierarchy}, we obtain containments and separations of new advised language families in direct comparison with   existing classical advised language families.
Our main theorems are particularly focused on two language families:
the family $\onerfa$ of all languages accepted by 1rfa's and
the family $\oneqfa$ of languages recognized by 1qfa's with bounded-error probability.
Associated with these language families, we will introduce their corresponding advised language families\footnote{To clarify the types of advice, we generally use the following specific suffixes. The suffixes ``$/n$'' and ``$/Rn$'' respectively indicate the use of deterministic advice and randomized advice of input size, whereas  ``$/lin$'' and ``$/Rlin$'' respectively indicate the use of deterministic advice and randomized advice of linear size. Similarly, ``$/Qn$'' and ``$/Qlin$'' indicate the use of quantum advice of input size and of linear size.}: $\onerfa/n$, $\onerfa/Rn$, $\oneqfa/n$,  $\oneqfa/Rn$, and $\oneqfastar/Qn$, except that  $\oneqfastar/Qn$ uses a slightly relaxed 1qfa model\footnote{Such a relaxation does not affect classical advice families. For example, $\reg^{\diamond}\!/n=\reg/n$ holds.} discussed earlier.
In Fig.~\ref{fig:hierarchy},
``$\mathrm{ALL}$'' indicates the collection of {\em all} languages.
Language families  $\cfl$ (context-free) and $\reg$ (regular) are respectively based on classical one-way finite automata with stacks and with no stacks. Moreover, language families $\onedlin$ (deterministic), $\onebplin$ (bounded-error probabilistic), and $\onebqlin$ (bounded-error quantum) \cite{TYL10}, which are  viewed respectively as ``scaled-down'' versions of the well-known complexity classes $\p$, $\bpp$, and $\bqp$, are based on the models of one-tape  one-head two-way off-line Turing machines running in ``linear time,'' in the sense of a so-called {\em  strong definition} of running time (see \cite{Mic91,TYL10}).
Supplementing various types of advice to those families introduces the following advised language families: $\reg/n$ \cite{TYL10}, $\cfl/n$ \cite{Yam08}, $\reg/Rn$ \cite{Yam10}, $\cfl/Rn$ \cite{Yam10}, $\onedlin/lin$ \cite{TYL10}, $\onebplin/Rlin$  \cite{Yam10}, and $\onebqlin/Qlin$. The interested reader may refer to \cite{TYL10,Yam09,Yam11} for other advice language families not listed in Fig.~\ref{fig:hierarchy}.

\section{Basic Terminology}\label{sec:basics}

Let $\real$ (resp., $\complex$) denote the set of all {\em real numbers} (resp., all {\em complex numbers}). Moreover, we write $\nat$ for the set of all {\em natural numbers} (\ie nonnegative integers) and set $\nat^{+}=\nat-\{0\}$. Given any two integers $m,n$ with $m\leq n$, the notation $[m,n]_{\integer}$ expresses the integer interval $\{m,m+1,m+2,\ldots,n\}$. In particular, we set $[n]$ to be $[1,n]_{\integer}$ if $n\geq1$.
An {\em alphabet} is a nonempty finite set of ``symbols'' (or ``letters''). For any alphabet $\Sigma$, a {\em string} (or a word) over $\Sigma$ is a finite sequence of symbols in $\Sigma$ and $|x|$ denotes the {\em length} of string $x$ (\ie the total number of occurrences of symbols in $x$). Let $\Sigma^*$ be composed of all strings over $\Sigma$. The {\em empty string} is always
denoted by $\lambda$.
For any string $x$ and any number $i\in[0,|x|]_{\integer}$, $Pref_{i}(x)$ expresses a unique string $u$ satisfying both $x=uv$ and $|u|=i$ for a certain string $v$; moreover, we set $Pref_{i}(x)=x$ for any index $i>|x|$.


For convenience, we abbreviate as 1dfa (resp., 1npda) a {\em one-way deterministic finite automaton} (resp., {\em one-way nondeterministic pushdown automaton}).
For ease of our later analysis, we ``explicitly'' assume, unless otherwise stated, that (1) every finite automaton is equipped with a single read-only input tape on which
each input string is initially surrounded by two endmarkers (the left endmarker $\cent$ and the right endmarker $\dollar$), (2)  every finite automaton has a single tape head that is initially situated at the left endmarker, and (3) every finite automaton moves its tape head rightward {\em without stopping}  until the automaton finally enters any ``halting'' inner state.
For a later reference, we formally define a 1dfa as a sextuple $M = (Q,\Sigma,\delta,q_0,Q_{acc},Q_{rej})$, where $Q$ is a finite set of {\em inner states}, $\Sigma$ is an input alphabet, $\delta:Q\times\check{\Sigma}\rightarrow Q$ is a transition function,\footnote{In this deterministic case, it may be more appropriate to define $\delta$ as a map from $(Q-Q_{halt})\times \check{\Sigma}$ to $Q$. We take the current definition because we can extend it to the quantum case in Section \ref{sec:QFA/n}.}
$q_0$ ($\in Q$) is a unique initial state, $Q_{acc}$ ($\subseteq Q$) is a set of accepting states, and $Q_{rej}$ ($\subseteq Q-Q_{acc}$) is a set of rejecting states, where $\check{\Sigma}$ denotes the set  $\Sigma\cup\{\cent,\dollar\}$ of {\em tape symbols}. For convenience, we also set $Q_{halt} = Q_{acc}\cup Q_{rej}$ and $Q_{non} = Q - Q_{halt}$.  Inner states in $Q_{halt}$ (resp., $Q_{non}$) are generally called {\em halting} (resp., {\em non-halting}) {\em states} and, whenever $M$ enters any halting state, it must {\em halt} immediately.
An {\em extended transition function} induced from $\delta$ is defined as $\hat{\delta}(q,\lambda) =q$ and $\hat{\delta}(q,x\sigma) =\delta(\hat{\delta}(q,x),\sigma)$ for any $x\in(\check{\Sigma})^*$ and $\sigma\in\check{\Sigma}$.

To introduce a notion of {\em (deterministic) advice} that is fed to finite automata beside input strings,  we adopt the ``track'' notation from  \cite{TYL10}. For two symbols $\sigma\in\Sigma$ and $\tau\in\Gamma$, where $\Sigma$ and $\Gamma$ are two alphabets, the notation $\track{\sigma}{\tau}$ expresses a new symbol made up of $\sigma$ and $\tau$. Graphically, this new symbol is written on a single input tape cell, which is split into two tracks  whose upper track contains $\sigma$ and lower track contains
$\tau$. Since the symbol $\track{\sigma}{\tau}$ is in one tape cell, a tape
head scans the two track symbols $\sigma$ and $\tau$ simultaneously.
When two strings $x$ and $y$ are of the same length $n$, the notation $\track{x}{y}$ denotes a concatenated string $\track{x_1}{y_1}\track{x_2}{y_2}\cdots\track{x_n}{y_n}$, provided that $x=x_1x_2\cdots x_n\in\Sigma^n$ and $y=y_1y_2\cdots y_n\in\Gamma^n$. In particular, when $n=0$, we conveniently identify $\track{\lambda}{\lambda}$
with $\lambda$.
Using this track notation, we define $\Sigma_{\Gamma}$ to be  a new alphabet $\{\track{\sigma}{\tau}\mid \sigma\in\Sigma,\tau\in\Gamma\}$ induced from the two alphabets $\Sigma$ and $\Gamma$.
An {\em advice function} $h$ is a function mapping $\nat$ to $\Gamma^*$, where $\Gamma$ is particularly called an {\em advice alphabet}, but $h$ is
not required to be ``computable.'' Such a function is further called {\em length-preserving} whenever $|h(n)|=n$ holds for every length $n\in\nat$.
The advised language family
$\reg/n$ of Tadaki, Yamakami, and Lin \cite{TYL10} is the family of all languages $L$ over certain alphabets $\Sigma$ satisfying the following condition: there exist an advice alphabet $\Gamma$, a 1dfa $M$ (whose input alphabet is $\Sigma_{\Gamma}$), and a length-preserving advice function $h:\nat\rightarrow\Gamma^*$ such that, for every string $x\in\Sigma^*$, $x\in L$ iff $M$ accepts the input $\track{x}{h(|x|)}$. Similarly,  $\cfl/n$ is defined in \cite{Yam08} using 1npda's in place of 1dfa's.

\section{Properties of Advice for Quantum Computation}\label{sec:QFA/n}

Since its introduction by Karp and Lipton \cite{KL82}, the usefulness of advice has been revealed for various models of underlying computations.
Following this line of study,
we are now focused on a simple and concise model of
{\em one-way measure-many quantum finite automata} (hereafter abbreviated as 1qfa's), each of which permits only one-way head moves and performs a (projective) measurement at every step to see if the machine enters any halting states.
We will discuss characteristic features of 1qfa's that are assisted by powerful pieces of deterministic advice and by examining how the 1qfa's process the advice with bounded-error probability.

\subsection{The Metric Vector Space $\YY_{\HH}$}\label{sec:metric-space}

To describe precisely the {\em time-evolution} of a 1qfa,
it is quite helpful to consider a new vector space $\YY_{\HH} = \HH \times \real\times\real$ induced from a target Hilbert space $\HH$.
Let  $\psi=(\ket{\phi},\gamma_1,\gamma_2)$ and $\psi'=(\ket{\phi'},\gamma'_1,\gamma'_2)$ be any two elements in $\YY_{\HH}$ and let $c$ be any {\em scalar} in the field $\real$. Now, we define the {\em scalar multiplication} $c\cdot \psi$ with respect to $\real$ as $(c\qubit{\phi},c\gamma_1,c\gamma_2)$. For convenience, we write $-\psi$ instead of $(-1)\cdot \psi$. Moreover, we define the {\em (vector) addition} $\psi+\psi'$ as $(\ket{\phi}+\ket{\phi'},\gamma_1+ \gamma'_1,\gamma_2+ \gamma'_2)$ and the {\em (vector) subtraction} $\psi-\psi'$ as $\psi+ (-\psi')$. Those operators naturally make $\YY_{\HH}$ a vector space and we then call all elements in $\YY_{\HH}$ {\em vectors}.

To further make $\YY_{\HH}$ a metric space, we first introduce an appropriate norm, which will induce a metric.
Our {\em norm}\footnote{Our definition of ``norm'' is quite different in its current  form from the norm defined in \cite{KW97,Gru00}.} of $\psi=(\ket{\phi},\gamma_1,\gamma_2)$ is denoted by $\|\psi\|$ and defined as
\begin{equation*}
\|\psi\| = \|(\ket{\phi},\gamma_1,\gamma_2)\| = \sqrt{ \|\ket{\phi}\|^2+ |\gamma_1|^2 + |\gamma_2|^2 }.
\end{equation*}
Using this norm, we define the {\em metric} (or the {\em distance function}) $d:\YY_{\HH}\times\YY_{\HH}\to\real$ as  $d(\psi,\psi') = \|\psi - \psi'\|$. As shown in Lemma \ref{metric-space-prop}, the pair $(\YY_{\HH},\|\cdot \|)$ forms a normed vector space, and thus $(\YY_{\HH},d)$ forms a metric space. For brevity, we drop ``$d$'' and simply call $\YY_{\HH}$ the metric space. To improve readability, we place  the proof of the lemma in Appendix.

\begin{lemma}\label{metric-space-prop}
Let $\psi,\psi',\psi''$ be any vectors in the metric vector space $\YY_{\HH}$.
\begin{enumerate}\vs{-1}
  \setlength{\topsep}{0mm}%
  \setlength{\itemsep}{0mm}
  \setlength{\parskip}{0cm}%

\item\label{triangle-inequality}
$\| \psi+\psi'\|  \leq \| \psi\|  + \| \psi'\|$.

\item\label{norm-triangle-prop}
$\|\psi-\psi''\| \leq \|\psi-\psi'\| + \|\psi'-\psi''\|$.
\end{enumerate}
\end{lemma}

\subsection{Basic Properties of 1QFA/{n}}\label{sec:basic-property-QFA/n}

Formally,  a 1qfa $M$ is a sextuple $(Q,\Sigma,\{U_{\sigma}\}_{\sigma\in\check{\Sigma}},q_0,Q_{acc},Q_{rej})$, where $\check{\Sigma} = \Sigma\cup\{\cent,\dollar\}$, whose {\em time-evolution operator} $U_{\sigma}$ is a unitary operator acting on the Hilbert space $E_{Q}= span\{\qubit{q}\mid q\in Q\}$ of dimension $|Q|$.
The series $\{U_{\sigma}\}_{\sigma\in\check{\Sigma}}$ describes the {\em time evolution} of $M$ on any input.  Let $P_{acc}$, $P_{rej}$, and $P_{non}$ be respectively {\em projections} of $E_{Q}$ onto three subspaces $E_{acc}= span\{\qubit{q}\mid q\in Q_{acc}\}$, $E_{rej} = span\{\qubit{q}\mid q\in Q_{rej}\}$, and $E_{non} = span\{\qubit{q}\mid q\in Q_{non}\}$.
Associated with a symbol $\sigma\in\check{\Sigma}$, we define a {\em transition operator} $T_{\sigma}$ as $T_{\sigma} = P_{non}U_{\sigma}$.
For each fixed string  $x=\sigma_1\sigma_2\cdots \sigma_n$ in $\check{\Sigma}^{*}$ of length $n$, we write $T_{x}$ for $T_{\sigma_n}T_{\sigma_{n-1}}\cdots T_{\sigma_2}T_{\sigma_1}$.

Let us consider a metric vector space $\YY_{E_Q}$ induced from Section \ref{sec:metric-space} by setting $\HH=E_{Q}$.
The aforementioned transition operator $T_{\sigma}$ is expanded into another operator $\hat{T}_{\sigma}:\YY_{E_Q}\to \YY_{E_Q}$ as follows.
First, we define the {\em sign function} $sgn:\real\to\{+1,-1\}$ as $sgn(\gamma)=+1$ if $\gamma\geq0$ and $sgn(\gamma)=-1$ otherwise. With this sign function, define
\begin{equation*}
\hat{T}_{\sigma}(\qubit{\phi},\gamma_1,\gamma_2)
= \left( T_{\sigma} \qubit{\phi}, sgn(\gamma_1) \sqrt{\gamma_1^2
+ \|  P_{acc}U_{\sigma}\qubit{\phi} \| ^2},
 sgn(\gamma_2) \sqrt{\gamma_2^2
+ \|  P_{rej}U_{\sigma}\qubit{\phi}\| ^2} \right).
\end{equation*}
Similarly to the definition of $T_x$, we further define $\hat{T}_{x}$ to be  the functional composition $\hat{T}_{\sigma_n}\hat{T}_{\sigma_{n-1}}\cdots \hat{T}_{\sigma_1}$. Notice that this extended operator $\hat{T}_{x}$ is no longer a linear operator; however, it satisfies useful properties listed in Lemma \ref{norm-property},  which will play a key role in the  proof of Theorem \ref{oneqfa-character}.
For convenience and clarity, we denote by $\YY_{[0,1]}$ the subspace of $\YY_{E_Q}$ consisting only of elements $(\ket{\phi},\gamma_1,\gamma_2)$ satisfying $0\leq \gamma_1,\gamma_2\leq 1$.


\begin{lemma}\label{norm-property}
Let $x\in\check{\Sigma}^*$ be any string and let $\psi=(\ket{\phi},\gamma_1,\gamma_2)$ and $\psi'=(\ket{\phi'},\gamma'_1,\gamma'_2)$ be two elements in $\YY_{[0,1]}$ satisfying that $\hat{T}_{x}\psi, \hat{T}_{x}\psi'\in \YY_{[0,1]}$ and $\|\ket{\phi}\|, \|\ket{\phi'}\|\leq 1$. Each of the following statements holds.
\begin{enumerate}\vs{-1}
  \setlength{\topsep}{0mm}%
  \setlength{\itemsep}{0mm}
  \setlength{\parskip}{0cm}%

\item\label{diff-estimate}
 $\| \qubit{\phi} - \qubit{\phi'}\| ^2 - \|  T_{x}(\qubit{\phi}-\qubit{\phi'})\| ^2 \leq 2 [ (\| \qubit{\phi}\| ^2 - \|  T_{x}\qubit{\phi}\| ^2) + (\| \qubit{\phi'}\| ^2 - \|  T_{x}\qubit{\phi'}\| ^2) ]$.

\item\label{upper-bound}
$\| \hat{T}_{x}\psi - \hat{T}_{x}\psi'\|  \leq \| \psi - \psi'\|$.

\item\label{lower-bound}
 $(\|\ket{\phi}\|^2-\|T_{x}\ket{\phi}\|^2)+ (\|\ket{\phi'}\|^2-\|T_{x}\ket{\phi'}\|^2)  + 4  \sqrt{ \left( \|\ket{\phi}\|^2-\|T_{x}\ket{\phi}\|^2 \right) + \left( \|\ket{\phi'}\|^2-\|T_{x}\ket{\phi'}\|^2 \right) }
 \geq  \| \psi - \psi'\| ^2 - \| \hat{T}_{x}\psi - \hat{T}_{x}\psi'\| ^2$.
\end{enumerate}
\end{lemma}

The proof of Lemma \ref{norm-property} is postponed until Appendix.


Each length-$n$ input string $x$ given to the 1qfa $M$ is expressed on the machine's input tape in the form $\cent x\dollar=\sigma_0\sigma_1\cdots\sigma_{n+1}$, including the two endmarkers $\cent$ and $\dollar$; in particular,  $\sigma_0=\cent$, $\sigma_{n+1}=\dollar$, and $x\in\Sigma^n$. The {\em acceptance probability} of $M$ on the input $x$ at step $i+1$ ($0\leq i\leq n+1$), denoted by  $p_{acc}(x,i+1)$, is  $\|  P_{acc}U_{\sigma_i}\qubit{\phi_{i}} \| ^2$, where  $\qubit{\phi_0} = \qubit{q_0}$  and $\qubit{\phi_{j+1}} = T_{\sigma_j}\qubit{\phi_{j}}$ for any index $j\in[0,n+1]_{\integer}$, and the {\em acceptance probability}  $p_{acc}(x)$ of $M$ on $x$  is $\sum_{i=1}^{n+2}p_{acc}(x,i)$.
Likewise, we define the {\em rejection probabilities} $p_{rej}(x,i+1)$ and $p_{rej}(x)$ using $P_{rej}$ instead of $P_{acc}$ in the above definition.
A computation of the 1qfa $M$ terminates after scanning the right endmarker. Notice that $\ket{\phi_j}$ could be $0$ prior to the $n+2$nd step.  In the end of the computation of $M$ on $x$, $M$ produces a vector  $\hat{T}_{\cent x\dollar}(\qubit{q_0},0,0) = (\qubit{\phi_{n+2}},\sqrt{p_{acc}(x)},\sqrt{p_{rej}(x)})$ in the metric space $\YY_{E_Q}$.
Conventionally, we say that $M$ {\em accepts} (resp., {\em rejects}) $x$ with probability $p_{acc}(x)$ (resp., $p_{rej}(x)$).

Regarding language recognition, we say that a language $L$ is {\em recognized} by $M$ (or $M$ {\em recognizes} $L$) with error probability at most  $\varepsilon$ if (i) for every string $x\in L$, $M$ accepts $x$ with probability at least $1-\varepsilon$ and (ii) for every string $x\in\Sigma^*-L$, $M$ rejects $x$ with probability at least $1-\varepsilon$. By viewing $M$ as a machine outputting two values, $0$ (rejection) and $1$ (acceptance), Conditions (i) and (ii) can be rephrased succinctly as follows: for every string $x\in\Sigma^*$, $M$ on the input $x$ {\em outputs} $L(x)$ with probability at least $1-\varepsilon$,
where
we set $L(x)=1$ for any $x\in L$ and $L(x)=0$ for any $x\in\Sigma^*-L$.
The notation $\oneqfa$ expresses the family of all languages recognized by 1qfa's with {\em bounded-error probability} (i.e.,  the error probability is upper-bounded by an absolute constant $\varepsilon$
in the real interval $[0,1/2)$).
For a later use, we also introduce another notation $\oneqfa_{(a(n),b(n))}$
for any two functions $a(n)$ and $b(n)$ mapping $\nat$ to $[0,1]$ as  the collection of all languages $L$ for which there exists a 1qfa $M$ satisfying: for every length $n\in\nat$ and every input $x\in\Sigma^n$, if $x\in L$ then $M$ accepts $x$ with probability {\em more than} $a(n)$, and if $x\not\in L$ then $M$ rejects $x$ with probability {\em more than} $b(n)$.


Naturally, we can supply deterministic advice to 1qfa's.
By analogy with $\reg/n$ and $\cfl/n$, the notation $\oneqfa/n$ refers to
the collection of all languages $L$ over alphabets $\Sigma$ that satisfy the following condition: there exist  an advice alphabet $\Gamma$, a 1qfa $M = (Q,\Sigma_{\Gamma}, \{U_{\sigma}\}_{\sigma\in\check{\Sigma}_{\Gamma}}, q_0,Q_{acc},Q_{rej})$, an error bound $\varepsilon\in[0,1/2)$, and an advice function
$h:\nat\rightarrow\Gamma^*$  such that (i) $|h(n)|=n$ for each length $n\in\nat$ (\ie $h$ is length-preserving) and (ii) for every $x\in\Sigma^*$, $M$ on input $\track{x}{h(|x|)}$ outputs $L(x)$ with probability at least $1-\varepsilon$ (abbreviated as $\prob_{M}[M(\track{x}{h(|x|)}) = L(x)]\geq 1-\varepsilon$, with  $M(\track{x}{h(|x|)})$ being treated as a random variable). The proof of the containment $\oneqfa\subseteq \reg$ given in \cite{KW97} can be carried over to assert that $\oneqfa/n\subseteq \reg/n$.

An immediate benefit of supplementing 1qfa's with appropriately chosen advice
is the elimination of the two endmarkers on their input tapes. Earlier, Brodsky and Pippenger \cite{BP02} demonstrated how to eliminate the left endmarker $\cent$ from 1qfa's input tapes. The use of advice further enables us to eliminate the right endmarker $\dollar$ as well. Intuitively, this elimination is done by marking the end of an input string by a piece of advice.

\begin{lemma}\label{endmarker}{\rm [endmarker lemma]}\hs{1}
Given any language $L$, $L$ is in $\oneqfa/n$ iff there exist a 1qfa $M$, a constant $\varepsilon\in[0,1/2)$, an advice alphabet $\Gamma$, and a length-preserving  advice function $h$ that satisfy the following conditions:
(i) $M$'s input tape contains no endmarker, (ii) $M$'s tape head starts at the leftmost input symbol, (iii) after $M$'s tape head reads the rightmost input symbol, $M$ stops operating, and (iv) for any nonempty string $x\in\Sigma^*$, $M$  on input $\track{x}{h(|x|)}$  outputs $L(x)$ with probability at least $1-\varepsilon$.
\end{lemma}

\begin{proof}
Let $L$ be any language over alphabet $\Sigma$.

(Only If--part)
Assume that $L$ is in $\oneqfa/n$.
Associated with this language $L$, we prepare a length-preserving advice function $h:\nat\rightarrow\Gamma^*$ for a certain advice alphabet
$\Gamma$ and an appropriate 1qfa  $M=(Q,\Sigma_{\Gamma},\{U_{\sigma}\}_{\sigma\in\check{\Sigma}_{\Gamma}}, q_0,Q_{acc},Q_{rej})$ using two endmarkers. Moreover, we assume that,  on any input of the form $\track{x}{h(|x|)}$
with $x\in \Sigma^*$, $M$ outputs $L(x)$ with success probability at least $1-\epsilon$, where $\epsilon$ is a certain constant in $[0,1/2)$.
For appropriate constants $k_0,k_1,k_2\in\nat^{+}$, assume also that $Q_{non} = \{q_{i}\mid 1\leq i\leq k_0\}$, $Q_{acc} = \{q_{k_0+i}\mid 1\leq i\leq k_1\}$, and $Q_{rej} = \{q_{k_0+k_1+i}\mid 1\leq i\leq k_2\}$.
We therefore obtain $Q=Q_{non}\cup Q_{acc}\cup Q_{rej}$ and we write $k$ for $|Q|$.
Following an argument of Brodsky and Pippenger \cite{BP02}, we can eliminate the left endmarker $\cent$, and hereafter we assume that $M$'s input tape has no $\cent$ for simplicity.

In the following manner, we will modify $M$ and $h$ to obtain the desired $M'$ and $h'$ that satisfy the lemma. Let us assume that $h$ has the form $h(n) = \tau_1\tau_2\cdots \tau_{n-1}\tau_n$, where $\tau_1,\tau_2,\ldots,\tau_n\in\Gamma$.   A new advice function $h'$ is defined to satisfy $h'(n) = \tau_1\cdots \tau_{n-1}\tau'_n$, where the last symbol $\tau'_n$ is $\track{\tau_n}{\dollar}$, indicating the end of any input string  of length $n$.
To describe a new 1qfa $M'$ equipped with no endmarker, we need to embed each operator $U_{\sigma}$ into a slightly larger space, say, $E_{Q'}$. For this purpose, we first define  $Q'_{acc} = \{q_{k+i} \mid 1\leq i\leq k_1\}$ and $Q'_{rej} = \{q_{k+k_1+i}\mid 1\leq i\leq k_2\}$, and we then set $Q' = Q\cup Q'_{acc}\cup Q'_{rej}$.
To describe new operators $U'_{\sigma}$, we use a special unitary matrix $S$, which is called ``sweeping matrix'' in
\cite{BP02}, defined as
\[
{\small S = \ninematrices{I_{non}}{O}{O}{O}{O}{I_{halt}}{O}{I_{halt}}{O}},
\]
where $I_{non}$ (resp., $I_{halt}$) is the {\em identity matrix} of size $k_0$
(resp., $k_1+k_2$). This matrix $S$ swaps ``old'' halting states of $M$ with ``new'' non-halting states so that, after an application of unitary matrix $U_{\track{\sigma}{\tau}}$, we can deter the effect of an application of  the measurement operator $P_{non}$ that normally comes immediately after $U_{\track{\sigma}{\tau}}$. Using this operator $S$, we further define
\[
{\small U'_{\track{\sigma}{\tau}} = S \matrices{U_{\track{\sigma}{\tau}}}{O}{O}{I_{halt}}}
\;\;\;\text{and}\;\;\;
{\small U'_{\track{\sigma}{\tau'}} = S \matrices{U_{\dollar}}{O}{O}{I_{halt}} \matrices{U_{\track{\sigma}{\tau}}}{O}{O}{I_{halt}}},
\]
where $\tau' = \track{\tau}{\dollar}$. The measurement operator $P_{acc}$ is also expanded naturally to the space $E_{Q'}$, and it is succinctly denoted by $P'_{acc}$. It is not difficult to show that the operator $P'_{acc}U'_{\track{\sigma}{\tau'}}$ produces a similar effect as the operator $P_{acc}U_{\dollar}P_{non}U_{\track{\sigma}{\tau}}$ does. Therefore, using
the advice  function $h'$, $M'$ accepts any given input $x$ with the same probability as $M$ does on $x$ with the advice function $h$.

(If--part) Take $\Gamma,\varepsilon,h,M$ given in the lemma. Since $M$'s input tape uses no endmarker, let $M$ be of the form $(Q,\Sigma_{\Gamma}, \{U_{\sigma}\}_{\sigma\in \Sigma_{\Gamma}}, q_0, Q_{acc},Q_{rej})$ and assume that, for all {\em nonempty} inputs $x$, $M$ on input $\track{x}{h(|x|)}$ correctly outputs $L(x)$ with probability at least $1-\varepsilon$. From this $M$, we want to construct another 1qfa $N$ equipped with two endmarkers so that $N$ recognizes $L$ with error probability at most $\varepsilon$. Here, we consider only the case where $\lambda\in L$, because the other case can be similarly handled.
For convenience, let $\ket{\psi_0} = \ket{q_0}$ and let $\ket{\phi_{i+1}}$ express a quantum state in $E_{Q}$ generated by $M$ after step $i$.

Choose a fresh inner state $q_f$ and set $Q'_{acc}= Q_{acc}\cup \{q_f\}$ and $Q'_{rej}= Q_{rej}$. Moreover, define $Q'= Q\cup\{q_f\}$.
We then expand the scope of each unitary operator $U_{\sigma}$ to $E_{Q'}$ as follows.
Taking any symbol $\sigma\in \Sigma_{\Gamma}$, we set $U'_{\sigma}\ket{q_f}$ to be $\ket{q_f}$ and, for each $q\in Q$, we further define $U'_{\sigma}\ket{q}$ to be $U_{\sigma}\ket{q}$. Next, we want to define two new operators $U'_{\cent}$ and $U'_{\dollar}$ associated with the two endmarkers. Let $U'_{\cent}\ket{q} = \ket{q}$ for any $q\in Q'$. Furthermore, let $U'_{\dollar}\ket{q_0} = \ket{q_f}$, $U'_{\dollar}\ket{q_f} = \ket{q_0}$, and $U'_{\dollar}\ket{q} = \ket{q}$ for any other $q$ in $Q'$.
Let $x$ be any nonempty input.
When $N$ applies $P_{acc}U'_{\dollar}$ to a quantum state  $P_{non}\ket{\phi_{n+1}}$ of $M$, any occurrence of $\ket{q_0}$ in $P_{non}\ket{\phi_{n+1}}$ (if $q_0\notin Q_{halt}$) changes to $\ket{q_f}$ and this change contributes to an increase of the acceptance probability of $M$.  Recall that,  while reading $x$, $M$ achieves the acceptance or rejection probability of at least $1-\varepsilon$.
This fact implies that
$N$ cannot sway the decision of acceptance or rejection made by $M$ and that the error probability of $N$ is not higher than $M$'s.
In contrast, when $x=\lambda$, since  $U'_{\dollar}U'_{\cent}\ket{q_0} =\ket{q_f}$, $N$ accepts $x$ with certainty. Therefore, $N$ recognizes $L$ with error probability at most $\varepsilon$.
\end{proof}


For  our analyses of languages in $\oneqfa/n$, not all well-known properties proven for $\oneqfa$ turn out to be as useful as we have hoped them to be.
One of such properties is a criterion, known as a {\em partial order condition}\footnote{A language satisfies the partial order condition exactly when its minimal 1dfa contains no two inner states $q_1,q_2\in Q$ such that (i) there is a string $z$ for which $\hat{\delta}(q_1,z)\in Q_{acc}$ and $\hat{\delta}(q_2,z)\not\in Q_{acc}$ or vice versa, and (ii) there are two nonempty strings $x$ and $y$ for which $\hat{\delta}(q_1,x)=\hat{\delta}(q_2,x)=q_2$ and $\hat{\delta}(q_2,y)=q_1$.}
of Brodsky and Pippenger \cite{BP02}.  Earlier, Kondacs and Watrous \cite{KW97} proved  that $\reg\nsubseteq\oneqfa$ by considering a padded language $L_a = \{wa\mid w\in\Sigma^*\}$ over a binary alphabet $\Sigma=\{a,b\}$.
Brodsky and Pippenger \cite{BP02} then pointed out that this result follows from a more general fact in which every language in $\oneqfa$ satisfies the  partial order condition but $L_a$ does not.
Unlike $\oneqfa$, the advised family $\oneqfa/n$ violates this criterion because the above language $L_a$ falls into $\oneqfa/n$.
This fact is a typical example that makes an analysis of $\oneqfa/n$ look quite different from an analysis of $\oneqfa$.

\begin{lemma}\label{partial-order-cond}
The advised language family $\mathrm{1QFA}/n$ does not satisfy the criterion of the partial order condition.
\end{lemma}

\begin{proof}
Let $\Sigma=\{a,b\}$ and consider the aforementioned language $L_{a}= \{wa\mid w\in\Sigma^*\}$. We aim at proving that this language belongs to $\oneqfa/n$ by constructing an appropriate 1qfa $M$ and a certain length-preserving advice function $h$.
Since $L_a$ does not satisfy the partial order condition, the lemma immediately follows.

It suffices, by Lemma \ref{endmarker}, to build an advised 1qfa without any endmarker.
Our advice alphabet $\Gamma$ is $\{0,1\}$, and the desired 1qfa $M$ is defined as $(Q,\Sigma_{\Gamma},\{U_{\sigma}\}_{\sigma\in\check{\Sigma}_{\Gamma}}, q_0,Q_{acc},Q_{rej})$, where  $Q=\{ q_0,q_1,q_2\}$, $Q_{acc} = \{q_1\}$, and $Q_{rej}=\{q_2\}$.  Time-evolution operators of $M$ are
$U_{\track{e}{0}} = I$ (identity) for each symbol $e\in\Sigma$ and
\[
\hs{-5}
{\small U_{\track{a}{1}} = \ninematrices{0}{1}{0}{1}{0}{0}{0}{0}{1}}
\;\;\text{and}\;\;
{\small U_{\track{b}{1}} =  \ninematrices{0}{0}{1}{0}{1}{0}{1}{0}{0}}.
\]
Finally, for any $n\geq1$,  we set an advice function $h$ to be $h(n)=0^{n-1}1$, which gives a cue to our 1qfa $M$ to check whether the last input symbol equals $a$.
An initial configuration of $M$ is $\qubit{\psi_0} = (1,0,0)^{T}$, indicating  that $\qubit{q_0}$ has amplitude $1$.

A direct calculation shows that  $U_{\track{w}{0^{n-1}}\track{a}{1}}\qubit{q_0} = \qubit{q_1}$ and $U_{\track{w}{0^{n-1}}\track{b}{1}}\qubit{q_0} = \qubit{q_2}$.  Since $q_1\in Q_{acc}$ and $q_2\in Q_{rej}$,  $M$ should recognize $L_{a}$ with certainty, leading to the desired conclusion that  $L_a$ belongs
to $\oneqfa/n$.
\end{proof}

\subsection{A Necessary Condition for 1QFA/{n}}\label{sec:QFA-theorem}

A quick way to understand a source of the power of advised 1qfa's  may be to find a machine-independent, algebraic characterization of languages in $\oneqfa/n$. Such a characterization for other machine models has already turned out to be a useful tool in studying the computational complexity of languages (\eg \cite{Yam10}).  What we plan to prove here is a slightly weaker result: a machine-independent, algebraic {\em necessary} condition for those languages that properly fall into $\oneqfa/n$.

Let us give a precise description of our first main theorem, Theorem \ref{oneqfa-character}.
Following a standard convention, for any given partial order $\leq$ defined on a given finite set, we always use the notation $x=y$ exactly when both $x\leq y$ and $y\leq x$ hold; moreover, we write $x<y$ in the case where both $x\leq y$ and $x\neq y$ hold.
With respect to $<$, a sequence $(s_1,s_2,\ldots,s_m)$ of length $m$ ($m\geq1$) is called a {\em strictly descending chain} if $s_{i+1}< s_{i}$ holds for any index $i\in[m-1]$.
For our convenience, we call a reflexive, symmetric, binary relation a {\em closeness relation}. Given any closeness relation $\cong_{S}$, an {\em $\cong_{S}$-discrepancy set} is a set $S$ satisfying that, for any two elements $x,y\in S$, if $x$ and $y$ are ``different'' elements,
then $x\not\cong_{S} y$.

\begin{theorem}\label{oneqfa-character}
Let $S$ be any language over alphabet $\Sigma$ and let $\Delta = \{(x,n)\in\Sigma^*\times\nat \mid |x|\leq n\}$. If $S$ belongs to $\oneqfa/n$, then there exist two constants $c,d\in\nat^{+}$, an equivalence relation $\equiv_{S}$ over $\Delta$, a partial order $\leq_{S} $ over $\Delta$, and a closeness  relation $\cong_{S}$ over $\Delta$ that satisfy the seven conditions listed below. In the list, we assume that $(x,n),(y,n)\in\Delta$, $z\in\Sigma^*$, and $\sigma\in\Sigma$ with $|x|=|y|$.
\begin{enumerate}\vs{-1}
  \setlength{\topsep}{0mm}%
  \setlength{\itemsep}{0mm}
  \setlength{\parskip}{0cm}%

\item\label{item:equiv-class} The cardinality of the set $\Delta/\!\equiv_{S}$ of equivalence classes is at most $d$.

\item\label{item:cond-to-equiv} If $(x,n)\cong_{S} (y,n)$, then $(x,n)\equiv_{S} (y,n)$.

\item\label{item:less-than} If $|x\sigma|\leq n$, then $(x\sigma,n)\leq_{S}  (x,n)$ and, if $|x|=n>0$, then $(x,n)<_{S}(\lambda,n)$.

\item\label{item:reverse} When $(x,n)=_{S}(xz,n)$ and  $(y,n)=_{S}(yz,n)$ with $|xz|\leq n$, $(xz,n)\cong_{S} (yz,n)$ implies   $(x,n)\equiv_{S}(y,n)$.

\item\label{item:equiv-S} $(x,n)\equiv_{S}(y,n)$ iff $S(xz)=S(yz)$ for all strings $z\in\Sigma^*$ with $|xz|=n$.

\item\label{item:chain} Any strictly descending chain (with respect to  $<_{S}$) in $\Delta$ has length at most $c$.

\item\label{item:discrepancy} Any $\cong_{S}$-discrepancy subset of $\Delta$ has cardinality at most $d$.
\end{enumerate}
\end{theorem}

The meanings of the above three relations $\simeq$, $\leq_{S}$, and $\equiv_{S}$ will be clarified in the following proof of Theorem \ref{oneqfa-character}.
Since our proof of the theorem
heavily relies on Lemma  \ref{norm-property}, the proof requires  only
basic properties of the norm in the metric vector space $\YY_{E_Q}$
discussed in Section \ref{sec:basic-property-QFA/n}.

\begin{proofof}{Theorem \ref{oneqfa-character}}
Let $\Sigma$ be any alphabet, let $\Delta = \{(x,n)\mid x\in\Sigma^*,n\in\nat, |x|\leq n\}$, and let $S$ be any language  in $\oneqfa/n$ over $\Sigma$. For this language $S$, by Lemma \ref{endmarker}, take an advice alphabet $\Gamma$, an error bound $\varepsilon\in[0,1/2)$, a 1qfa $M = (Q,\Sigma_{\Gamma}, \{U_{\sigma\in\Sigma_{\Gamma}}, q_0,Q_{acc},Q_{rej})$, and a length-preserving advice function $h:\nat\rightarrow\Gamma^*$ satisfying that (i) $M$'s inpout tape uses no endmarker and (ii) $\prob_{M}[M(\track{x}{h(|x|)})=S(x)]\geq 1-\varepsilon$ for every nonempty string $x\in\Sigma^*$.
Without loss of generality, we hereafter assume that $\varepsilon>0$.

Recalling the notation $\Sigma_{\Gamma}$ for $\{\track{\sigma}{\tau}\mid \sigma\in \Sigma, \tau\in\Gamma\}$, we set $e=|\Sigma_{\Gamma}|$. For simplicity, write $\psi_{0}$ for the triplet $(\qubit{q_0},0,0)$ in the metric vector space
$\YY_{E_Q}$ ($ = span\{E_{Q}\}\times\real\times\real$) of dimension $|E_Q|+2$. 
For technicality, we set $\hat{T}_{\track{\lambda}{\lambda}}=I$ so that, when $x=w=\lambda$, $\hat{T}_{\track{x}{w}}\psi_0$ coincides with $\psi_0$.
Given any element $(x,n)\in\Delta$ and its associated string $w=Pref_{|x|}(h(n))$,  we assume that $\hat{T}_{\track{x}{w}}\psi_0$ has the form $(\qubit{\phi_x},\gamma_{x,1},\gamma_{x,2})$.

As the first stage, we intend to define a closeness relation $\cong_{S}$ on $\Delta$. For our  purpose, we set $\varepsilon^*= \sqrt{\varepsilon(1-\varepsilon)}$ and choose a constant $\mu$ satisfying $0< \mu < 2(1-2\varepsilon^*)/9$. Notice that $2(1-2\varepsilon^*)/9<\varepsilon$. Since $0\leq \varepsilon^*<1/2$, $\mu< 2/9$ follows. Given two elements $(x,n),(y,m)\in\Delta$, we write $(x,n)\cong_{S} (y,m)$ exactly when $\| \hat{T}_{\track{x}{w}}\psi_0 - \hat{T}_{\track{y}{v}}\psi_0 \| ^2 < \mu$ holds, where $w=Pref_{|x|}(h(n))$ and  $v=Pref_{|y|}(h(m))$.
To see that Condition \ref{item:discrepancy} is satisfied, let us consider an arbitrary $\cong_{S}$-discrepancy subset $G$ of $\Delta$.
For any two distinct elements $(x,n),(y,m)\in G$, it must hold that
$\| \hat{T}_{\track{x}{w}}\psi_0 -  \hat{T}_{\track{y}{v}}\psi_0 \| ^2 \geq \mu$ . Since $\mu$ is a positive constant, $G$ must be a finite set. More precisely, let $d= 2(|E_{Q}|+2)^2/\mu$. This value $d$ upper-bounds the cardinality $|G|$ of $G$.

\begin{claim}\label{cardinality-G}
The cardinality $|G|$ is upper-bounded by $d$, independent of the choice of $G$.
\end{claim}

By Claim \ref{cardinality-G},    Condition \ref{item:discrepancy} is immediately met.

\begin{proofof}{Claim \ref{cardinality-G}}
For brevity, we set $k=|E_{Q}|$.
Consider the set $V(G)=\{\hat{T}_{\track{x}{w}}\psi_0 \mid (x,n)\in G, w=Pref_{|x|}(h(n))\}$. Since $|V(G)|=|G|$, it suffices to show the inequality  $|V(G)|\leq 2(k+2)^2/\mu$, which directly implies the claim.

Let $(x,n),(y,m)\in G$, $w=Pref_{|x|}(h(n))$, and $v=Pref_{|y|}(h(m))$.
Assume that $\hat{T}_{\track{x}{w}}\psi_0$ and $\hat{T}_{\track{y}{v}}\psi_0$ respectively have the form  $(\ket{\phi_{x}}, \gamma_{x,1},\gamma_{x,2})$ and
 $(\ket{\phi_{y}},\gamma_{y,1},\gamma_{y,2})$ with      $\|\ket{\phi_x}\|,\|\ket{\phi_y}\|,\gamma_{x,j},\gamma_{y,j}\in[0,1]$.
Moreover, let $\ket{\phi_{x}} = (\xi_{x,1},\xi_{x,2},\ldots,\xi_{x,k})^{T}$ and $\ket{\phi_y} = (\xi_{y,1},\xi_{y,2},\ldots,\xi_{y,k})^{T}$.  By the definition of our norm, $\| \hat{T}_{\track{x}{w}}\psi_0 - \hat{T}_{\track{y}{v}}\psi_0 \| ^2$ equals $\sum_{i=1}^{k}|\xi_{x,i}-\xi_{y,i}|^2 + \sum_{j=1}^{2}|\gamma_{x,j}-\gamma_{y,j}|^2$.
Note that this value is at least $\mu$ if $\hat{T}_{\track{x}{w}}\psi_0$ and $\hat{T}_{\track{y}{v}}\psi_0$ are distinct vectors in $V(G)$. It therefore follows that either (i) there exists an index $i\in[k]$ satisfying (*) $\mu/(k+2)\leq |\xi_{x,i} -\xi_{y,i}|^2\leq 2$ or (ii) there exists an index $j\in[2]$ satisfying (**) $\mu/(k+2)\leq |\gamma_{x,j} -\gamma_{y,j}|^2\leq 1$.
From this fact, we can conclude that, for each fixed index $i\in[k]$ (resp., $j\in[2]$), there are only at most $2(k+2)/\mu$ (resp., $(k+2)/\mu$) distinct elements $\xi_{x,i}$ (resp., $\gamma_{x,j}$) satisfying Condition (*) (resp., Condition (**)). Since the cardinality $|V(G)|$ is upper-bounded by the total number of those elements, it follows that
\[
|V(G)|\leq k\cdot \frac{2(k+2)}{\mu} + 2\cdot \frac{k+2}{\mu} \leq (k+2)\cdot \frac{2(k+2)}{\mu} = \frac{2(k+2)^2}{\mu}=d.
\]
Obviously, the value $d$ is irrelevant to the choice of $G$.
\end{proofof}

As the second stage, we aim at defining a relation $\equiv_{S}$ to satisfy Condition \ref{item:equiv-S}. For the time being, however, we define $\equiv_{S}$ as a subset of $\bigcup_{n\in\nat}(\Delta_n\times\Delta_n)$, where $\Delta_n$ denotes the set $\{(x,n)\in \Delta \mid |x|\leq n\}$; later, we will expand it to $\Delta\times\Delta$, as required by the lemma. For any two elements $(x,n),(y,n)\in\Delta_n$, we write $(x,n)\equiv_{S}(y,n)$ whenever $S(xz)=S(yz)$ holds for all strings $z$ satisfying  $|xz|=n$. {}From this definition, it is not difficult to show that $\equiv_{S}$ satisfies the properties of reflexivity, symmetry, and transitivity; thus, $\equiv_{S}$ is indeed an equivalence relation. This shows Condition \ref{item:equiv-S}.

To show Condition \ref{item:cond-to-equiv} for $\cong_{S}$ and $\equiv_{S}$, we start with the following statement.

\begin{claim}\label{T-diff-bound}
For any two elements $(x,n),(y,n)\in\Delta$ with $|x|=|y|$, if  $\| \hat{T}_{\track{x}{w}}\psi_0 - \hat{T}_{\track{y}{w}}\psi_0 \| ^2 < 2(1-2\varepsilon^*)$, then $(x,n)\equiv_{S}(y,n)$ holds.
\end{claim}

Condition \ref{item:cond-to-equiv} follows directly from Claim \ref{T-diff-bound} as follows. Assume that $(x,n)\cong_{S} (y,n)$. From this assumption, it follows that $\| \hat{T}_{\track{x}{w}}\psi_0 - \hat{T}_{\track{y}{w}}\psi_0 \| ^2 < \mu < 2(1-2\varepsilon^*)/9 <2(1-2\varepsilon^*)$. Applying Claim \ref{T-diff-bound}, we then obtain $(x,n)\equiv_{S}(y,n)$, as requested.

In order to prove Claim \ref{T-diff-bound}, we need to prove two key claims, Claims \ref{T-acc-rej} and \ref{distance-S}.

\begin{claim}\label{T-acc-rej}
For any two elements $(x,n),(y,n)\in\Delta$ and any string $z\in\Sigma^*$ with $|x|=|y|$ and $|xz|=n$, it holds that
$
\| \hat{T}_{\track{x}{w}}\psi_{0} - \hat{T}_{\track{y}{w}}\psi_{0} \| ^2
\geq (\sqrt{p_{acc}(xz)} - \sqrt{p_{acc}(yz)})^2 + (\sqrt{p_{rej}(xz)} - \sqrt{p_{rej}(yz)})^2.
$
\end{claim}

\begin{proof}
By a direct calculation of the norm, we obtain
\begin{eqnarray*}
\lefteqn{\left\| \hat{T}_{\track{xz}{h(n)}}\psi_{0}
- \hat{T}_{\track{yz}{h(n)}}\psi_{0} \right\| ^2} \hs{10} \\
&=& \left\|  \left(\qubit{\phi_{xz}} - \qubit{\phi_{yz}}, \sqrt{p_{acc}(xz)} - \sqrt{p_{acc}(yz)}, \sqrt{p_{rej}(xz)} - \sqrt{p_{rej}(yz)} \right) \right\| ^2 \\
&=& \|  \qubit{\phi_{xz}} - \qubit{\phi_{yz}}\| ^2 + \left( \sqrt{p_{acc}(xz)} - \sqrt{p_{acc}(yz)} \right)^2 + \left( \sqrt{p_{rej}(xz)} - \sqrt{p_{rej}(yz)} \right)^2 \\
&\geq&  \left( \sqrt{p_{acc}(xz)} - \sqrt{p_{acc}(yz)} \right)^2 + \left( \sqrt{p_{rej}(xz)} - \sqrt{p_{rej}(yz)} \right)^2.
\end{eqnarray*}
On the contrary, since $\hat{T}_{\track{xz}{wu}}\psi_{0} = \hat{T}_{\track{z}{u}}(\hat{T}_{\track{x}{w}}\psi_{0})$ and $\hat{T}_{\track{yz}{wu}}\psi_{0} = \hat{T}_{\track{z}{u}}(\hat{T}_{\track{y}{w}}\psi_{0})$, Lemma \ref{norm-property}(\ref{upper-bound}) leads to  the following inequality:
\begin{eqnarray*}
\left\| \hat{T}_{\track{xz}{h(n)}}\psi_{0} - \hat{T}_{\track{yz}{h(n)}}\psi_{0} \right\| ^2
\leq  \left\|  \hat{T}_{\track{x}{w}}\psi_{0} -
\hat{T}_{\track{y}{w}}\psi_{0} \right\| ^2.
\end{eqnarray*}
By combining the above two inequalities, the claim immediately follows.
\end{proof}

\begin{claim}\label{distance-S}
Let $(x,n),(y,n)\in\Delta$ with $|x|=|y|$. If $\| \hat{T}_{\track{x}{w}}\psi_0 - \hat{T}_{\track{y}{w}}\psi_0 \| ^2< 2(1-2\varepsilon^*)$, then $S(xz) = S(yz)$ holds for all strings $z\in\Sigma^*$ satisfying $|xz|=n$.
\end{claim}

\begin{proof}
Assume that $\| \hat{T}_{\track{x}{w}}\psi_0 - \hat{T}_{\track{y}{w}}\psi_0 \| ^2< 2(1-2\varepsilon^*)$. To lead to a contradiction, we further assume that a certain string $z$ satisfies both $|xz|=n$ and $S(xz) \neq S(yz)$. The latter assumption (concerning $z$) implies that either  (i) $p_{acc}(xz) \geq 1-\varepsilon$ and $P_{acc}(yz)\leq \varepsilon$, or (ii) $p_{rej}(xz) \geq 1-\varepsilon$ and $P_{rej}(yz)\leq \varepsilon$.
In either case, since $\varepsilon^*=\sqrt{\varepsilon(1-\varepsilon)}$, we conclude that
$( \sqrt{p_{acc}(xz)} - \sqrt{p_{acc}(yz)})^2 \geq (\sqrt{1-\varepsilon}-\sqrt{\varepsilon})^2 = 1-2\varepsilon^*$ and, similarly,  $(\sqrt{p_{rej}(xz)} - \sqrt{p_{rej}(yz)})^2 \geq 1-2\varepsilon^*$.
By appealing to Claim \ref{T-acc-rej}, we obtain
\begin{eqnarray*}
\left\|  \hat{T}_{\track{x}{w}}\psi_{0} -
\hat{T}_{\track{y}{w}}\psi_{0} \right\| ^2
&\geq&  \left(\sqrt{p_{acc}(xz)} - \sqrt{p_{acc}(yz)}\right)^2 + \left(\sqrt{p_{rej}(xz)} - \sqrt{p_{rej}(yz)}\right)^2  \geq 2(1-2\varepsilon^*).
\end{eqnarray*}
This contradicts our first assumption that
$\|  \hat{T}_{\track{x}{w}}\psi_{0} -
\hat{T}_{\track{y}{w}}\psi_{0} \| ^2 < 2(1-2\varepsilon^*)$.
Therefore, the equation $S(xz)=S(yz)$ should hold for any string $z$ of length $n-|x|$.
\end{proof}

Finally, Claim \ref{T-diff-bound} can be proven in the following way.
Assuming  $\| \hat{T}_{\track{x}{w}}\psi_0 - \hat{T}_{\track{y}{w}}\psi_0 \| ^2< 2(1-2\varepsilon^*)$, Claim \ref{distance-S} yields the equality $S(xz)=S(yz)$ for any string $z$ of length $n-|x|$. This obviously implies the equivalence $(x,n)\equiv_{S}(y,n)$ because of the definition of $\equiv_{S}$. Therefore, Claim \ref{T-diff-bound} should be true.

As announced earlier, we want to expand a scope of $\equiv_{S}$ from $\bigcup_{n\in\nat}(\Delta_n\times\Delta_n)$ to $\Delta\times\Delta$.  Before giving a precise definition of $\equiv_{S}$,
we briefly discuss an upper-bound of the cardinality $|\Delta_n/\!\equiv_{S}\!|$. Recall that $d= 2(|E_{Q}|+2)^2/\mu$.

\begin{claim}\label{equivalance-class}
For every length $n\in\nat$, $|\Delta_n/\!\equiv_{S}\!|\leq d$ holds.
\end{claim}

\begin{proof}
Let us assume otherwise; namely, $|\Delta_n/\!\equiv_{S}\!| > d$ for a certain number $n\in\nat$. Clearly, $n$ must be at least $d+1$. Now, take $d+1$ different  strings $x_1,x_2,\ldots,x_{d+1}\in\Sigma^n$ so that   $(x_i,n)\not\equiv_{S}(x_j,n)$ holds for every distinct pair $i,j\in[d+1]$.
{}From Condition \ref{item:cond-to-equiv} follows the inequality $(x_i,n)\not\cong_{S}(x_j,n)$.
Next, let us consider the set $G=\{(x_i,n)\mid i\in[d+1]\}$. Obviously, $|G|=d+1$ holds. Since $G$ is a $\cong_{S}$-discrepancy subset of $\Delta_n$, Condition \ref{item:discrepancy} implies $|G|\leq d$. This is clearly a contradiction. Therefore, the claim should be  true.
\end{proof}

Since $|\Delta_n/\!\equiv_{S}\!|\leq d$ holds for each length $n\in\nat$ by Claim \ref{equivalance-class}, each set $\Delta_n/\!\equiv_{S}$ can be expressed as $\{A_{n,1},A_{n,2},\ldots,A_{n,d}\}$, provided that, in the case of  $|\Delta_n/\!\equiv_{S}\!|<d$ for a certain $n$, we automatically set $A_{n,i}=\setempty$ for any index $i$ satisfying  $|\Delta_n/\!\equiv_{S}\!|<i\leq d$.
Now, we expand $\equiv_{S}$ in the following natural way. For two arbitrary elements $(x,n)$ and $(y,m)$ in $\Delta$ with $n\neq m$, let $(x,n)\equiv_{S}(y,m)$ if there exists an index $i\in[d]$ such that $(x,n)\in A_{n,i}$ and $(y,m)\in A_{m,i}$. Note that this extended version of $\equiv_{S}$ is also an equivalence relation. {}From the above definition of $\equiv_{S}$, the set $\Delta/\!\equiv_{S}$ is obviously finite, and hence Condition \ref{item:equiv-class}  is satisfied.

As the third stage, we will define the desired partial order $\leq_{S}$ on $\Delta$. Here, we write $(x,n)\leq_{S}(y,m)$ if there exist two numbers $s,s'\in\nat$ for which (i) $0\leq s\leq s'\leq \ceilings{1/\mu^2}$, (ii) $(s-1)\mu^2<\| \qubit{\phi_x}\| ^2\leq s\mu^2$, and (iii) $(s'-1)\mu^2<\| \qubit{\phi_y}\| ^2\leq s'\mu^2$.
As remarked earlier, we write $(x,n)=_{S}(y,m)$ exactly when $(x,n)\leq_{S}(y,m)$ and $(y,m)\leq_{S}(x,n)$. In particular, when
$(x,n)=_{S}(y,m)$ holds, we obtain
$\| |\qubit{\phi_x}\| ^2-\| \qubit{\phi_y}\| ^2|<\mu^2$.
It is easy to check that $\leq_{S}$ is reflexive, antisymmetric, and transitive; thus, $\leq_{S}$ is truly a partial order.
Since $\| \qubit{\phi_{x\sigma}}\| \leq \| \qubit{\phi_x}\|$ always holds for any pair
$(x,\sigma)\in \Sigma^*\times\Sigma$, we conclude that $(x\sigma,n)\leq_{S}(x,n)$ if $|x\sigma|\leq n$.

When $|x|=n>0$, $\|\ket{\phi_x}\|^2\leq\varepsilon$ always holds because either $p_{acc}(x)\geq 1-\varepsilon$ or $P_{rej}(x)\geq 1-\varepsilon$ must hold.
From $0<\mu<2/9$ and $0<\varepsilon< 1/2$,
it follows that $1-\varepsilon> \frac{1}{2} >\frac{27}{2}\mu^2$.
Since $\|\ket{\phi_{\lambda}}\|^2= \|\ket{q_0}\|^2 = 1$, we obtain $\|\ket{\phi_{\lambda}}\|^2 - \|\ket{\phi_x}\|^2 \geq 1-\varepsilon > \frac{27}{2}\mu^2$, in other words, $(x,n)<_{S}(\lambda,n)$.
Therefore, Condition \ref{item:less-than} is met.

Regarding Condition \ref{item:chain}, we set the desired constant $c$ to be $\ceilings{1/\mu^2}+1$.
Consider any strictly descending chain in $\Delta$ with respect to $<_{S}$:
$(x_e,n_e)<_{S}(x_{e-1},n_{e-1})<_{S}\cdots <_{S}(x_1,n_1)$, where $e$ is the length of the chain. It should hold that $\|\qubit{\phi_{x_i}}\|^2 - \|\qubit{\phi_{x_{i+1}}}\|^2\geq\mu^2$ for any index $i\in[0,e-1]_{\integer}$. This implies
\[
\|\qubit{\phi_{x_1}}\|^2 - \|\qubit{\phi_{x_e}}\|^2 = \sum_{i=1}^{e-1} \left(\|\qubit{\phi_{x_i}}\|^2 - \|\qubit{\phi_{x_{i+1}}}\|^2\right) \geq (e-1)\mu^2.
\]
Since $\|\qubit{\phi_{x_1}}\|^2-\|\qubit{\phi_{x_e}}\|^2\leq 1$,  $(e-1)\mu^2 \leq 1$ immediately follows; therefore, we conclude that $e\leq 1+1/\mu^2 \leq c$.  Condition \ref{item:chain} thus follows.

The remaining conditions to verify is only Condition \ref{item:reverse}. To show this condition, firstly we will prove  Claim \ref{T-norm-diff}, which follows from
Lemma \ref{norm-property}(\ref{lower-bound}).

\begin{claim}\label{T-norm-diff}
Let $\alpha,\gamma\in(0,1]$. Let $(x,n),(y,n)\in\Delta$ and assume that $|x|=|y|$ and $|xz|\leq n$. If  $\| \hat{T}_{\track{xz}{wu}}\psi_0 - \hat{T}_{\track{yz}{wu}}\psi_0 \| ^2<\gamma$, $\| \qubit{\phi_{x}}\| ^2 - \| \qubit{\phi_{xz}}\| ^2<\alpha$, and $\| \qubit{\phi_{y}}\| ^2 - \| \qubit{\phi_{yz}}\| ^2<\alpha$, then
$\| \hat{T}_{\track{x}{w}}\psi_0 - \hat{T}_{\track{y}{w}}\psi_0 \| ^2<\gamma + 8\sqrt{\alpha}$, where $w$ and $u$ satisfy $wu=Pref_{|xz|}(h(n))$ with $|w|=|x|$ and $|u|=|z|$.
\end{claim}

\begin{proof}
Note that  $\qubit{\phi_{xz}} = T_{\track{z}{u}}\qubit{\phi_x}$ and
$\qubit{\phi_{yz}} = T_{\track{z}{u}}\qubit{\phi_y}$.
For convenience, we set $\psi = \hat{T}_{\track{x}{w}}\psi_0$ and $\psi' = \hat{T}_{\track{y}{w}}\psi_0$. Those vectors $\psi$ and $\psi'$ satisfy that $\hat{T}_{\track{z}{u}}\psi = \hat{T}_{\track{xz}{wu}}\psi_0$ and
$\hat{T}_{\track{z}{u}}\psi' = \hat{T}_{\track{yz}{wu}}\psi_0$.
Since $\psi,\psi'\in\YY_{[0,1]}$ and thus $ \hat{T}_{\track{z}{u}}\psi, \hat{T}_{\track{z}{u}}\psi'\in\YY_{[0,1]}$,  it is possible to apply Lemma \ref{norm-property}(\ref{lower-bound}) and then obtain
\begin{eqnarray*}
\|  \psi - \psi' \| ^2
&\leq& \left\|  \hat{T}_{\track{z}{u}}\psi - \hat{T}_{\track{z}{u}}\psi' \right\| ^2 +
 \left( \|\qubit{\phi_x}\|^2 - \|\qubit{\phi_{xz}}\|^2 \right) +
\left( \|\qubit{\phi_y}\|^2 - \|\qubit{\phi_{yz}}\|^2 \right) \\
&& \hs{10} + 4 \sqrt{ \left( \|\qubit{\phi_x}\|^2 - \|\qubit{\phi_{xz}}\|^2 \right) +
\left( \|\qubit{\phi_y}\|^2 - \|\qubit{\phi_{yz}}\|^2 \right) } \\
&<& \gamma + 2\alpha + 4 \sqrt{2\alpha} \\
&<& \gamma + 2\sqrt{\alpha} + 6\sqrt{\alpha} \;\;=\;\; \gamma + 8\sqrt{\alpha},
\end{eqnarray*}
where the last inequality comes from $\alpha\leq1$. Therefore, we obtain $\|  \psi - \psi' \| ^2 <  \gamma + 8\sqrt{\alpha}$, as requested.
\end{proof}

To verify Condition \ref{item:reverse}, let us assume that $(xz,n)\cong_{S}(yz,n)$, $(xz,n)=_{S}(x,n)$, and $(yz,n)=_{S}(y,n)$. In other words,  $\| \hat{T}_{\track{xz}{wu}}\psi_0 - \hat{T}_{\track{yz}{wu}}\psi_0 \| ^2 < \mu$,
$\| \qubit{\phi_{xz}}\| ^2-\| \qubit{\phi_{x}}\| ^2<\mu^2$, and $\| \qubit{\phi_{yz}}\| ^2-\| \qubit{\phi_{y}}\| ^2<\mu^2$. By setting $\gamma=\mu$ and $\alpha=\mu^2$ in
Claim \ref{T-norm-diff}, we conclude that  $\| \hat{T}_{\track{x}{w}}\psi_0 - \hat{T}_{\track{y}{w}}\psi_0 \| ^2 < \mu + 8\sqrt{\mu^2} \leq 9\mu$. Since $9\mu<2(1-2\varepsilon^*)$, Claim \ref{T-diff-bound} yields the equivalence $(x,n)\equiv_{S}(y,n)$. Therefore, Condition \ref{item:reverse} is true.

The proof of Theorem \ref{oneqfa-character} is now completed.
\end{proofof}

Theorem \ref{oneqfa-character} reveals a certain aspect of the characteristic  features of advised 1qfa's, from which we can deduce several important   consequences. Here, we intend to apply Theorem \ref{oneqfa-character} to demonstrate a class separation between $\reg$ and $\oneqfa/n$.
Without any use of advice, Kondacs and Watrous \cite{KW97} proved that
$\reg\nsubseteq \oneqfa$. Our class separation naturally extends their result and further indicates that 1qfa's are still not as powerful as 1dfa's even with a great help of advice.

\begin{corollary}\label{reg-oneqfa-advice}
 $\reg\nsubseteq \oneqfa/n$, and thus $\oneqfa/n\neq \reg/n$.
\end{corollary}

\begin{proof}
Our example language $S$ over a binary alphabet $\Sigma=\{a,b\}$ is
expressed in the form of {\em regular expression} as $(aa+ab+ba)^*$.
Since $S$ is obviously a regular language, hereafter we aim at verifying
that $S$ is located outside of $\oneqfa/n$. Assume otherwise; that is,  $S$ belongs to $\oneqfa/n$. Letting $\Delta = \{(x,n)\in\Sigma^*\times\nat\mid |x|\leq n\}$, Theorem \ref{oneqfa-character} guarantees the existence of two constants $c,d\in\nat^{+}$, an equivalence relation $\equiv_{S}$, a partial order $\leq_{S}$, and a closeness relation $\cong_{S}$ that satisfy
Conditions 1--7 given in the theorem.
We set $k =\max\{c,d\}$. Moreover, let $n$ denote the minimal {\em even} integer satisfying $n \geq (2k+1)(\ceilings{\log{k}}+1)$.

To draw a contradiction, we want to construct a special string, say, $x$ of length at most $n$. Inductively, we will build a series
$x_1,x_2,\ldots,x_m$ of
strings, each of which has length at most $2(\ceilings{\log{k}}+1)$, as long as the total length $|x_1\cdots x_m|$ does not exceed $n$.
For our convenience, set $x_0=\lambda$. The construction of such a series is described as follows. Assuming that $x_0,x_1,x_2,\ldots,x_i$ are already defined and they satisfy $|x_1\cdots x_i|< n- 2(\ceilings{\log{k}}+1)$, we
define $x_{i+1}$ in the following way. Let us denote by $\overline{x}_i$ the concatenated string $x_1x_2\cdots x_{i}$ and denote by $z_{i,w}$ the string $\overline{x}_iw$ for each given string $w$ in $((a+b)a)^*$ satisfying the inequity $|\overline{x}_iw|\leq n$. Now, we claim our key statement.

\begin{claim}\label{reduction-L}
There exists a nonempty string $w$ in $((a+b)a)^*$ such that both $|w|\leq 2(\ceilings{\log{k}}+1)$ and $(z_{i,w},n)<_{S} (\overline{x}_i,n)$ hold.
\end{claim}

Assuming that Claim \ref{reduction-L} is true, we choose the lexicographically-first nonempty string $w$ in $((a+b)a)^*$ that satisfies both $|w|\leq 2(\ceilings{\log{k}}+1)$ and $(z_{i,w},n)<_{S} (\overline{x}_{i},n)$.
The desired string $x_{i+1}$ in our construction is defined to be  this special  string $w$. Note that $\overline{x}_{i+1} = \overline{x}_{i}x_{i+1}$ holds. After the whole construction ends, let us assume that we have obtained $x_1,x_2,\ldots,x_m$. Obviously, it holds that
$|x_1x_2\cdots x_m|\leq 2m(\ceilings{\log{k}}+1)$.
Our construction also ensures that
$(\overline{x}_{m},n)<_{S} (\overline{x}_{m-1},n)<_{S} \cdots <_{S} (\overline{x}_1,n)$; thus, the sequence $((\overline{x}_{m},n), (\overline{x}_{m-1},n), \ldots,(\overline{x}_1,n))$ forms a strictly descending chain in $\Delta$.
Since $m\leq c$ by Condition \ref{item:chain}, $m\leq k$ follows.
Thus, we deduce $|x_1x_2\cdots x_m|\leq 2k(\ceilings{\log{k}}+1)$. Moreover, it holds that $|x_1x_2\cdots x_m|> n - 2(\ceilings{\log{k}}+1)$ because, otherwise, there still remains enough room for another string $x_{m+1}$ to satisfy, by Claim \ref{reduction-L}, both $|x_{m+1}|\leq 2(\ceilings{\log{k}}+1)$ and $(\overline{x}_{m+1},n)<_{S} (\overline{x}_{m},n)$, contradicting the maximality of the length $|x_1x_2\cdots x_m|$.  As a result, we obtain $n-2(\ceilings{\log{k}}+1) <|x_1x_2\cdots x_m|\leq 2k (\ceilings{\log{k}}+1)$, from which we conclude that $n < (2k+1)(\ceilings{\log{k}}+1)$. This is clearly a contradiction against  $n\geq (2k+1)(\ceilings{\log{k}}+1)$.   Therefore, $S$ cannot belong to $\oneqfa/n$.

To complete the proof of Corollary \ref{reg-oneqfa-advice}, it still remains to prove Claim \ref{reduction-L}.  This claim can be proven by a way of contradiction with a careful use of Conditions \ref{item:reverse}, \ref{item:equiv-S}, and \ref{item:discrepancy}.
Let us assume that $\overline{x}_i$ is already defined in our construction process. Toward a contradiction, we suppose that the claim fails; that is, for any nonempty string $w\in((a+b)a)^*$ with $|w|\leq 2(\ceilings{\log{k}}+1)$,  the equality
$(z_{i,w},n) =_{S}(\overline{x}_i,n)$ always holds. Under this assumption, it is possible to prove the following statement.

\begin{claim}\label{w-a-b}
For any two distinct pair $w,w'$ in $S$ with $|w|=|w'|\leq n-2$, it holds that $(wa,n)\not\equiv_{S} (w'b,n)$.
\end{claim}

For the time being, let us assume that Claim \ref{w-a-b} is true.
Let $X_k$ denote the set of all strings in $((a+b)a)^*$ of length exactly $2(\ceilings{\log{k}}+1)$.
Note that the total number of strings in $X_k$ is $2^{\ceilings{\log{k}}+1}\geq 2k$. We then define $G_n$ to be the set of all elements $(z_{i,w},n)\in \Delta$ associated with certain strings $w$ in $X_k$. Note that $|G_n|=|X_k|\geq 2k$.
Now, we want to show that $G_n$ is a $\cong_{S}$-discrepancy set.  Assume otherwise; that is,  two {\em distinct} strings $w,w'\in X_k$ satisfy $(z_{i,w},n)\cong_{S} (z_{i,w'},n)$.
For those special strings, there are (possibly empty) strings $y,y',z$ for which  $w=yaaz$ and $w'=y'baz$. Note that $|\overline{x}_iy|=|\overline{x}_iy'|\leq |z_{i,w}|-2\leq  n-2$ since $|z_{i,w}|\leq n$.
By applying Claim \ref{w-a-b} to the two strings $\overline{x}_iy$ and $\overline{x}_iy'$, we conclude that $(\overline{x}_iya,n) \not\equiv_{S}(\overline{x}_iy'b,n)$.
Since $(z_{i,w},n) =_{S} (\overline{x}_i,n) =_{S} (z_{i,w'},n)$ holds by our assumption, from $(z_{i,w},n)\cong_{S} (z_{i,w'},n)$,  Condition \ref{item:reverse} implies that $(\overline{x}_{i}ya,n)\equiv_{S}(\overline{x}_iy'b,n)$. This is a contradiction, and therefore $G_n$ is indeed a $\cong_{S}$-discrepancy subset of $\Delta$.  Condition \ref{item:discrepancy} then implies that $|G_n|\leq d\leq k$. However, this contradicts $|G_n|\geq 2k$. Therefore, Claim \ref{reduction-L} should hold.

Finally, let us prove Claim \ref{w-a-b} by induction on length $|w|$. Consider the case where $|w|=0$. Assume that $(a,n)\equiv_{S}(b,n)$.
The definition of $S$ implies the existence of a string $z$ for which  $|az|=n$ and $S(az)\neq S(bz)$. For instance, when $n=2$, it holds that
$S(ab)\neq S(bb)$.
However, Condition  \ref{item:equiv-S} yields $S(az) = S(bz)$, leading to a contradiction. Thus, it follows that $(a,n)\not\equiv_{S}(b,n)$.
Next, consider the case where $0<|w|\leq n-2$. Since $w,w'\in S$, there exists a string $z$ such that $|wabz|=n$ and $S(wabz) \neq S(w'bbz)$. If $(wa,n)\equiv_{S}(w'b,n)$, then Condition  \ref{item:equiv-S} also yields the equality  $S(wabz)=S(w'bbz)$, a contradiction. We thus conclude that $(wa,n)\not\equiv_{S}(w'b,n)$.
\end{proof}

\section{Power of Reversible Computation with Advice}\label{sec:reversible-advice}

As a special case of quantum computation, we turn our attention to {\em error-free} quantum computation and we wish to discuss characteristic behaviors of such computation, particularly assisted by resourceful deterministic advice. Since such error-free quantum computation has been known to coincide with ``reversible'' computation, we are focused on a model of {\em one-way (deterministic) reversible finite automaton}\footnote{This machine model is different from those defined in \cite{Ang82,Pin92}. See \cite{AF98} for more details.} (or 1rfa, in short), which was discussed in \cite{AF98}.
In this paper, a 1rfa is introduced as a 1dfa  $M=(Q,\Sigma,\delta,q_0,Q_{acc},Q_{rej})$ whose transition function  $\delta:Q\times\check{\Sigma}\to Q$ satisfies a special condition, called a {\em reversibility condition}; namely, for every inner state $q\in Q$ and every symbol $\sigma\in\check{\Sigma}$, there exists at most one inner state $q'\in Q$ that makes a transition $\delta(q',\sigma)=q$.
In a similar way as 1qfa's, as soon as $M$ enters any halting state, it instantly stops operating and accepts (resp., rejects) a given input instance if any accepting state (resp., rejecting state) is reached. Moreover, we demand that, for every input $x$, before or on reading the right endmarker $\dollar$, $M$
must halt (because, otherwise, we cannot determine whether $M$ accepts or rejects $x$).
We use the notation $\onerfa$ for the family of all languages recognized by those 1rfa's.
Analogous to $\reg/n$, the advised language family $\onerfa/n$ is composed of all languages $L$ over appropriate alphabets $\Sigma$  for which  there exist a 1rfa $M$ and a length-preserving advice function $h$ satisfying
$M(\track{x}{h(|x|)}) = L(x)$ for every string $x\in\Sigma^*$.  {}From the obvious relation $\onerfa\subseteq\oneqfa$ follows the containment $\onerfa/n\subseteq\oneqfa/n$.

In what follows, we will discuss more intriguing features of advised 1rfa's.


\subsection{A Necessary and Sufficient Condition for 1RFA/n}\label{sec:N-S-condition-1RFA}

In Theorem \ref{oneqfa-character}, we have presented a machine-independent, algebraic necessary condition for languages recognized by advised 1qfa's with bounded-error probability. When underlying finite automata are restricted to 1rfa's, it is possible to strengthen the theorem by giving a precise machine-independent, algebraic characterization of languages by advised 1rfa's.
This is our second main theorem, Theorem \ref{onerfa-characterization}.

\sloppy
\begin{theorem}\label{onerfa-characterization}
Let $S$ be any language over alphabet $\Sigma$ and define $\Delta = \{(x,n)\mid x\in\Sigma^*, n\in\nat$, $|x|\leq n\}$. The following two statements are logically equivalent.
\renewcommand{\labelitemi}{$\circ$}
\begin{enumerate}\vs{-2}
  \setlength{\topsep}{-2mm}%
  \setlength{\itemsep}{0mm}
  \setlength{\parskip}{0cm}%

\item $S$ is in $\onerfa/n$.

\item There are a total order $\leq_{S}$ over $\Delta$ and two equivalence relations $\simeq_{S}$ and $\equiv_{S}$ over $\Delta$ such that

\begin{enumerate}\vs{-1}
  \setlength{\topsep}{-2mm}%
  \setlength{\itemsep}{0mm}
  \setlength{\parskip}{0cm}%

\item[(i)] two sets $\Delta/\!\simeq_{S}$ and $\Delta/\!\equiv_{S}$ are both finite,

\item[(ii)] any strictly descending chain (with respect to $<_{S}$) in $\Delta$ has length at most $2$, and

\item[(iii)] for any length parameter $n\in\nat$, any two symbols $\sigma,\xi \in\Sigma$, and any three elements $(x,n),(y,n),(z,n)\in\Delta$ with $|x|=|y|$, the following seven  conditions hold.

\begin{enumerate}
  \setlength{\topsep}{-2mm}%
  \setlength{\itemsep}{0mm}
  \setlength{\parskip}{0cm}%

\item[(a)] If $|x\sigma|\leq n$, then $(x\sigma,n)\leq_{S}(x,n)$ and,
     if $|x|=n>0$, then $(x,n)<_{S}(\lambda,n)$.

\item[(b)] Whenever $|x\sigma|\leq n$, $(x\sigma,n)\simeq_{S}(y\sigma,n)$
 iff  $(x,n)\simeq_{S}(y,n)$.

\item[(c)] If $(x\sigma,n)<_{S}(x,n)=_{S}(z,n)$ with $|x\sigma|\leq n$, then $(x\sigma,n)\not\simeq_{S}(z,n)$.

\item[(d)] In the case where $(\lambda,n)=_{S} (x,n)=_{S}(z,n)$,  $(x,n)\equiv_{S}(z,n)$ iff $(x,n)\simeq_{S}(z,n)$.

\item[(e)] If $(x\sigma,n)<_{S}(x,n)$ and $(y\xi,n)<_{S}(y,n)$ with $|x\sigma|\leq n$ and $|y\xi|\leq n$, then $(x\sigma,n)\equiv_{S}(y\xi,n)$ iff  $(x\sigma,n)\simeq_{S}(y\xi,n)$. 
\item[(f)] If $(xz,n)=_{S}(x,n)$ with $|xz|=n$, then $(xz,n)\equiv_{S} (x,n)$.

\item[(g)] If $(x,n)\equiv_{S}(y,n)$, then $S(xz) = S(yz)$ holds for all strings $z\in\Sigma^*$ satisfying $|xz|=n$.
\end{enumerate}
\end{enumerate}
\end{enumerate}
\end{theorem}

This theorem requires three relations $\leq_{S}$, $\simeq_{S}$, and $\equiv_{S}$ as in Theorem \ref{oneqfa-character}; however, their roles are slightly different.  Condition (b) in this theorem particularly concerns the {\em reversibility} of a transition function of an underlying automaton.
Hereafter, we intend to give the proof of Theorem \ref{onerfa-characterization}.

\begin{proofof}{Theorem \ref{onerfa-characterization}}
Let $\Sigma$ be any alphabet, set $\Delta = \{(x,n)\mid x\in\Sigma^*, n\in\nat$, $|x|\leq n\}$, and pay our attention to
an arbitrary language $S$ over $\Sigma$. Before proceeding further, it is important to note that Lemma \ref{endmarker} also folds for $\onerfa/n$.

($1 \Rightarrow 2$) Assume that $S\in\onerfa/n$. By a 1rfa-version of
Lemma \ref{endmarker}, we can take an advice alphabet $\Gamma$, a 1rfa $M=(Q,\Sigma_{\Gamma},\delta,q_0,Q_{acc},Q_{rej})$ whose input tape uses no endmarker, and a length-preserving advice function $h:\nat\rightarrow\Gamma^*$ satisfying $M(\track{x}{h(|x|)}) = S(x)$ for all nonempty strings $x\in\Sigma^*$.
Now, we will introduce the desired relations $\leq_{S}$, $\simeq_{S}$, and $\equiv_{S}$ over $\Delta$.

Firstly, we define a function $f:\Delta\rightarrow Q$ by setting $f(x,n)$ to be an inner state $q\in Q$ satisfying $\hat{\delta}(q_0,\track{x}{w})=q$,
where $w$ is a unique string specified by $w= Pref_{|x|}(h(n))$.  In particular, $f(\lambda,0)=q_0$ holds.
Let $\simeq_{S}$ denote the kernel relation of $f$ (i.e., $(x,n)\simeq_{S}(y,m)$ iff $f(x,n)=f(y,m)$).
Clearly, the relation $\simeq_{S}$ is reflexive, symmetric, and transitive; thus, it is an equivalence relation.

Secondly, we define another function $\mu:\Delta\to \{0,1\}$. Let $\mu(\lambda,0) =0$. Given any nonempty string $x$ with $|x|\leq n$, if  $M$ never enters any halting state while reading $\track{x}{w}$, then we define $\mu(x,n)=1$; otherwise, we set $\mu(x,n)=0$. Notice that $\mu(\lambda,n)=1$ holds when $n>0$. The desired
total order $\leq_{S}$ is simply defined as follows: $(x,n)\leq_{S}(y,m)$ if $\mu(x,n)\leq \mu(y,m)$.
For convenience, let $\Phi=\{(x\sigma,n)\in\Delta\mid \mu(x\sigma,n)<\mu(x,n)\}\cup \{(x,n)\in\Delta\mid \mu(x,n)=1\}$.
For any $(x,n),(y,m)\in \Phi$, we set $(x,n)\equiv_{S}(y,m)$ exactly when $f(x,n)=f(y,m)$. In addition, when $\mu(xz,n)=\mu(x,n)=0$, we also set  $(xz,n)\equiv_{S}(x,n)$. It is not difficult to show that $\equiv_{S}$ is an equivalence relation.

The next goal is to establish Conditions (i)--(iii).

(i) The number of equivalence classes in the set $\Delta/\!\simeq_{S}$ is at most $|Q|$ because the range of $f$ is $Q$. Since $Q$ is finite,  $\Delta/\!\simeq_{S}$ is obviously a finite set. Since $|\Delta/\!\equiv_{S}| \leq 2|Q|$, the set $\Delta/\!\equiv_{S}$ is also finite.

(ii) Since $\mu$  takes at most two values, any strictly descending chain must have length at most $2$ as well.

(iii)  Take any symbol $\sigma\in\Sigma$, any string $z\in\Sigma^*$, and two arbitrary elements $(x,n),(y,n)$ in $\Delta$ satisfying $|x|=|y|$. Hereafter, we intend to show Conditions (a)--(e). When $n=0$, those conditions are trivially met; thus, we will consider only the case of $n\geq1$.

(a) Notice that $\mu(x\sigma,n)\leq \mu(x,n)$ for all $x$ with $|x|\leq n$. This implies $(x\sigma,n)\leq_{S}(x,n)$, as requested. For every string $x$ in $\Sigma^n$, $M$ always enters a halting state while reading $\track{x}{w}$; thus, $\mu(x,n)=0$ must hold. This implies that $(x,n)<_{S}(\lambda,n)$ whenever $n\geq1$.

(b) Assume that $|x\sigma|\leq n$. Using the aforementioned string $w$, let $\tau$ denote an advice symbol satisfying $w\tau = Pref_{|x\sigma|}(h(n))$.
If $(x\sigma,n)\simeq_{S}(y\sigma,n)$, then $\hat{\delta}(q_0,\track{x\sigma}{w\tau}) =\hat{\delta}(q_0,\track{y\sigma}{w\tau})=q$ holds for a certain inner state $q\in Q$. Let $p$ and $p'$ be two inner states for which  $\hat{\delta}(q_0,\track{x}{w})=p$ and $\hat{\delta}(q_0,\track{y}{w})=p'$.
It thus follows that $\delta(p,\track{\sigma}{\tau})=\delta(p',\track{\sigma}{\tau})=q$. The reversibility condition of $\delta$ then ensures that $p=p'$;
in other words,  $f(x,n)=f(y,n)$.
This clearly leads to the desired conclusion $(x,n)\simeq_{S}(y,n)$.
Moreover, since $M$ is deterministic by nature,
$(x,n)\simeq_{S}(y,n)$ implies $(x\sigma,n)\simeq_{S}(y\sigma,n)$. Therefore, Condition (b) in the theorem is satisfied.

(c) Assume that $(x\sigma,n)<_{S}(x,n) =_{S}(z,n)$ with $|x\sigma|\leq n$; that is, $\mu(x\sigma,n)=0$ and $\mu(x,n)=\mu(z,n)=1$ when $n\geq1$.
Note that $M$ enters an appropriate halting state just after reading off $\track{x\sigma}{w\tau}$ but $M$ never enters any halting state while reading $\track{z}{u}$, where $u=Pref_{|z|}(h(n))$. This means that $f(x\sigma,n)\neq f(z,n)$; thus, $(x\sigma,n)\not\simeq_{S}(z,n)$ holds.

(d) Assume that $(\lambda, n)=_{S} (x,n)=_{S}(y,n)$. Note that  $\mu(x,n)$ and $\mu(y,n)$ take the value $1$ since $\mu(\lambda,n)=1$. By the above definition, $(x,n)\simeq_{S}(y,n)$ is equivalent to $(x,n)\equiv_{S}(y,n)$.

(e) Assume that $(x\sigma,n)<_{S}(x,n)$ and $(y\xi,n)<_{S}(y,n)$. Since $\mu(x\sigma,n)=\mu(y\xi,n)=0$ and $\mu(x,n)=\mu(y,n)=1$, it follows that  $(x\sigma,n),(y\xi,n)\in\Phi$. By the definition of $\equiv_{S}$,
$(x\sigma,n)\equiv_{S}(y\xi,n)$ iff  $f(x\sigma,n)=f(y\xi,n)$. Recall that $f(x\sigma,n)=f(y\xi,n)$ iff $(x\sigma,n)\simeq_{S}(y\xi,n)$.
Therefore, we conclude that $(x\sigma,n)\equiv_{S}(y\xi,n)$ iff $(x\sigma,n)\simeq_{S}(y\xi,n)$.

(f) Assume that $(xz,n) =_{S} (x,n)$ with $|xz|=n$. Since $M$ halts on every nonempty input string, it must follow that $\mu(x,n)= \mu(xz,n) =0$. This implies $(x,n)\equiv_{S}(xz,n)$.

(g) Let $z$ be any string satisfying $|xz|=n$.  Assume that $(x,n)\equiv_{S}(y,n)$ holds. There are two cases to consider separately. In the first case where $\mu(x,n)=\mu(y,n)=1$, since $(x,n)\equiv_{S}(y,n)$,
$M$ enters the same inner state after reading $\track{x}{w}$ as well as after reading $\track{y}{w}$. Since $M$ is deterministic, $M$ must behave exactly in the same way while reading the remaining input  string $\track{z}{u}$,
provided that $u$ satisfies $wu=Pref_n(h(n))$. Therefore, $M$ accepts $\track{xz}{wu}$ iff $M$ accepts $\track{yz}{wu}$.  In other words, $S(xz) = S(yz)$ holds, as requested.
Next, consider the second case where $\mu(x,n)=\mu(y,n)=0$.
Take symbols $\sigma,\xi$ and strings $x',x'',y',y''$ for which $x=x'\sigma x''$, $y=y'\xi y''$, $\mu(x'\sigma,n)<\mu(x',n)$, and $\mu(y'\xi,n)<\mu(y',n)$.
Since $(x,n)\equiv_{S}(y,n)$, $M$ must enter the same halting state just after reading both $\track{x'\sigma}{w'}$ and $\track{y'\xi}{w''}$, where $w'=Pref_{|x'\sigma|}(h(n))$ and $w''=Pref_{|y'\xi|}(h(n))$.
Therefore, for any $z$ with $|xz|=n$,  we derive $S(x'\sigma x''z)=S(y'\xi y''z)$, that is, $S(xz)=S(yz)$.

\s

(2 $\Rightarrow$ 1)
Assume that we have three relations $\leq_{S}$, $\simeq_{S}$, and $\equiv_{S}$ satisfying Conditions (i)--(iii) of the theorem.
In what follows, we will construct an advice function $h$ and a 1rfa $M$ having the two endmarkers for which $M$ on input $\track{x}{h(|x|)}$ outputs $S(x)$ for {\em all} strings $x$ in $\Sigma^*$. This implies that $S$ is indeed in $\onerfa/n$.
Meanwhile, we ignore the empty string and consider only the set
$\Sigma^{*}-\{\lambda\}$.
By Condition (i), we set $d=|\Delta/\!{\simeq_{S}}|$ and assume that  $\Delta/\!{\simeq_{S}} = \{A_1,A_2,\ldots,A_d\}$, where each $A_i$ is an equivalence class. Notice that $\Delta$ coincides with
$\bigcup_{i=1}^{d} A_i$.

Let us introduce a useful notion, called a {\em turning point}, which roughly marks the transition point of the value $\mu(\cdot,\cdot)$ along a series $\{(Pref_{i}(x),n)\}_{i\in[0,|x|]_{\integer}}$ of substrings of input $x$.
Formally, for any nonempty string $x=\sigma_1\sigma_2\cdots \sigma_n\in\Sigma^n$, if $(\sigma_1,n)<_{S}(\lambda,n)$, then the turning point of $x$ is $1$; otherwise, the turning point of $x$
is an index $i\in[2,n]_{\integer}$ satisfying $(\sigma_1\sigma_2\cdots \sigma_{i},n)<_{S} (\sigma_1\sigma_2\cdots \sigma_{i-1},n)$. Notice that, by Conditions (ii) and (a), the turning point of $x$ is unique.

Given any length $n\in\nat^{+}$, we set $C^{(n)}_{acc}$ as the collection of all elements $q$ in $[d]$ satisfying the following condition: there exists a string $x_0\in\Sigma^n$ such that $(Pref_{i_0}(x_0),n)\in A_{q}$ and $S(x_0)=1$ hold for the turning point $i_0$ of $x_0$.
Likewise, we define $C^{(n)}_{rej}$ by replacing ``$S(x_0)=1$'' in the above definition with ``$S(x_0)=0$.''
It is easy to see that two sets $\{C^{(n)}_{acc}\mid n\in\nat^{+}\}$ and $\{C^{(n)}_{rej}\mid n\in\nat^{+}\}$ are both finite.

\begin{claim}\label{dfa-property}
For any $n\in\nat^{+}$, $x,y\in\Sigma^*$ with $|x|=|y|$, $\sigma\in\Sigma$, and $q\in[d]$, the following four properties hold.
\begin{enumerate}\vs{-1}
  \setlength{\topsep}{-2mm}%
  \setlength{\itemsep}{0mm}
  \setlength{\parskip}{0cm}%

\item If $(x,n)\in\Delta$, then there exists a unique index $q'\in[d]$ such that   $(x,n)\in A_{q'}$.

\item If $(x,n),(y,n)\in A_{q}$ and $|x|<n$, then there is a unique index $q'\in [d]$ such that $(x\sigma,n),(y\sigma,n)\in A_{q'}$.

\item If $(x\sigma,n),(y\sigma,n)\in A_{q}$ and $|x\sigma|\leq n$, then there is a unique index $q'\in[d]$ for which  $(x,n),(y,n)\in A_{q'}$.

\item It holds that $C^{(n)}_{acc}\cap C^{(n)}_{rej}=\setempty$.
\end{enumerate}
\end{claim}

\begin{proof}
(1)  Since the union $\bigcup_{i=1}^{d}A_i$ covers $\Delta$, each element $(x,n)$ in $\Delta$ belongs to a certain set $A_{q'}$ for $q'\in[d]$. The uniqueness of this index $q'$ comes from the fact that all sets in $\Delta/\!\simeq_{S}$ are mutually disjoint.

(2) Note that $(x,n),(y,n)\in A_{q}$ implies $(x,n)\simeq_{S}(y,n)$. Moreover, since $|x\sigma|\leq n$,   $(x\sigma,n)\simeq_{S}(y\sigma,n)$ immediately follows from $(x,n)\simeq_{S}(y,n)$ by Condition (b).
For this element $(x\sigma,n)$, we apply Claim \ref{dfa-property}(1) to obtain a unique index $q'\in[d]$ satisfying $(x\sigma,n)\in A_{q'}$.
In a similar fashion, from $(x\sigma,n)\simeq_{S}(y\sigma,n)$ follows the membership $(y\sigma,n)\in A_{q'}$ as well.

(3) Since $(x\sigma,n),(y\sigma,n)\in A_{q}$, it holds that   $(x\sigma,n)\simeq_{S}(y\sigma,n)$. Condition (b) then ensures that $(x,n)\simeq_{S}(y,n)$. The desired consequence follows from Claim \ref{dfa-property}(1).

(4) For the disjointness of $C^{(n)}_{acc}$ and $C^{(n)}_{rej}$, let us assume that there is a common element $q\in[d]$ inside $C^{(n)}_{acc}\cap C^{(n)}_{rej}$. For such a $q$, take two strings $x,y\in\Sigma^n$ for which  $(Pref_{i_x}(x),n),(Pref_{i_y}(y),n)\in A_{q}$ and $S(x)\neq S(y)$, provided that $i_x$ and $i_y$ are respectively the turning points of $x$ and $y$. For simplicity, write $x'$ (resp., $y'$) for $Pref_{i_x}(x)$ (resp.,  $Pref_{i_y}(y)$).
Since $(x',n),(y',n)\in A_{q}$ leads to $(x',n)\simeq_{S}(y',n)$, it follows from Condition (e) that $(x',n)\equiv_{S}(y',n)$. Moreover,  since $(x,n)=_{S}(x',n)$, by Condition (f), we immediately obtain $(x,n)\equiv_{S}(x',n)$. Similarly, we obtain $(y,n)\equiv_{S}(y',n)$.
The transitivity of $\equiv_{S}$ thus yields $(x,n)\equiv_{S}(y,n)$. Condition (g) finally implies $S(x)=S(y)$. This is obviously  a contradiction, and hence $C^{(n)}_{acc}\cap C^{(n)}_{rej}$ should be empty.
\end{proof}

Based on Claim \ref{dfa-property}, we wish to define an appropriate advice function $h$. For our purpose, let   $n$ ($\in\nat^{+}$) be an arbitrary length and let $\#$ be a special symbol not in $\Sigma\cup[d]$.
Given any index $i\in[n]$, we will introduce a series of {\em finite functions} $h_{n,i}:[d]\times \Sigma\to ([d]\cup\{\#\})\times \{C^{(n)}_{acc}\}_{n\in\nat^{+}}\times \{C^{(n)}_{rej}\}_{n\in\nat^{+}}$. Let $q$ and $q'$ be any two indices in $[d]$ and let $\sigma$ be any symbol in $\Sigma$.
\begin{enumerate}
  \setlength{\topsep}{-2mm}%
  \setlength{\itemsep}{0mm}
  \setlength{\parskip}{0cm}%

\item[(i)] Let $h_{n,1}(1,\sigma) =   (q',C_{acc}^{(n)},C_{rej}^{(n)})$ if $(\sigma,n)\in A_{q'}$ holds. For any $q\neq1$, let $h_{n,1}(q,\sigma)=(\#,C^{(n)}_{acc},C^{(n)}_{rej})$.

\item[(ii)] For any $i\in[2,n]_{\integer}$, let $h_{n,i}(q,\sigma) = (q',C_{acc}^{(n)},C_{rej}^{(n)})$ if   both $(x,n)\in A_{q}$ and $(x\sigma,n)\in A_{q'}$ hold for an appropriate string $x\in\Sigma^{i-1}$. If there is no such string $x\in\Sigma^{i-1}$, then we set  $h_{n,i}(q,\sigma)= (\#,C_{acc}^{(n)},C_{rej}^{(n)})$ for any symbol $\sigma\in\Sigma$.
\end{enumerate}
Finally, we set $\Gamma=\{h_{n,i}\mid n\geq 1, i\in[n]\}$. Since $\Gamma$ is a finite set, we enumerate all elements in $\Gamma$ as $h'_1,h'_2,\ldots,h'_{e}$ and we treat each element $h'_i$ as a new ``advice symbol.'' Our advice string $h(n)$ of length $n$ is set to be $h_{n,1}h_{n,2}\cdots h_{n,n}$, where each $h_{n,i}$ corresponds to a unique advice symbol listed above.

\begin{claim}\label{function-h}
\begin{enumerate}
  \setlength{\topsep}{-2mm}%
  \setlength{\itemsep}{0mm}
  \setlength{\parskip}{0cm}%

\item The above defined $h$ is indeed a function.

\item Let $q_1,q_2,q'\in[d]$. If $h_{n,i}(q_1,\sigma)=h_{n,i}(q_2,\sigma)=(q',C^{(n)}_{acc},C^{(n)}_{rej})$, then $q_1=q_2$.
\end{enumerate}
\end{claim}

\begin{proof}
(1) For every symbol $\sigma\in\Sigma$, Claim~\ref{dfa-property}(1) provides a unique index $q'\in [d]$ that satisfies $(\sigma,n)\in A_{q'}$. This proves that $h_{n,1}$ is indeed a function.
Next, let $i\in[2,n]_{\integer}$ and assume that $h_{n,i}(q,\sigma)= (q',C^{(n)}_{acc},C^{(n)}_{rej})$ and $h_{n,i}(q,\sigma)=(q'',C^{(n)}_{acc},C^{(n)}_{rej})$ for two indices $q',q''\in[d]$.
By the definition of $h_{n,i}$, there exist two strings $x,y\in\Sigma^{i-1}$ for which $(x,n),(y,n)\in A_{q}$, $(x\sigma,n)\in A_{q'}$, and $(y\sigma,n)\in A_{q''}$. Since $|x|=|y|<n$,  the uniqueness condition of Claim~\ref{dfa-property}(2) implies $q'=q''$.
Therefore, $h_{n,i}$ is also a function.

(2) The case of $i=1$ follows from the definition of $h_{n,1}$ and Claim \ref{dfa-property}(1).  Hereafter, we consider the case of $i\geq2$.
Let us assume that
$h_{n,i}(q_1,\sigma)=h_{n,i}(q_2,\sigma) = (q',C^{(n)}_{acc},C^{(n)}_{rej})$. Since $q'\neq\#$, certain two strings $x,y\in\Sigma^{i-1}$ satisfy that $(x,n)\in A_{q_1}$, $(y,n)\in A_{q_2}$, and $(x\sigma,n),(y\sigma,n)\in A_{q'}$. By Claim \ref{dfa-property}(3), the equality $q_1=q_2$ follows immediately.
\end{proof}

In the following argument, we will abbreviate $(q,C_{acc}^{(n)},C_{rej}^{(n)})$ as $\bar{q}$ as long as ``$n$'' is clear from the context, and we write $\bar{h}_{n,i}(q,\sigma)=q'$ whenever $h_{n,i}(q,\sigma)= (q',C_{acc}^{(n)},C_{rej}^{(n)})$ holds.

Next, we  will define a finite automaton $M = (Q,\Sigma_{\Gamma},\delta,q_0,Q_{acc},Q_{rej})$ with $q_0=0$.
In the proof below, we assume that $\lambda\in L$. The proof for the
other case of $\lambda\not\in L$ is essentially the same. Now, let $q_f = -1$.
The set $Q$ is composed of $q_0$, $q_f$, and all triplets of the form $(q,C^{(n)}_{acc},C^{(n)}_{rej})$ for any $n\in\nat^{+}$ and any $q\in [d]$.
Likewise, the set $Q_{acc}$  consists of $q_f$ and all triplets   $(q,C^{(n)}_{acc},C^{(n)}_{rej})$ in $Q$ satisfying
$q\in C^{(n)}_{acc}$. The set $Q_{rej}$ is similarly defined using $C^{(n)}_{rej}$ except for $q_f$.

As usual, we set $Q_{halt} = Q_{acc}\cup Q_{rej}$.
Our transition function $\delta: Q\times (\Sigma_{\Gamma}\cup\{\cent,\dollar\})\to Q$ is defined as follows.
Fix $n\in\nat^{+}$ arbitrarily.
Initially, we set $\delta(q_0,\cent)=q_0$, $\delta(q_0,\dollar)=q_f$, $\delta(q_f,\dollar)=q_0$,
and $\delta(\bar{q},\dollar) = \bar{q}$ for every $q\in[d]\cup\{\#\}$.
For any symbol $\sigma\in\Sigma$, we further define
 $\delta(q_0,\track{\sigma}{h_{n,1}}) = h_{n,1}(1,\sigma)$.
Given any index $i\in[2,n]_{\integer}$ and any $q\in [d]$, whenever  $\bar{h}_{n,i}(q,\sigma)\neq\#$, we set
$\delta(\bar{q},\track{\sigma}{h_{n,i}}) = h_{n,i}(q,\sigma)$.
For convenience,  all inputs of the form $(\bar{q},\track{\sigma}{\tau})$
defined so far are said to be {\em legitimate} for $\delta$.
For the other remaining inputs $(\bar{q},\track{\sigma}{\tau})$, which are distinctively called {\em illegitimate}, we define the values of $\delta(\bar{q},\track{\sigma}{\tau})$ arbitrarily as long as $\delta$ is ``reversible'' on the set of all illegitimate inputs.

\begin{claim}\label{series-A-vs-delta}
Let $n\in\nat^{+}$ and let $x=\sigma_1\sigma_2\cdots \sigma_n\in\Sigma^n$.
Assume that $(\sigma_1,n)\in A_{q_1}$,  $(\sigma_1\sigma_2,n)\in A_{q_2}$,  $(\sigma_1\sigma_2\sigma_3,n)\in A_{q_3}$, $\ldots$, $(x,n)\in A_{q_n}$. It follows that $\bar{q}_i = \hat{\delta}(q_0,\track{\sigma_1\cdots\sigma_i}{h_{n,1}\cdots h_{n,i}})$ for any index $i\in[n]$.
\end{claim}

\begin{proof}
We prove by induction
on $i\in[n]$   that $\bar{q}_i = \hat{\delta}(q_0,\track{\sigma_1\cdots\sigma_i}{h_{n,1}\cdots h_{n,i}})$.
Consider the basis case of $i=1$. In this case,  we obtain  $\hat{\delta}(q_0,\track{\sigma_1}{h_{n,1}}) = \delta(q_0,\track{\sigma_1}{h_{n,1}}) = h_{n,1}(1,\sigma_1) = \bar{q}_1$. For induction step $i\geq1$, our induction hypothesis guarantees that $\bar{q}_i = \hat{\delta}(q_0,\track{\sigma_1\cdots\sigma_i}{h_{n,1}\cdots h_{n,i}})$. It thus immediately follows that
\[
\hat{\delta}(q_0,\track{\sigma_1\cdots\sigma_{i+1}}{h_{n,1}\cdots h_{n,i+1}})
= \delta(\bar{q}_i,\track{\sigma_{i+1}}{h_{n,i+1}})
=
h_{n,i+1}(q_i,\sigma_{i+1})
= \bar{q}_{i+1}.
\]
By the mathematical induction, the claim is true.
\end{proof}

\begin{claim}
The transition function $\delta$ is reversible.
\end{claim}

\begin{proof}
By induction on $i$ (for $h_{n,i}$), we want to prove that $\delta$ is reversible. To achieve our goal, it suffices to  verify the reversibility condition of $\delta$ only on the set of all legitimate inputs; namely, for every $\sigma\in \Sigma$, $q'\in [d]$, $m,n\in\nat^{+}$, and $i\in[n]$,
(*) there is at most one index $q\in [d]$ such that $\delta(\bar{q},\track{\sigma}{h_{n,i}}) = (q',C^{(m)}_{acc},C^{(m)}_{rej})$, where $\bar{q}$ denotes $(q,C^{(n)}_{acc},C^{(n)}_{rej})$.

If $(C^{(n)}_{acc},C^{(n)}_{rej}) \neq (C^{(m)}_{acc},C^{(m)}_{rej})$, then $h_{n,i}(q,\sigma)\neq (q',C^{(m)}_{acc},C^{(m)}_{rej})$ holds for any $q\in[d]$, and thus Statement (*) is clearly true. Henceforth, let us
consider the case of $m=n$. Now, let $\bar{q}'=(q',C^{(n)}_{acc},C^{(n)}_{rej})$.
We assume that two appropriate legitimate inputs $(\bar{q}_1,\track{\sigma}{h_{n,i}})$ and  $(\bar{q}_2,\track{\sigma}{h_{n,i}})$ satisfy  $\delta(\bar{q}_1,\track{\sigma}{h_{n,i}}) = \delta(\bar{q}_2,\track{\sigma}{h_{n,i}}) = \bar{q}'$,

[Case: $i=1$]
The legitimacy of $(\bar{q}_1,\track{\sigma}{h_{n,1}})$ and $(\bar{q}_2,\track{\sigma}{h_{n,1}})$ implies  that $h_{n,1}(q_1,\sigma)\neq (\#,C^{(n)}_{acc},C^{(n)}_{rej})$. This leads to the equality $\bar{q}_1 = q_0 =\bar{q}_2$, implying Statement (*).

[Case: $i\geq 2$]
Note that $\delta(\bar{q}_1,\track{\sigma}{h_{n,i}}) = h_{n,i}(q_1,\sigma) = \bar{q}'$ and $\delta(\bar{q}_2,\track{\sigma}{h_{n,i}}) = h_{n,i}(q_2,\sigma) = \bar{q}'$. If $q_1,q_2\in[d]$, then Claim \ref{function-h}(2) implies $q_1=q_2$; thus, Statement (*) holds.
\end{proof}

Hereafter, we aim at proving that $M$ correctly recognizes $S$ with the help of the advice $h$. For this purpose, we will give a supplemental claim, Claim \ref{C-acc-C-rej}.  Let $n\in\nat^{+}$ and let $x=\sigma_1\sigma_2\cdots \sigma_n\in\Sigma^n$ with its turning point $i_0$. Assume that $(\sigma_1,n)\in A_{q_1}$,  $(\sigma_1\sigma_2,n)\in A_{q_2}$,  $(\sigma_1\sigma_2\sigma_3,n)\in A_{q_3}$, $\ldots$, $(x,n)\in A_{q_n}$.

\begin{claim}\label{C-acc-C-rej}
\begin{enumerate}
  \setlength{\topsep}{-2mm}%
  \setlength{\itemsep}{0mm}
  \setlength{\parskip}{0cm}%

\item  It holds that $q_{i_0}\in C^{(n)}_{acc}\cup C^{(n)}_{rej}$ and $q_j\notin C^{(n)}_{acc}\cup C^{(n)}_{rej}$ for any $j\in[i_0-1]$.

\item Given any $x\neq\lambda$, it follows that $x\in S$ implies $q_{i_0}\in C^{(n)}_{acc}$ and that $x\notin S$ implies $q_{i_0}\in C^{(n)}_{rej}$.
\end{enumerate}
\end{claim}

\begin{proof}
(1) Since $i_0$ is the turning point of $x$, $q_{i_0}$ must belong to $C^{(n)}_{acc}\cup C^{(n)}_{rej}$. Next, we will show that $q_j\notin C^{(n)}_{acc}\cup C^{(n)}_{rej}$ for any $j\in[i_0-1]$.
It is enough to consider the case where $i_0\geq2$. Let $j\in[i_0-1]$ and recall that $(Pref_{j}(x),n)\in A_{q_j}$.
Toward a contradiction, we assume that $q_j\in C^{(n)}_{acc}\cup C^{(n)}_{rej}$. This means that, for a certain string $w\in\Sigma^n$, $(Pref_{j_0}(w),n)\in A_{q_j}$ holds; thus,   $(Pref_{j}(x),n)\simeq_{S}(Pref_{j_0}(w),n)$ follows. Since $j<i_0$, we deduce   $(Pref_{j}(x),n)=_{S}(\lambda,n)$. Therefore, it holds that  $(Pref_{j_0}(w),n)<_{S}(Pref_{j_0-1}(w),n)=_{S}(Pref_j(x),n)$. By Condition (c), we obtain $(Pref_{j_0}(w),n)\not\simeq_{S}(Pref_j(x),n)$. This is obviously a contradiction.

(2) We want to show that $x\in S$ implies $q_{i_0}\in C^{(n)}_{acc}$. The other statement is similarly proven. Now, we assume by contradiction that $x\in S$ and $q_{i_0}\notin C^{(n)}_{acc}$. Since $x$ has the turning point $i_0$,
$(Pref_{i_0}(x),n)\in A_{q_{i_0}}$ implies $q_{i_0}\in C^{(n)}_{acc}\cup C^{(n)}_{rej}$, from which $q_{i_0}\in C^{(n)}_{rej}$ follows
by our assumption.
This implies that, for a certain appropriate string $w\in\Sigma^n$, $(Pref_{j_0}(w),n)\in A_{q_{i_0}}$ and $S(w)=0$, where $j_0$ is the turning point of $w$. It follows from  $(Pref_{j_0}(w),n)=_{S}(w,n)$ that   $(Pref_{j_0}(w),n)\equiv_{S}(w,n)$ by Condition (f).
Moreover, since $(Pref_{j_0}(w),n)\in A_{q_{i_0}}$, we obtain  $(Pref_{i_0}(x),n)\simeq_{S} (Pref_{j_0}(w),n)$. Using Condition (e), we conclude that $(Pref_{i_0}(x),n)\equiv_{S}(Pref_{j_0}(w),n)$.
Since $(Pref_{i_0}(x),n)\equiv_{S}(x,n)$, the property of $\equiv_{S}$ implies that $(x,n)\equiv_{S}(w,n)$. Condition (g) then leads to $S(x)=S(w)=0$, a contradiction. Therefore, $x\in S$ implies $q_{i_0}\in C^{(n)}_{acc}$.
\end{proof}

Finally, we argue that $S = \{ x\mid \text{ $M$ accepts $\track{x}{h(|x|)}$ }\}$. First, we consider the case where $x\in S$.
When $x=\lambda$, we obtain $\hat{\delta}(q_0,\cent\lambda\dollar)=q_f$, and thus $M$ accepts $x$. Now, we assume that $|x|\geq1$.
By Claim \ref{C-acc-C-rej}(1), we obtain $q_{i_0}\in C^{(n)}_{acc}$ and $q_j\notin C^{(n)}_{acc}\cup C^{(n)}_{rej}$ for all $j\in[i_0-1]$. This implies that, by Claim \ref{C-acc-C-rej}(2),  $\bar{q}_{i_0}\in Q_{acc}$ and $\bar{q}_{j}\notin Q_{halt}$ for any index $j\in[i_0-1]$. Note that, by Claim \ref{series-A-vs-delta}, for each $j\in[i_0]$, $M$ enters a unique inner state $\bar{q}_j$ after scanning $\track{\sigma_j}{h_{n,j}}$.
From this situation, $M$ should enter an accepting state $\bar{q}_{i_0}$ just after reading $\track{\sigma_1\cdots \sigma_{i_0}}{h_{n,1}\cdots h_{n,i_0}}$; in short, $M$ accepts $\track{x}{h(n)}$. The other case $x\not\in S$ is handled similarly as in the previous case, since the essential difference is only the final step.

This completes the proof of Theorem \ref{onerfa-characterization}.
\end{proofof}

As an immediate consequence of Theorem \ref{onerfa-characterization}, we will show  that $\oneqfa$ is not included in $\onerfa/n$. This result can be viewed as the strength of {\em bounded-error} quantum computation over {\em error-free} advised quantum computation.

\begin{corollary}\label{oneqfa-vs-onerfa}
$\oneqfa\nsubseteq \onerfa/n$, and thus $\onerfa/n\subsetneq \oneqfa/n$.
\end{corollary}

\begin{proof}
Let us consider a regular language $L=\{0^m1^n\mid m,n\in\nat\}$ over a   binary alphabet $\Sigma=\{0,1\}$. Ambainis and Freivalds \cite{AF98} showed how to recognize this language $L$ on a certain 1qfa with success probability at least $0.68$. To obtain the desired consequence, we need to show that $L\not\in\onerfa/n$. Note that $L$ was already proven in \cite{AF98} to be located outside of $\onerfa$ by a use of technical tool called a {\em forbidden construction}. Our result therefore not only extends this result but also provides a new proof technique based on Theorem \ref{onerfa-characterization}.

To lead to a contradiction, we assume that $L$ belongs to $\onerfa/n$. Theorem \ref{onerfa-characterization} guarantees the existence of three relations $\leq_{S}$, $\simeq_{S}$, and $\equiv_{L}$ on $\Delta$ that satisfy Conditions (i)--(iii) of the theorem. We denote by $k_1$ (resp., $k_2$)
the cardinality of the set $\Delta/\!\simeq_{L}$ (resp., $\Delta/\!\equiv_{L}$) of equivalence classes. Additionally, let $k=\max\{k_1,k_2\}$. For each index $n\in\nat$, write $L_n$ for the subset $\{0^i1^{n-i}\mid 0\leq i\leq n\}$ of $L$.

Now, for any index $n\geq4$, we want to assert that
\begin{quote}
(*) for each string $x\in L_n$, its turning point $i_0$ is at least $n-1$.
\end{quote}
Let us assume that there is a string $x\in L_n$ whose turning point $i_0$ is less than $n-1$. For simplicity, write $x'$ for $Pref_{i_0}(x)$ and let  $x'=x''\sigma$ with $\sigma\in\Sigma$ and $x''\in\Sigma^*$.  Since  $(x,n)=_{L}(x',n)$ holds by the definition of the turning point, we obtain $(x,n)\equiv_{L}(x',n)$ by Condition (f)  of the theorem.
We define  $z$ as $\bar{\sigma}0^{n-i_0-1}$.  We then obtain $(x'z,n)=_{L}(x',n)$. It thus follows that $(x'z,n)\equiv_{L}(x',n)$. The transitivity of $\equiv_{L}$ implies that $(x,n)\equiv_{L}(x'z,n)$.  Hence,
we conclude by Condition (g) that $L(x)=L(x'z)$. From  $x'z=x''\sigma\bar{\sigma}0^{n-i_0-1}$, it must hold that $x'z\notin L_n$. Therefore, we obtain $x\notin L_n$, a contradiction. As a result, Statement (*) holds.

Let us fix a number $n$ to satisfy $n>\max\{3,k+2\}$.
Since $i_0\geq n-1$, we want to consider the set $\Phi_n =\{(0^i1^{n-i-2},n)\mid 1\leq i\leq n-2\}$. From $|\Phi_n|=n-2>k\geq k_1$, it follows that there are
at least two indices $i,j\in[n-2]$ with $i<j$ for which  $(0^i1^{n-i-2},n)\simeq_{L}(0^j1^{n-j-2},n)$ holds.
By applying Condition (b) repeatedly, we obtain  $(0^i1^{j-i},n)\simeq_{L}(0^j,n)$.
It thus follows by Condition (d) that $(0^i1^{j-i},n)\equiv_{S}(0^j,n)$.
If we choose $z=0^{n-j}$, then Condition (g) further leads to
the equality $L(0^i1^{j-i}z) = L(0^jz)$. Since $i<j\leq n-2$, however, it holds that $L(0^i1^{j-i}z) = L(0^i1^{j-i}0^{n-j})=0$ and that $L(0^jz) = L(0^n) = 1$. This is
a contradiction. Therefore, $L$ cannot belong to $\onerfa/n$.

To see the second part of the corollary, we first recall that $\onerfa/n\subseteq \oneqfa/n$. Since $\oneqfa\subseteq \oneqfa/n$, the equality $\onerfa/n=\oneqfa/n$ leads to the containment $\oneqfa\subseteq \onerfa/n$. Clearly, this contradicts the first part. Therefore, we conclude that $\onerfa/n\neq \oneqfa/n$.
\end{proof}


\subsection{Reversible Computation with Randomized Advice}\label{sec:reverse-random-advice}

As a probabilistic variant of deterministic advice, {\em randomized advice}
was observed in \cite{Yam10} to endow an enormous computational power to one-way finite automata, where
randomized advice refers to a {\em probability ensemble} $\{D_n\}_{n\in\nat}$ consisting of an infinite series of probability distributions $D_n$ over the set $\Gamma^n$ of advice strings.
Those randomly chosen advice strings are given on the lower track of an input tape so that a tape head can scan a standard input and advice simultaneously.

Let us give a quick remark on the power of randomized advice. The notation $\onebplin/Rlin$ denotes the family of all languages recognized with bounded-error probability by one-tape one-head two-way off-line probabilistic Turing machines whose computation paths {\em all} terminate within {\em linear time} in the presence of randomized advice of {\em linear size} \cite{Yam10}. When such probabilistic Turing machines are replaced by 1dfa's and 1npda's, we obtain language families $\reg/Rn$ and $\cfl/Rn$, respectively, from $\onebplin/Rlin$ using advice of input size. It was shown in \cite{Yam10} that  $\reg/Rn$ is powerful enough to coincide with $\onebplin/Rlin$.  Moreover, it was proven  that $\reg/Rn\nsubseteq \cfl/n$ \cite{Yam10}, and thus $\cfl/n\neq \cfl/Rn$ follows.

Like the above notations $\reg/Rn$ and $\cfl/Rn$ introduced
in \cite{Yam10},
$\onerfa/Rn$ expresses the family of all languages $L$ that satisfy the following condition: there exist a 1rfa $M$, an error bound  $\varepsilon\in[0,1/2)$, an advice alphabet $\Gamma$,
and an advice probability ensemble $\{D_n\}_{n\in\nat}$
($D_n:\Gamma^n\rightarrow [0,1]$)  such that,
for every length $n\in\nat$ and any string $x$ of length $n$,
\begin{quote}
(*) $M$ on input $\track{x}{y}$ outputs $L(x)$ with probability at least $1-\varepsilon$ when $y$ is chosen at random according to $D_n$ (i.e.,  $y$ is chosen with probability $D_n(y)$).
\end{quote}
For our notational convenience, we introduce a succinct notation $\track{x}{D_n}$ to denote a {\em random variable}  expressing a string $\track{x}{y}$, provided
that $y$ ($\in\Gamma^n$) is chosen with probability $D_n(y)$.
With this notation, we rephrase Condition (*) as $\prob_{D_{n}}[M(\track{x}{D_{n}})=L(x)]\geq 1-\varepsilon$.

In what follows,  we demonstrate a strength of 1rfa's when they take randomized advice. Let $\dcfl$ express the deterministic counterpart of $\cfl$.

\begin{proposition}\label{det-vs-randomized-advice}
\renewcommand{\labelitemi}{$\circ$}
\begin{enumerate}\vs{-2}
  \setlength{\topsep}{-2mm}%
  \setlength{\itemsep}{0mm}
  \setlength{\parskip}{0cm}%

\item $\dcfl\cap\onerfa/Rn\nsubseteq \reg/n$.
\item $\onerfa/Rn \nsubseteq \cfl/n$.
\end{enumerate}
\end{proposition}

\begin{proof}
The following proof is in essence similar to the proof of \cite[Proposition 17]{Yam10}.

(1) For our purpose, we use a ``marked'' version of $Pal$, the set of {\em even-length palindromes}.
Now, define $Pal_{\#} =\{w\#w^R\mid w\in\{0,1\}^*\}$ as a language  over a  ternary alphabet $\Sigma=\{0,1,\#\}$. Similarly to the separation $Pal \not\in \reg/n$  \cite{Yam08}, it is possible to prove that $Pal_{\#}\not\in \reg/n$
by employing the so-called {\em swapping lemma} \cite{Yam08}.

Since $Pal_{\#}$ is known to be in $\dcfl$, the remaining task is to
show that $Pal_{\#}$ belongs to $\onerfa/Rn$.
As in Lemma \ref{endmarker}, we assume that input tapes of advised 1qfa's
have no endmarker.
Our advice alphabet $\Gamma$ is $\{0,1,\#\}$ and
our randomized advice $D_n$ of size $n$ is defined as follows. If $n=2m+1$, then  $D_n$  generates every string $w\# w^R$ ($w\in\{0,1\}^m$)
with equal probability $2^{-m}$; otherwise,  $D_n$ generates $\#^{n}$ with probability $1$.
Next, let us define a {\em one-way probabilistic finite automaton} (or a {\em 1pfa}) $M = (Q,\Sigma_{\Gamma},\delta,q_0,Q_{acc},Q_{rej})$ with
$Q_{acc}= \{q_0,q_2\}$, $Q_{rej} = \{q_1,q_3\}$, and $Q = Q_{acc}\cup Q_{rej}$.
The transition function $\delta$ of $M$ is defined as follows.
For any bits $\sigma,\tau\in\{0,1\}$ and any index $i\in\{0,1\}$,  we set   $\delta(q_i,\track{\sigma}{\tau}) = q_{\sigma\tau+i\;\mathrm{mod}\;2}$ and  $\delta(q_i,a) = q_{i+1\;\mathrm{mod}\;2}$, where $a=\track{\#}{\#}$ and $\sigma\tau$ is the numerical multiplication of $\sigma$ and $\tau$.
For any other state/symbol pair $(q,\sigma)$, we make two new transitions from $(q,\sigma)$ to both $q_2$ and $q_3$ with probability exactly $1/2$.

On any input of the form $x\#x'$, if $x'=x^{R}$, then $M$ enters an accepting state using $D_n$ with probability $1$, where this probability is calculated according to the transition probabilities of $M$ as well as the probability distribution $D_n$.
On the contrary, if $x'\neq x^{R}$, then $M$ enters an accepting state with probability exactly $1/2$, and thus an error probability is $1/2$.  To reduce this error probability to $1/4$, we need to make two runs of the above procedure in parallel.
To make sure that an input is of the form $x\#x'$, we also need to force the 1pfa to check that $\#$ occurs exactly once in scanning the entire input string.

It is not quite difficult to translate our 1pfa into an appropriate reversible automaton (by modifying randomized advice slightly), and thus we omit a detailed description of the desired 1rfa.

(2) In a way similar to (1), another language $Dup =\{ww\mid w\in\{0,1\}^*\}$ over a binary alphabet $\{0,1\}$ can be proven to fall into  $\onerfa/Rn$. Since  $Dup$ does not belong to $\cfl/n$ \cite{Yam08}, the proposition instantly follows.
\end{proof}

Since $\onerfa/Rn\subseteq \reg/Rn$, Proposition \ref{det-vs-randomized-advice}(2) in fact strengthens the early result of $\reg/Rn\nsubseteq \cfl/n$ \cite{Yam10}. Let us discuss briefly an
immediate consequence of Proposition \ref{det-vs-randomized-advice}(2). If $\onerfa/n=\onerfa/Rn$, then the obvious containment $\onerfa/n\subseteq \cfl/n$ yields a conclusion $\onerfa/Rn\subseteq \cfl/n$; however, this contradicts  Proposition \ref{det-vs-randomized-advice}(2). Therefore, we obtain a  class separation between $\onerfa/n$ and $\onerfa/Rn$.  This separation can be compared with $\reg/n\neq \reg/Rn$ in \cite{Yam10}.

\begin{corollary}\label{rand-det-rfa}
$\onerfa/n\neq \onerfa/Rn$.
\end{corollary}

\section{{}From Randomized Advice to Quantum Advice}\label{sec:randomized}

In Section \ref{sec:reverse-random-advice}, we have witnessed the extraordinary power of 1rfa's when augmented with appropriate randomized advice. In particular, we have shown in  Corollary \ref{rand-det-rfa} that randomized advice is much more useful for 1rfa's than deterministic advice is.
In a similar fashion, we want to supply randomized advice to assist 1qfa's and
we will discuss how much randomized advice enhances the recognition power of the 1qfa's. Next, we will extend randomized advice further to {\em quantum advice}.
After examining a situation surrounding the 1qfa's in the presence of quantum advice, we will consider how to make the most of the quantum advice to strengthen the computational power of the 1qfa's.

\subsection{Computational Complexity of 1QFA/{Rn}}\label{sec:complexity-Rn}

By natural analogy with $\onerfa/Rn$, we intend to introduce an advised language family $\oneqfa/Rn$. The most reasonable way to define a language $L$ in $\oneqfa/Rn$ is to demand the existence of  a 1qfa $M$, a constant $\varepsilon\in[0,1/2)$, an advice alphabet $\Gamma$,
and an advice probability ensemble $\{D_n\}_{n\in\nat}$
($D_n:\Gamma^n\rightarrow [0,1]$)  for which  (*)
$\prob_{M,D_{n}}[M(\track{x}{D_n})=L(x)]\geq 1-\varepsilon$ holds
for every length $n\in\nat$ and every string $x$ of length $n$.
Since $M$ performs quantum operations rather than deterministic operations of 1rfa's, we need to state Condition (*) more precisely. Let $M = (Q,\Sigma_{\Gamma},\{U_{\sigma}\}_{\sigma\in\check{\Sigma}_{\Gamma}}, q_0,Q_{acc},Q_{rej})$ denote any underlying 1qfa and let $\{D_n\}_{n\in\nat}$
be an advice probability ensemble over $\Gamma^*$.
Let $x=\sigma_1\sigma_2\cdots \sigma_n\in\Sigma^n$ and $y=\tau_1\tau_2\cdots\tau_n\in\Gamma^n$.
To simplify our notation, write $\track{\sigma_0}{\tau_0}$ and $\track{\sigma_{n+1}}{\tau_{n+1}}$ for $\cent$ and $\dollar$, respectively. Notice that the $i$th tape cell consists of symbol
$\track{\sigma_i}{\tau_i}$.
Let us define quantum states $\qubit{\phi^{(x,y)}_{0}} = \qubit{q_0}$ and $\qubit{\phi^{(x,y)}_{i+1}} = T_{\track{\sigma_i}{\tau_i}}\qubit{\phi^{(x,y)}_{i}}$ in the space $E_{Q}$ for any index $i\in[0,n]_{\integer}$.
In the presence of randomized advice $D_n$, the acceptance probability $p_{acc}(x,D_n)$ of $M$ on this input $x$ is defined as
\[
p_{acc}(x,D_n) = \sum_{y\in\Gamma^n}D_n(y)\sum_{i=0}^{n+1} \left\|  P_{acc}U_{\track{\sigma_i}{\tau_i}}\qubit{\phi^{(x,y)}_{i}} \right\| ^2.
\]
Likewise, the rejection probability $p_{rej}(x,D_n)$ is defined using $P_{rej}$ in place of $P_{acc}$. With those notations, Condition (*) is now  understood as asserting that  $p_{acc}(x,D_n)\geq 1-\varepsilon$ holds for all $x\in L\cap\Sigma^n$ and that $p_{rej}(x,D_n)\geq 1-\varepsilon$ holds for all $x\in \Sigma^n-L$.

Let us start with a simple observation on significance of the  ``bounded-error probability'' requirement for 1qfa's.
By augmenting  $\oneqfa_{(a(n),b(n))}$ with randomized advice,
we can define $\oneqfa_{(a(n),b(n))}/Rn$ as a parametrization of $\oneqfa/Rn$. Recall from Section \ref{sec:introduction} that the notation ``${\rm ALL}$'' indicates the collection of all languages.   When  the error probability of 1qfa's becomes arbitrarily close to $1/2$ (known as {\em unbounded-error probability}), a dexterous choice of randomized advice can make those underlying 1qfa's recognize all languages; thus, the lemma below follows.

\begin{lemma}
$\oneqfa_{(1/2,1/2)}/Rn = {\rm ALL}$.
\end{lemma}

\begin{proof}
Let $L$ be any language over alphabet $\Sigma$. We set our advice alphabet $\Gamma$ to be $\Sigma\cup\{\#\}$, where $\#$ indicates a special symbol not in $\Sigma$.  We intend to define the desired 1qfa $M =(Q,\Sigma_{\Gamma},\{U_{\xi}\}_{\xi\in\check{\Sigma}_{\Gamma}}, q_0,Q_{acc},Q_{rej})$ and randomized advice $\{D_n\}_{n\in\nat}$ that together recognize $L$ with unbounded-error probability.

Fix an arbitrary length $n\in\nat$. For simplicity of the proof, we consider only the case of $\lambda\in L$. We abbreviate each  set $L\cap\Sigma^n$ as   $L_n$. Now, we assume that $n\geq1$.
In the case where $L_n=\setempty$, $D_n$ generates $\#^n$ with probability $1$. By scanning the first symbol of $\#^n$, $M$ easily concludes that  $L_n=\setempty$, and thus it immediately rejects any input string of length $n$.
Next, assuming $L_n \neq \setempty$, we set our randomized advice $D_n$ as  $D_n(y) =  1/|L_n|$ for any string $y\in L_n$ and $D_n(y)=0$ for all the other strings $y$ in $\Sigma^n$.

Our 1qfa $M$ is designed to work as follows.
If our input $x$ is $\lambda$, then we force $M$ to accept it after reading $\cent\lambda\dollar$. Otherwise, for each advice string $s$, (i) $M$ checks whether its input $\track{x}{s}$ satisfies $x=s$, (ii) if so, then $M$ accepts the input with certainty, and (iii) otherwise, $M$ accepts and rejects the input with equal probability.
To perform these steps, we first define $Q=\{q_0,q_1,q_2\}$, $Q_{acc}=\{q_1\}$, and $Q_{rej}=\{q_2\}$. The time-evolution operators $\{U_{\xi}\}_{\xi\in\check{\Sigma}_{\Gamma}}$ are defined
as $U_{\cent}= U_{\track{\sigma}{\sigma}} = I$ (identity), and
\[
\hs{-5}
{\small U_{\track{\sigma}{\tau}} = \ninematrices{0}{1}{0}{\frac{1}{\sqrt{2}}}{0}{\frac{1}{\sqrt{2}}}{\frac{1}{\sqrt{2}}}{0}{-\frac{1}{\sqrt{2}}}}, 
\;\;\;\;
{\small U_{\track{\sigma}{\#}} =  \ninematrices{0}{0}{1}{0}{1}{0}{1}{0}{0}},
\;\;\text{and}\;\;
{\small U_{\dollar} =  \ninematrices{0}{1}{0}{1}{0}{0}{0}{0}{1}},
\]
where $\sigma\neq\tau$.  Note that an initial quantum state of $M$ is $\qubit{q_0}$ ($=(1,0,0)^T$).
It is straightforward to verify that, for any $n\in\nat$ and any $x\in\Sigma^n$,  $x\in L_n$ iff $\prob_{M,D_n}[M(\track{x}{D_n})=1]>1/2$.
Therefore, $L$ belongs to $\oneqfa_{(1/2,1/2)}/Rn$.
\end{proof}


As for deterministic advice, we have remarked in Section \ref{sec:basic-property-QFA/n}  that  $\oneqfa/n$ is contained in $\reg/n$.
When randomized advice is concerned,  a similar containment holds between $\oneqfa/Rn$ and $\reg/Rn$; however, this fact is not quite obvious from their definitions.

\begin{lemma}\label{randomized-advice-vs-reg}
$\oneqfa/Rn\subseteq \reg/Rn$.
\end{lemma}

\begin{proof}
Fixing an input alphabet $\Sigma$, take any language $L$ in $\oneqfa/Rn$ over $\Sigma$. Let $M$ be a 1qfa, $\Gamma$ be an advice alphabet, and $\{D_n\}_{n\in\nat}$ be an advice probability ensemble over $\Gamma^*$, and assume that, for every length $n\in\nat$ and any string $x\in\Sigma^n$, $\prob_{M,D_n}[M(\track{x}{D_n})=L(x)] \geq 1-\varepsilon$ holds.
In what follows, we fix $n\in\nat$ and $x\in\Sigma^n$ arbitrarily.
We enumerate all strings in $\Gamma^n$ as $\{y_1,y_2,\ldots,y_{c^n}\}$ with $c=|\Gamma|$.
For each index $i\in[c^n]$, let $p_i = D_n(y_i)$ and $r_i= \prob_{M}[M(\track{x}{y_i})=L(x)]$ so that the value
$\prob_{M,D_n}[M(\track{x}{D_n})=L(x)]$ is succinctly expressed as
$\sum_{i=1}^{c^n}p_ir_i$.

Now, consider the set $A=\{i\in[c^n]\mid r_i\geq 1- 3\varepsilon\}$. First, we want to show that $\sum_{i\in A}p_i\geq 2/3$. By the definitions of $p_i$'s and $r_i$'s, it follows that
\[
\sum_{i=1}^{c^n} p_ir_i \leq \sum_{i\in A}p_i\cdot 1 + \sum_{i\not\in A}p_i(1-3\varepsilon) = 1-3\varepsilon + 3\varepsilon \sum_{i\in A}p_i,
\]
where the equality comes from the fact that $\sum_{i\not\in A}p_i= 1- \sum_{i\in A}p_i$. Since $\sum_{i=1}^{c^n}p_ir_i\geq 1-\varepsilon$ by our assumption, we conclude that $\sum_{i\in A}p_i\geq 2/3$.

As shown in \cite{KW97}, we can translate the underlying 1qfa $M$ into a certain  ``equivalent'' 1dfa, say, $N$. Unfortunately, this 1dfa $N$ may not always produce the same output as the original 1qfa does on the same input with ``high'' probability. Nonetheless, as far as we restrict our attention within the indices $i\in A$, $N$ correctly outputs $L(x)$  using $\{D_n\}_{n\in\nat}$ with probability at least $2/3$.
Therefore, $L$ belongs to $\reg/Rn$.
\end{proof}


Using Proposition \ref{det-vs-randomized-advice}(1), we can exemplify the usefulness of randomized advice for 1qfa's.

\begin{lemma}
$\oneqfa/n\neq \oneqfa/Rn$.
\end{lemma}

\begin{proof}
Assume that $\oneqfa/n=\oneqfa/Rn$.  Proposition \ref{det-vs-randomized-advice}(1) implies that $\onerfa/Rn\nsubseteq \reg/n$.  Since
$\onerfa/Rn\subseteq \oneqfa/Rn$ holds,  it follows that $\oneqfa/Rn\nsubseteq \reg/n$.  Thus, our assumption leads to a conclusion that  $\oneqfa/n\nsubseteq\reg/n$. This contradicts a fact stated in Section \ref{sec:basic-property-QFA/n} that $\oneqfa/n$ is a subclass of $\reg/n$.  Therefore, $\oneqfa/Rn$ is different from $\oneqfa/n$.
\end{proof}


Since quantum computation is well capable of handling quantum information, it is natural to consider a piece of special advice, known as {\em quantum advice}, which is a series of {\em pure quantum states}, introduced in  \cite{NY04b} for Turing machines.  In the past literature, quantum advice has been discussed chiefly in the context of polynomial-time computations (see, e.g.,  \cite{Aar05,NY04b,Raz09}).
Associated with an advice alphabet $\Gamma$, we denote by $\qubit{\phi_n}$  a {\em normalized} quantum state in a Hilbert space of dimension $|\Gamma|^n$ if $n\geq1$. Using a computational basis $\Gamma^n$,  $\qubit{\phi_n}$ can be expressed as  a superposition of the form $\sum_{s\in\Gamma^n}\alpha_{s}\qubit{s}$ with appropriate amplitudes $\alpha_{s}\in\complex$ satisfying  $\sum_{s\in\Gamma^n}|\alpha_{s}|^2=1$.
For our later convenience, the succinct notation $\qubit{\track{x}{\phi_n}}$ indicates a particular quantum state $\sum_{s\in\Gamma^n}\alpha_s\qubit{\track{x}{s}}$ represented in computational basis $\Sigma_{\Gamma}^{n}= \{\track{x}{s}\mid x\in\Sigma^n, s\in\Gamma^n\}$.
For our convenience, when $n=0$, we set $\ket{\phi_0}=\ket{\lambda}$.

To treat the quantum advice formally, it is more convenient to {\em rephrase} the earlier definition of advised 1qfa given in Section \ref{sec:QFA/n} by expanding its original Hilbert space $E_{Q} = span\{\qubit{q}\mid q\in Q\}$ used  to  a larger Hilbert space $E_n = span\{\qubit{q}\qubit{y}\mid q\in Q,y\in\Gamma^n\}$,  where $n$ refers to input size.
First, we express our ``new'' advised 1qfa, say, $M$ as a septuple $(Q,\Sigma,\Gamma,\{U_{\sigma}\}_{\sigma\in\check{\Sigma}}, q_0,Q_{acc},Q_{rej})$. Each entry $U_{\sigma}$ associated with symbol  $\sigma\in\Sigma$ is a unitary operator acting on $span\{\ket{q}\ket{\tau}\mid q\in Q,\tau\in \Gamma\}$ and two special entries $U_{\cent}$ and $U_{\dollar}$ corresponding to the endmarkers $\cent$ and $\dollar$ are unitary operators acting only on $span\{\ket{q}\mid q\in Q\}$. For every operator $U_{\sigma}$ with $\sigma\in\Sigma$,
since the input tape is {\em read-only}, although $U_{\sigma}$ accesses its second register holding advice symbols in $\Gamma$, it cannot change the ``content'' of the second register; namely, for any pair $(q,\tau)\in Q\times\Gamma$, there exists a unit-length vector $\ket{\phi_{q,\tau}}$ satisfying $U_{\sigma}\ket{q}\ket{\tau} = \ket{\phi_{q,\tau}}\ket{\tau}$.

Let us fix our input size $n$ (ignoring the endmarkers), let $i\in[n]$, and assume that $\sigma$ is the $i$th input symbol.
Each basic operator $U_{\sigma}$ induces an extended operator $U^{(i)}_{\sigma}$, acting on $E_n$, which
is, intuitively, applied to $M$'s inner states as well as the content of the the $i$th tape cell (composed of both an input symbol and an advice symbol).
Given any pair $(q,y)\in Q\times \Gamma^n$ with $y=y_1y_2\cdots y_n$, we set $U^{(i)}_{\sigma,n}\ket{q}\ket{y}$ to be a quantum state in $E_n$ obtained by applying $U_{\sigma}$ to $\ket{q}\ket{y_i}$ only.
As an example, it holds that $U^{(1)}_{\sigma,n}\ket{q}\ket{y} = (U_{\sigma}\ket{q}\ket{y_1})\otimes \ket{y_2y_3\cdots y_n}$. Moreover, let $U^{(0)}_{\cent,n}\ket{q}\ket{y} = (U_{\cent}\ket{q})\otimes \ket{y}$ and $U^{(n+1)}_{\dollar,n}\ket{q}\ket{y} = (U_{\dollar}\ket{q})\otimes \ket{y}$. All other extended operators $U^{(i)}_{\sigma,n}$ for any $i\in\nat$ and $\sigma\in\check{\Sigma}$ are simply set as $I$ (identity).
The original projection operators $P_{acc}$, $P_{rej}$, and $P_{non}$ in  Section \ref{sec:QFA/n} are appropriately  modified to act on $E_n$ as well.
Notice that those projection operators are applied only to the first register containing inner states in $Q$, not to the second register holding advice strings in $\Gamma^n$.

To simplify our notation, when ``$n$'' is clear from the context, we drop the script ``$n$'' and write $U^{(i)}_{\sigma}$ instead of $U^{(i)}_{\sigma,n}$. Hereafter, we fix $n\in\nat$.
For each extended operator $U^{(i)}_{\sigma}$, we set $T^{(i)}_{\sigma} = P_{non}U^{(i)}_{\sigma}$.
Given any input $\cent x\dollar = \sigma_0\sigma_1\sigma_2\cdots \sigma_{n}\sigma_{n+1}$ with $\sigma_0=\cent$, $\sigma_{n+1}=\dollar$, and $x\in\Sigma^n$ as well as any index $i\in[0,n+1]_{\integer}$,
an extended operator $T_{\sigma_0\sigma_1\cdots \sigma_i}$  acting on $E_n$ is defined to be $T_{\sigma_i}^{(i)}T_{\sigma_{i-1}}^{(i-1)} \cdots T_{\sigma_1}^{(1)} T_{\sigma_0}^{(0)}$.
On such an input $x$, an advised 1qfa $M$ starts with an initial quantum state $\qubit{q_0}\qubit{\phi_n} = \sum_{y\in\Gamma^n}\alpha_y\qubit{q_0}\qubit{y}$. At time $i+1$, performing the measurement $P_{acc}$ gives the acceptance probability $p_{acc}(x,\phi_n,i+1) = \|  P_{acc}U^{(i)}_{\sigma_i}T_{\sigma_0\sigma_1\cdots \sigma_{i-1}}\qubit{q_0}\qubit{\phi_n} \| ^2$, which equals
$\|  \sum_{y\in\Gamma^n}\alpha_y P_{acc}U^{(i)}_{\sigma_i}T_{\sigma_0\sigma_1\cdots \sigma_{i-1}} \qubit{q_0}\qubit{y} \| ^2$, provided that we ignore $T_{\sigma_0\sigma_1\cdots \sigma_{i-1}}$ when $i=0$.
After the 1qfa $M$ halts, the (total) acceptance probability
$p_{acc}(x,\phi_n)$ becomes  $\sum_{i=1}^{n+2}p_{acc}(x,\phi_n,i)$.  The rejection probabilities $p_{rej}(x,\phi_n,i+1)$ and $p_{rej}(x,\phi_n)$ are  similarly defined using $P_{rej}$ in place of $P_{acc}$.

Unlike a model of quantum Turing machine, our current model of 1qfa
is equipped with two {\em  read-only} tape tracks and, unintentionally, this ``read-only'' restriction severely limits the potential power of quantum advice. To comprehend this limitation, let us first observe that each basis advice string given in a quantum advice state is unaltered during computation,  and therefore any two quantum computations associated with different basis advice strings never interfere with each another. This observation yields
the following new characterization of $\oneqfa/Rn$ in terms of quantum advice.
For succinctness, we use the notation $\prob_{M}[M(\track{x}{\phi_{|x|}})=L(x)]$ to denote the total probability of $M$ on input $\qubit{\track{x}{\phi_n}}$ producing output value $L(x)$.

\begin{proposition}\label{random-quantum-advice}
Let $L$ be any  language over alphabet $\Sigma$. The following two statements are logically equivalent.
\begin{enumerate}\vs{-1}
  \setlength{\topsep}{0mm}%
  \setlength{\itemsep}{0mm}
  \setlength{\parskip}{0cm}%

\item  $L\in\oneqfa/Rn$.
\item \sloppy There exist a 1qfa $M$ with two read-only tape tracks, an advice alphabet $\Gamma$, a series $\Phi = \{\qubit{\phi_n}\}_{n\in\nat}$ of quantum advice states over $\Gamma^*$, and an error bound  $\varepsilon\in[0,1/2)$ satisfying $\prob_{M}[M(\track{x}{\phi_{|x|}}) = L(x)]\geq 1-\varepsilon$ for any input $x\in\Sigma^*$.
\end{enumerate}
\end{proposition}

\begin{proof}
($1 \Rightarrow 2$) Note that a piece of randomized advice, say, $D_n$ over $\Gamma^n$ can be embedded into the aforementioned Hilbert space $E_n$ as a quantum state of the form $\qubit{\phi_n} = \sum_{y\in\Gamma^n}\sqrt{D_n(y)}\qubit{y}$.
Statement (2) thus follows immediately by replacing $D_n$ with $\qubit{\phi_n}$.

($2 \Rightarrow 1$)
Take $M$, $\Gamma$, $\Phi$, and $\varepsilon$ described in the lemma and take any number $n\in\nat$.
Assume that each advice quantum state $\qubit{\phi_n}\in\Phi$ is of the form  $\qubit{\phi_n} = \sum_{y\in\Gamma^n}\alpha_{y}\qubit{y}$ for appropriate amplitudes $\alpha_i\in\complex$.  Let $x=\sigma_1\sigma_2\cdots\sigma_n$ be any string in $\Sigma^n$ and choose an appropriate string  $y=\tau_1\tau_2\cdots \tau_n$ of length $n$.
For convenience, write $\sigma_0=\cent$ and $\sigma_{n+1}=\dollar$.
For the pair $(x,y)$,  we define $\qubit{\phi_{0}^{(x,y)}} = \qubit{q_0}\qubit{y}$ and $\qubit{\phi_{i+1}^{(x,y)}} = T^{(i)}_{\sigma_i}\qubit{\phi^{(x,y)}_{i}}$ for each index $i\in[0,n+1]_{\integer}$.
As remarked earlier, $T^{(i)}_{\sigma_i}$ modifies only $M$'s inner states; thus, $\qubit{\phi^{(x,y)}_{i+1}}$ can be expressed as $\qubit{\psi^{(x,y)}_{i+1}}\qubit{y}$ for a certain quantum state $\qubit{\psi^{(x,y)}_{i+1}}$.
Similarly, $U^{(i)}_{\sigma_i}\qubit{\phi^{(x,y)}_{i}}$ can be expressed as $\qubit{\tilde{\psi}^{(x,y)}_{i+1}}\qubit{y}$. The total acceptance probability $p_{acc}(x,\phi_n)$ is then calculated as
\[
p_{acc}(x,\phi_n) = \left\|  \sum_{y\in\Gamma^n}|\alpha_y|^2
\sum_{i=0}^{n+1} P_{acc}U^{(i)}_{\sigma_{i}}\qubit{\phi_{i}^{(x,y)}} \right\| ^2
= \sum_{y\in\Gamma^n}|\alpha_y|^2 \left\|
\sum_{i=0}^{n+1} P_{acc}\qubit{\tilde{\psi}_{i+1}^{(x,y)}}\qubit{y} \right\| ^2.
\]
The rejection probability $p_{rej}(x,\phi_n)$ is also calculated similarly by replacing $P_{acc}$ with $P_{rej}$. To obtain the desired consequence, it suffices to take an advice probability ensemble $\{D_n\}_{n\in\nat}$ defined as $D_n(y) = |\alpha_{y}|^2$ for each string $y\in\Gamma^n$.
\end{proof}

Proposition  \ref{random-quantum-advice} indicates that, if a 1qfa has only read-only tape tracks, then the power of quantum advice is diminished to that of randomized advice.
The proposition therefore guides us to an inevitable introduction of a notion of ``rewritable'' advice tracks in the next subsection.

\subsection{Rewritable Advised Quantum Finite Automata for Quantum Advice}\label{sec:rewritable-1qfa}

We begin with a brief discussion on how to extend the original advised 1qfa model (given in Section \ref{sec:basic-property-QFA/n}) in a simple and natural way.  First of all, we remind that, for most types of classical ``one-way'' finite automata, it is of no importance whether a tape head erases or modifies the content of any tape cell just before leaving off that tape cell, because the tape head never returns to this particular tape cell to retrieve any modified information.
Even if  the tape head is allowed to return to the modified tape cells, the computational power of the automata may not change in many natural cases. For instance, as noted earlier, advised 1dfa's (resp., bounded-error advised 1pfa's) are computationally equivalent to one-tape linear-time deterministic (resp., bounded-error probabilistic) Turing machines with linear-size advice; in short, both $\onedlin/lin = \reg/n$ \cite{TYL10} and $\onebplin/Rlin = \reg/Rn$ \cite{Yam10} hold.
These equalities suggest that the ``read-only'' requirement for  input tapes
is irrelevant to the computational power of 1dfa's and 1pfa's.
Now, suppose that we re-define two automata models---1qfa's and 1rfa's---used in the previous sections for deterministic and randomized advice so that they are further allowed to modify any {\em advice symbol} written in each tape cell of the lower tape track before their tape heads leave this current tape cell (but, importantly, the tape heads  never visit the same tape cell again).  For our reference, such a tape track is referred to as a {\em rewritable advice (tape)  track}. It is not difficult to see that such a new definition does not alter the advised language families, such as  $\reg/n$ and $\reg/Rn$, simply because underlying automata cannot remember more than a constant number of modified symbols.

When quantum advice is concerned, what would happen if we use 1qfa's equipped with rewritable advice tracks?  For our convenience, we refer to such an extended 1qfa as a {\em rewritable advised 1qfa}.
Nonetheless, we keep  an upper track that holds standard input strings
{\em read-only} as in the original model of 1qfa's.
Notice that what actually limits the power of 1qfa's is a prohibition of disposing of (or dumping) quantum information after it is read and its information is processed. In other words, we intend to utilize the advice  track as a device of {\em write-only memory}. Unlike classical computation, quantum computation can draw a considerable benefit from such write-only memory, despite the fact that a one-way head move still  hampers the
machine's potential ability. Similar ideas were discussed lately in, e.g., \cite{Pas00,YFS+12}.

Let us recall the rephrased description of advised 1qfa's presented in Section \ref{sec:complexity-Rn}. Using the same notations, let $\Phi=\{\ket{\phi_n}\}_{n\in\nat}$ be a series of quantum advice states over $\Gamma^*$. With this quantum advice $\Phi$,
a rewritable advised 1qfa $M = (Q,\Sigma, \Gamma, \{U_{\sigma}\}_{\sigma\in \check{\Sigma}}, q_0,Q_{acc},Q_{rej})$ starts with an initial quantum state  $\qubit{q_0}\qubit{\phi_n}$, where $\qubit{\phi_n}$ is an advice quantum  state in $span\{\qubit{z}\mid z\in \Gamma^n\}$ when a string $x$ of length $n$ is given as a standard input.
The machine's unary operator $U^{(i)}_{\sigma_i}$ is still applied to only $M$'s inner states and the content of the $i$th tape cell; however, it now freely modifies the content of the advice track. The acceptance probability $p_{acc}(x,\phi_n)$ of $M$ on $x$ with quantum advice $\qubit{\phi_n}$ is the sum, over all $i\in[0,n+1]_{\integer}$, of the values  $p_{acc}(x,\phi_n,i+1) = \|  P_{acc}U_{\sigma_i}^{(i)}T_{\sigma_0\sigma_1\cdots \sigma_{i-1}} \qubit{q_0}\qubit{\phi_n}  \| ^2$. The rejection probability $p_{rej}(x,\phi_n)$  is similarly defined.
To emphasize the use of rewritable advised tape tracks, a special notation $\oneqfastar/Qn$ will be used to  denote the family of all languages recognized with bounded-error probability by rewritable advised 1qfa's using quantum advice.

The actual power of rewritable advised 1qfa's geared by quantum advice is exemplified in  Lemma \ref{inclusion-oneqfa/Qn}. For this lemma, let us first review the language family $\onebqlin$, which was introduced in \cite{TYL10} as the family of all languages recognized by one-tape one-head two-way off-line quantum Turing machines whose error probabilities are upper-bounded by $1/4$, where all the ``classically-viewed'' computation paths generated by the machines must terminate {\em simultaneously} within a {\em linear} number of steps. Appending linear-size quantum advice to those machines, we naturally expand $\onebqlin$ to its advised version $\onebqlin/Qlin$, which can be also seen as a quantum analogue of $\onebplin/Rlin$ \cite{Yam10}.
Because of a nature of Turing machine, during its computation, the machine can freely alter not only a given advice string but also a given input string.


The next lemma specifies a location of $\oneqfastar/Qn$ in the landscape of low-complexity classes.

\begin{lemma}\label{inclusion-oneqfa/Qn}
$\reg/Rn \subseteq \oneqfastar/Qn \subseteq \onebqlin/Qlin$.
\end{lemma}

\begin{proof}
The second containment $\oneqfastar/Qn \subseteq \onebqlin/Qlin$ is obvious from the fact that any 1qfa can be viewed as a special case of one-tape quantum Turing machine because it naturally satisfies all the requirements needed to be a rewritable advised  1qfa.

The first containment $\reg/Rn \subseteq \oneqfastar/Qn$ is shown, roughly, with a similar idea used in, \eg  \cite[Proposition 4.2]{NY09}, by dumping the information regarding a current inner state of an underlying 1dfa onto a rewritable  advice track in order to convert a deterministic move into a quantum move.

More precisely, take any language $S$ in $\reg/Rn$. There are
a 1dfa $M$, an advice alphabet $\Gamma$, an advice probability ensemble $\{D_n\}_{n\in\nat}$, an error bound $\varepsilon\in[0,1/2)$ satisfying  $\prob_{M,D_n}[M(\track{x}{D_n}) = S(x)]\geq 1-\varepsilon$ for every length $n\in\nat$ and every string $x\in\Sigma^n$.
Let us  construct a rewritable advised 1qfa $N$. In an arbitrary configuration, assume that $M$ is in inner state $q$, is scanning $\track{\sigma}{\tau}$, and applies a transition $\delta(q,\track{\sigma}{\tau}) = q'$. Corresponding to this particular move,  $N$ scans $\track{\sigma}{\tau}$ in the inner state $q$,   modifies  $\track{\sigma}{\tau}$ to $\track{\sigma}{\tau_{q}}$, and enters the inner state $q'$, where $\tau_q = \track{q}{\tau}$ is a new advice symbol uniquely associated with $(q,\sigma)$.
It thus holds that, for each fixed advice string $y\in\Gamma^n$, $M$ accepts $\track{x}{y}$ iff $N$ on input $\track{x}{y}$ enters an accepting state with probability $1$.
Choose a quantum state $\qubit{\phi_n} = \sum_{y\in\Gamma^n}\sqrt{D_n(y)}\qubit{y}$ as our intended quantum advice.
It then follows that $\prob_{N}[N(\track{x}{\phi_n}) = S(x)] = \prob_{M,D_n}[M(\track{x}{D_n}) = S(x)]$. Therefore, $S$ is in $\oneqfastar/Qn$.
\end{proof}


An introduction of rewritable advice track also makes it possible to prove
a {\em closure property} of $\oneqfastar/Qn$ under Boolean operations.
By  contrast, with no advice, $\oneqfa$ violates those properties \cite{AKV01}, chiefly because a 1qfa alone is, in general, unable to amplify its success probability.

\begin{proposition}\label{closure-of-1QFA/Qn}
The advised language family $\oneqfastar/Qn$ is closed under union, intersection, and complementation.
\end{proposition}


Three closure properties of $\oneqfastar/Qn$ given in Proposition \ref{closure-of-1QFA/Qn} are a direct consequence of the facts shown in Lemmas  \ref{measure-reduction} and \ref{error-prob-reduction} that, by an appropriate use of quantum advice, (i) a
rewritable advised 1qfa can reduce the number of applications of measurement operations down to one and (ii) the rewritable advised 1qfa can reduce its error probability as well.

As shown in the next lemma, the use of rewritable advice track helps postpone all projective measurement operations until the very end of their computation and, consequently, it significantly simplifies the behaviors of 1qfa's.

\begin{lemma}\label{measure-reduction}
For any rewritable advised 1qfa $M$ with quantum advice  $\Psi = \{\qubit{\phi_{n}}\}_{n\in\nat}$, there exist another quantum advice $\Psi' = \{\qubit{\phi'_n}\}_{n\in\nat}$ and another rewritable advised 1qfa $N$ whose input tape has no endmarker  such that, for all nonempty input strings, (i) $N$ conducts a projective measurement only once just after scanning an entire input and (ii) after the measurement, the acceptance probability of $N$ on each input with $\Psi'$ equals the acceptance probability of $M$ on the same input with $\Psi$. In order to process the empty input $\lambda$, by adding the right endmarker $\dollar$, we can modify $N$ to conduct a measurement only once just after scanning $\dollar$ and to accept (or reject)
$\lambda$ with certainty.
\end{lemma}

With the help of quantum advice, the error bound of each rewritable advised 1qfa can be significantly reduced.  This error-reduction property is quite useful in constructing advised  1qfa's that recognize given target languages.

\begin{lemma}\label{error-prob-reduction}
Let $L$ be any language  in $\oneqfastar/Qn$ over alphabet $\Sigma$.
For any constant $\varepsilon\in(0,1/2)$,
there exist a rewritable advised 1qfa $M$ and a series $\{\qubit{\phi_n}\}_{n\in\nat}$ of quantum advice states such that, for every length $n\in\nat$, (i) for any string $x\in L\cap\Sigma^n$, $M$ accepts $\qubit{\track{x}{\phi_n}}$ with probability at least $1-\varepsilon$, and (iii) for any string $x\in\Sigma^n-L$, $M$ rejects  $\qubit{\track{x}{\phi_n}}$ with probability at least $1-\varepsilon$.
\end{lemma}


Before proving Lemmas \ref{measure-reduction} and \ref{error-prob-reduction}, we wish to finish the proof of Proposition \ref{closure-of-1QFA/Qn}.

\begin{proofof}{Proposition \ref{closure-of-1QFA/Qn}}
Our goal is to show three closure properties of $\oneqfastar/Qn$ listed in the proposition.
Let $L_1$ and $L_2$ be two arbitrary languages in $\oneqfastar/Qn$. For each index $i\in\{1,2\}$, by Lemmas \ref{measure-reduction}, there exists
a rewritable advised 1qfa $M_i$ with no left  endmarker $\cent$
such that $M_i$ recognizes $L_i$ with quantum advice $\Phi_i=\{\qubit{\phi_{i,n}}\}_{n\in\nat}$ over advice alphabet $\Gamma_i$ with error bound $\varepsilon_i\in [0,1/2)$  and that $M_i$ performs no measurement until scanning the right endmarker $\dollar$.
Let $M_i = (Q_i,\Sigma, \Gamma_i, \{U_{i,\sigma}\}_{\sigma\in\Sigma_{\dollar}}, q_{i,0},Q_{i,acc},Q_{i,rej})$, where $\Sigma_{\dollar} = \Sigma\cup\{\dollar\}$.
Without loss of generality, we further assume that $\Gamma_1=\Gamma_2$ and simply write $\Gamma$ for $\Gamma_1$.
For convenience, let  $\qubit{\phi_{i,n}}$ express a quantum state  $\sum_{y\in\Gamma}\alpha^{(i)}_{y}\qubit{y}$ for certain amplitudes $\{\alpha^{(i)}_{y}\}_{y\in\Gamma} \subseteq \complex$ with $\sum_{y\in\Gamma}|\alpha^{(i)}_{y}|^2=1$.

\s
\n{\bf [Complementation]}
Consider the complement $\overline{L_1}$ of $L_1$.  We modify $M_1$ by exchanging the roles of ``accepting states'' and ``rejecting states'' in $Q$. Since $M_1$ recognizes $L_1$ using $\Phi_1$ with bounded error probability, it is obvious that this new machine recognizes $\overline{L_1}$ using $\Phi_1$ with the same errors probability as $M_1$'s.

\s
\n{\bf [Intersection]}
By Lemma \ref{error-prob-reduction}, it is possible to reduce the error probability of $M_i$; thus, we can assume that $0\leq \varepsilon_i<1-\frac{\sqrt{2}}{2}$. For convenience, we set $\varepsilon = \varepsilon_1+\varepsilon_2-\varepsilon_1\varepsilon_2$. By the choice of
$\varepsilon_i$'s, it follows that $0\leq \varepsilon<1/2$.

For $L_1\cap L_2$, let us define a new rewritable advised 1qfa $M = (Q,\Sigma, \Gamma, \{V_{\sigma}\}_{\sigma\in \Sigma_{\dollar}}, q_0,Q_{acc},Q_{rej})$ having no left endmarker. Let  $Q=Q_1\times Q_2$, $Q_{acc} = Q_{1,acc}\times Q_{2,acc}$, $Q_{rej} = (Q_{1,rej}\times Q) \cup (Q\times Q_{2,rej})$, and $q_0 =(q_{1,0},q_{2,0})$.
Each unitary operator $V_{\sigma}$ is defined as  $V_{\sigma}\qubit{(q_1,q_2)}\qubit{\tau_1,\tau_2} = U_{1,\sigma}\qubit{q_1}\qubit{\tau_1}\otimes  U_{2,\sigma}\qubit{q_2}\qubit{\tau_2}$. Fix $n\in\nat$.  For each pair $(i,\sigma)\in [2]\times\Sigma_{\dollar}$ and any $j\in\nat^{+}$, we write $U^{(j)}_{i,\sigma,n}$ for an extended operator induced from $U_{i,\sigma}$.
Given any input $x\in\Sigma^n$, let $U_{i,x}$ denote $U^{(n+1)}_{i,\dollar,n} U^{(n)}_{i,x_n,n}U^{(n-1)}_{i,x_{n-1},n}\cdots U^{(1)}_{i,x_1,n}$. Likewise,
we take extended operators $V^{(j)}_{\sigma,n}$ and we set $V_{x} = V^{(n+1)}_{\dollar,n}V^{(n)}_{x_n,n}V^{(n-1)}_{x_{n-1},n}\cdots V^{(1)}_{x_1,n}$.
Now, our new quantum advice state $\qubit{\psi_n}$ has the form $\qubit{\phi_{1,n}}\otimes \qubit{\phi_{2,n}}$.
When $M$ reads the entire input string $x$, $M$ generates  a quantum state $V_{x}\qubit{q_{0}}\qubit{\psi_n} = U_{1,x}\qubit{q_{1,0}}\qubit{\phi_{1,n}}\otimes U_{2,x}\qubit{q_{2,0}}\qubit{\phi_{2,n}}$.
Because $M$'s computation is, in essence, decomposed into two independent computations of $M_1$ and $M_2$, it is easy to show that
\begin{equation*}\label{prob-multiplication}
\prob_{M}[M(\track{x}{\psi_n}) = 1] = \prob_{M_1}[M_1(\track{x}{\phi_{1,n}}) =1]\cdot \prob_{M_2}[M_2(\track{x}{\phi_{2,n}}) =1].
\end{equation*}
{}From this equality, we obtain the following.
\begin{itemize}\vs{-1}
  \setlength{\topsep}{0mm}%
  \setlength{\itemsep}{0mm}
  \setlength{\parskip}{0cm}%

\item If $x\in L$, then it holds that $\prob_{M}[M(\track{x}{\psi_n}) =1]\geq (1-\varepsilon_1)(1-\varepsilon_2) = 1- \varepsilon$.

\item If $x\not\in L$, then it holds that $\prob_{M}[M(\track{x}{\psi_n}) =0]\geq \max\{ 1-\varepsilon_1,1-\varepsilon_2\} \geq 1- \varepsilon$, because $\varepsilon$ is at least $\max\{\varepsilon_1,\varepsilon_2\}$.
\end{itemize}\vs{-1}
Therefore, $M$ recognizes $L_1\cap L_2$  with bounded-error probability using the quantum advice $\{\qubit{\psi_n}\}_{n\in\nat}$.

\s
\n{\bf [Union]}
Since $L_1\cup L_2 = \overline{\overline{L_1}\cap \overline{L_2}}$, this ``union'' case follows from the previous cases of ``complementation'' and ``intersection.''
\end{proofof}

Now, let us prove  Lemmas \ref{measure-reduction} and \ref{error-prob-reduction}.
Lemma \ref{measure-reduction} is shown intuitively as follows. Instead of measuring advised 1qfa's inner states at every step, we write them down on an advice track and enter new (but corresponding) non-halting states so that we can keep the 1qfa operating without performing any measurement until we make the last-minute measurement at the very end of a computation of the 1qfa.

\begin{proofof}{Lemma \ref{measure-reduction}}
Let $\Sigma$ and $\Gamma$ denote respectively an input alphabet and an advice alphabet. Let $M$ be any rewritable advised 1qfa and let  $\Phi =\{\qubit{\phi_n}\}_{n\in\nat}$ be a series of advice quantum states over  $\Gamma^*$. We remark that Lemma \ref{endmarker} holds also for $M$. Therefore,  we can assume that $M$'s input tape has no endmarker.
This assumption helps us set
$M$ to be $(Q,\Sigma,\Gamma,\{U_{\sigma}\}_{\sigma\in{\Sigma}}, q_0,Q_{acc},Q_{rej})$.
For the first part of the lemma, since we deal only with nonempty input strings, hereafter, we assume that $n\geq1$.
In addition, we use the notation $p_{acc}(x,\phi_n,i)$ to express
the acceptance probability of $M$ on input $x=x_1x_2\cdots x_n$  at time $i\in\nat^{+}$; namely, $p_{acc}(x,\phi_n,i) = \|   P_{acc}U^{(i)}_{x_i}T_{x_1x_2\cdots x_{i-1}}\qubit{q_0}\qubit{\phi_n}\| ^2$, provided that $T_{x_1x_2\cdots x_{i-1}}$ is ignored when $i=1$. The total acceptance probability $p_{acc}(x,\phi_n)$ of $M$ on $x$ then becomes $\sum_{i=1}^{n}p_{acc}(x,\phi_n,i)$.

By modifying $M$ properly, we wish to define a new rewritable advised 1qfa $N = (\tilde{Q},\Sigma, \tilde{\Gamma}, \{\hat{U}_{\sigma}\}_{\sigma\in{\Sigma}}, q_0,Q_{acc},Q_{rej})$ for which $N$ has no endmarker and conducts a projective measurement only once just after reading the entire nonempty input.
To each halting state $q\in Q_{halt}$, we assign a new {\em non-halting}  state $\hat{q}$, and we then define $\hat{Q}_{halt} = \{\hat{q}\mid q\in Q_{halt}\}$ and $\Gamma' = \{\track{\hat{q}}{\tau}\mid q\in Q_{halt},\tau\in\Gamma\}$.
Associated with each $q\in Q$, we also prepare a new inner state $\bar{q}$, and we set $\bar{Q} = \{\bar{q}\mid q\in Q\}$.
The desired set $\tilde{Q}$ of $N$'s inner states is defined as $Q\cup \hat{Q}_{halt} \cup\bar{Q}$. Moreover, we set $\bar{Q}_{acc} = \{\bar{q}\mid q\in Q_{acc}\}$  and $\bar{Q}_{rej} = \{\bar{q}\mid q\in Q_{rej}\}$.
Our new advice alphabet  $\tilde{\Gamma}$ is $\Gamma\cup \Gamma'  \cup\{ \track{\dollar}{\tau}\mid \tau\in\Gamma\}$ and our  new quantum advice $\Psi'$ consists of quantum states $\qubit{\phi^{\dollar}_n} = \sum_{y\in\Gamma^n}\gamma_{y}\qubit{y_1y_2\cdots y_{n-1}\track{\dollar}{y_n}}$ induced from  $\qubit{\phi_n} = \sum_{y\in\Gamma^n}\gamma_{y}\qubit{y}$,
provided that each $y$ has the form $y=y_1y_2\cdots y_n$. Each operator   $\hat{U}_{\sigma}$ of $N$ is defined as follows. Given any $q\in \tilde{Q}$ and any $\tau\in\Gamma$, let
(i) $\hat{U}_{\sigma}\ket{q}\ket{\tau} = U_{\sigma}\ket{q}\ket{\tau}$ if $q\in Q_{non}$,
(ii) $\hat{U}_{\sigma}\ket{\hat{q}}\ket{\tau} = \ket{\hat{q}}\ket{\track{\hat{q}}{\tau}}$ if $q\in Q_{halt}$,
(iii) $\hat{U}_{\sigma}\ket{q}\ket{\track{\dollar}{\tau}} = VU_{\sigma}\ket{q}\ket{\tau}$ if $q\in Q$, and
(iv)  $\hat{U}_{\sigma}\ket{\hat{q}}\ket{\track{\dollar}{\tau}} = \ket{\bar{q}}\ket{\track{\hat{q}}{\tau}}$ if $q\in Q_{halt}$, where $V$ behaves as $V\ket{p}\ket{\nu} = \ket{\bar{p}}\ket{\track{p}{\nu}}$. For the other cases, we can set the value of $\hat{U}_{\sigma}\ket{q}\ket{\tau}$ arbitrarily as long as $\hat{U}_{\sigma}$ maintains unitarity. We also expand three projections $P_{acc},P_{rej},P_{non}$  to $\tilde{P}_{acc},\tilde{P}_{rej},\tilde{P}_{non}$, respectively.

For brevity, let $H'_n$ denote the space $span\{ \ket{q}\ket{y} \mid q\in Q, y\in\tilde{\Gamma}^n \}$. Similarly to $\qubit{\phi^{\dollar}_{n}}$, if  $\qubit{\psi}$ has the form  $\sum_{q\in Q}\sum_{y\in\Gamma^n}\alpha_{q,y}\qubit{q}\qubit{y}$ with $y=y_1y_2\cdots y_{n-1}y_n$, then  we denote by $\qubit{\psi^{\dollar}}$ the quantum state $\sum_{q\in Q}\sum_{y\in\Gamma^n}\alpha_{q,y}\qubit{q}\qubit{y_1y_2\cdots y_{n-1}\track{\dollar}{y_n}}$. Obviously, $\qubit{\psi^{\dollar}}$
belongs to $H'_n$.
Let us look into a computation of $M$ and its associated computation of $N$.
Initially, the machine $M$ with the given quantum advice $\Phi$ starts with a quantum state $\qubit{\psi_0} = \qubit{q_0}\qubit{\phi_n}$, whereas $N$ with $\Phi'$ begins with  $\qubit{\psi'_0}=\qubit{q_0}\qubit{\phi^{\dollar}_n}$, which can be written as $\qubit{\psi^{\dollar}_0} + \qubit{\xi_0}$ with $\qubit{\xi_0}=0$.
Let $i$ be any index between $1$ and $n-1$.
Now, assume that, just after step $i-1$, $M$ generates a quantum state $\qubit{\psi_{i-1}}\in E_{non}$ and $N$ generates $\qubit{\psi'_{i-1}} = \qubit{\psi^{\dollar}_{i-1}} + \qubit{\xi_{i-1}}$, where $\qubit{\xi_{i-1}}$ belongs to the space $E'_n  = span\{\qubit{q}\qubit{y}\mid q\in \hat{Q}_{halt},y\in\tilde{\Gamma}^n\}$.

Now, let us consider the $i$th step.
Before performing a measurement, $M$ is assumed to have generated a quantum state  $U^{(i)}_{x_i}\qubit{\psi_{i-1}} =  \qubit{\psi_{i}} + \qubit{\psi_{i,acc}} + \qubit{\psi_{i,rej}}\in E_{non}\oplus E_{acc}\oplus E_{rej}$. After an application of the measurement, the acceptance (resp., rejection) probability $p_{acc}(x,\phi_n,i)$ (resp., $p_{rej}(x,\phi_n,i)$) becomes $\| \qubit{\psi_{i,acc}}\| ^2$
(resp., $\| \qubit{\psi_{i,rej}}\| ^2$).
Corresponding to  $\qubit{\psi_{i,acc}}$,
we set $\qubit{\psi'_{i,acc}}$ to be $\sum_{q\in Q_{acc}}\sum_{y\in\Gamma^n}\alpha_{q,y}\qubit{\hat{q}}\qubit{y_1\cdots y_{i-1}\track{q}{y_i}y_{i+1}\cdots y_{n-1}\track{\dollar}{y_n}}$ if $\qubit{\psi_{i,acc}}$ is expressed as $\sum_{q\in Q_{acc}}\sum_{y\in\Gamma^n}\alpha_{q,y}\qubit{q}\qubit{y}$. Likewise, $\qubit{\psi'_{i,rej}}$ is defined using $Q_{rej}$ in place of $Q_{acc}$.
Note that, for each $\hat{q}\in \hat{Q}_{halt}$,  $\hat{U}^{(i)}_{x_i}\qubit{\hat{q}}\qubit{y'_1\cdots y'_{i-1}y_iy_{i+1}\cdots  y_n}$ equals $\qubit{\hat{q}}\qubit{y'_1\cdots y'_{i-1}\track{\hat{q}}{y_i}y_{i+1}\cdots y_n}$, where $y'_1,\ldots,y'_{i-1}\in\Gamma'$ and $y_i,y_{i+1},\ldots,y_n\in\Gamma$.
It thus follows that $\hat{U}^{(i)}_{x_i}\qubit{\psi^{\dollar}_{i-1}} = \qubit{\psi^{\dollar}_i} + \qubit{\psi'_{i,acc}} + \qubit{\psi'_{i,rej}}$.
The machine $N$ therefore generates
\[
\qubit{\psi'_i} = \hat{U}^{(i)}_{x_i}\qubit{\psi'_{i-1}} = \hat{U}^{(i)}_{x_i}\qubit{\psi^{\dollar}_{i-1}} + \hat{U}^{(i)}_{x_i}\qubit{\xi_{i-1}} = \qubit{\psi^{\dollar}_i} + \qubit{\psi'_{i,acc}} + \qubit{\psi'_{i,rej}} + \hat{U}^{(i)}_{x_i}\qubit{\xi_{i-1}}.
\]
Finally, we set $\qubit{\xi_{i}}$ to be $\qubit{\psi'_{i,acc}} + \qubit{\psi'_{i,rej}} + \hat{U}^{(i)}_{x_i}\qubit{\xi_{i-1}}$,  which belongs to  $E'_n$. Hence, we obtain $\qubit{\psi'_{i}} = \qubit{\psi^{\dollar}_{i}} + \qubit{\xi_{i}}$.
Since $\| \qubit{\psi_{i,acc}}\|  = \| \qubit{\psi'_{i,acc}}\|$, we derive   $p_{acc}(x,\phi_n,i) = \| \qubit{\psi'_{i,acc}}\| ^2$. It is important to note that every vector $\qubit{\xi_i}$ is orthogonal to $\qubit{\xi_{i-1}}$ because $N$ generates different strings on its rewritable advice track
at time $i$.

At step $n$, by the help of the advice symbol $\track{\dollar}{y_n}$,
if $U^{(n)}_{x_n}\qubit{\psi_{n-1}}$ is expressed as  $\sum_{q\in Q}\sum_{y\in\Gamma^n} \alpha_{q,y}\qubit{q}\qubit{y}$, then
$\hat{U}^{(n)}_{x_n}\qubit{\psi^{\dollar}_{n-1}}$ must have the form
$\sum_{q\in Q}\sum_{y\in\Gamma^n} \alpha_{q,y} \qubit{\bar{q}}\qubit{y_1\cdots y_{n-1}\track{q}{y_n}}$, which is orthogonal to the space $H'_n$.
Concerning an inner state $\hat{q}\in\hat{Q}_{halt}$, it also holds that $\hat{U}^{(n)}_{x_n}\qubit{\hat{q}}\qubit{y'_1\cdots y'_{n-1}\track{\dollar}{y_n}} = \qubit{\bar{q}}\qubit{y'_1,\cdots y'_{n-1}\track{\hat{q}}{y_n}}$  for all  $y'_1,\ldots,y'_{n-1}\in\Gamma'$ and $y_n\in\Gamma$.
Thus, the acceptance probability $p_{acc}(x,\phi_n,n)$ of $M$ at step $n$, which equals $\| P_{acc}U^{(n)}_{x_n}\qubit{\psi_{n-1}} \| ^2$,
coincides with
$\| \tilde{P}_{acc}\hat{U}^{(n)}_{x_n}\qubit{\psi^{\dollar}_{n-1}} \| ^2$ by the definition of $\bar{Q}_{acc}$.
For the final quantum state $\ket{\psi'_n} = \hat{U}^{(n)}_{x_n}\qubit{\psi'_{n-1}}$ of $N$, we obtain
\[
\tilde{P}_{acc}\qubit{\psi'_{n}} =
\tilde{P}_{acc}\hat{U}^{(n)}_{x_n}\qubit{\psi'_{n-1}} =   \tilde{P}_{acc}\hat{U}^{(n)}_{x_n}\qubit{\psi^{\dollar}_{n-1}} + \sum_{i=1}^{n-1}\tilde{P}_{acc}\hat{U}^{(n)}_{x_n}\cdots \hat{U}^{(i+1)}_{x_{i+1}}\qubit{\psi'_{i,acc}}.
\]
After performing a measurement,   $N$ accepts $x$ with probability $p= \| \tilde{P}_{acc}\qubit{\psi'_{n}} \|^2$, which equals $\sum_{i=1}^{n-1}\| \qubit{\psi'_{i,acc}}\| ^2 + \|  \tilde{P}_{acc}\hat{U}^{(n)}_{x_n}\qubit{\psi^{\dollar}_{n-1}} \| ^2$ because $\| \tilde{P}_{acc}\hat{U}^{(n)}_{x_n}\cdots \hat{U}^{(i+1)}_{x_{i+1}}\qubit{\psi'_{i,acc}} \| = \| \qubit{\psi'_{i,acc}} \|$ holds. Since
$\|\tilde{P}_{acc}\hat{U}^{(n)}_{x_n}\qubit{\psi^{\dollar}_{n-1}} \|^2 = \|\tilde{P}_{acc}U^{(n)}_{x_n}\qubit{\psi_{n-1}} \|^2 = p_{acc}(x,\phi_n,n)$ and $\|\qubit{\psi'_{i,acc}}\|^2 = p_{acc}(x,\phi_n,i)$ for all $i\in[n-1]$, we conclude that $p$ equals $p_{acc}(x,\phi_n)$ of $M$.

We want to show the second part of the lemma. To cope with the empty input, we further modify $N$ by equipping $N$ with the right endmarker $\dollar$ and by forcing $N$ to perform a measurement once after applying an extended operator $\hat{U}^{(n+1)}_{\dollar}$ associated with $\dollar$. We write $N'$ to indicate this modified machine.
Let us consider only the case where we wish to ``accept'' $\lambda$ with certainty, since the other case can be similarly dealt with.
As in the proof of Lemma \ref{endmarker}, we prepare a fresh inner state $q_f$ and set $\tilde{Q}'=\tilde{Q}\cup\{q_f\}$. We also define $\bar{Q}'_{acc} = \bar{Q}_{acc}\cup\{q_f\}$ and $\bar{Q}'_{rej} = \bar{Q}_{rej}$ and prepare new projections $\tilde{P}'_{acc}$ and $\tilde{P}'_{rej}$ associated with  $\bar{Q}'_{acc}$ and $\bar{Q}'_{rej}$, respectively.
Additionally, we define $\hat{U}_{\dollar}$, which acts on $E_{\tilde{Q}'}$, as follows: let  $\hat{U}_{\dollar}\ket{q_0} = \ket{q_f}$,  $\hat{U}_{\dollar}\ket{q_f} = \ket{q_0}$ and $\hat{U}_{\dollar}\ket{q} = \ket{q}$ for all inner states $q\in \tilde{Q}'-\{q_0,q_f\}$.

First, assume that $n\geq1$. Note that $N'$ accepts input $x$ with probability $p'=\| \tilde{P}'_{acc}\hat{U}^{(n+1)}_{\dollar}\qubit{\psi'_n} \|^2$.
Notice that $\ket{\psi'_n}$ is orthogonal to the space $span\{\ket{q}\ket{y}\mid q\in Q\cup\{q_f\}, y\in\tilde{\Gamma}^n\}$.
Thus, an application of $\hat{U}^{(n+1)}_{\dollar}$ does not change  the vector $\ket{\psi'_n}$, namely, $\hat{U}^{(n+1)}_{\dollar}\qubit{\psi'_n} = \qubit{\psi'_n}$. These facts imply that $p' = \| \tilde{P}'_{acc} \ket{\psi'_{n}} \|^2 = \| \tilde{P}_{acc}\ket{\psi'_n} \|^2$; as a result, we obtain $p'=p$.
In contrast, when $n=0$ (\ie the input $x$ is $\lambda$), since $p' = \| \tilde{P}'_{acc}\hat{U}^{(n+1)}_{\dollar}\ket{q_0}\ket{\lambda} \|^2 = \|\tilde{P}'_{acc}\ket{q_f}\ket{\lambda} \|^2$,
$N'$ accepts $\lambda$ with certainty, as requested by the lemma.
\end{proofof}


Hereafter, we will prove Lemma \ref{error-prob-reduction}. The proof of the lemma is based on a technique of {\em parallel repetition} of the same quantum computation, and Lemma \ref{measure-reduction} essentially helps make this  technique applicable.

\begin{proofof}{Lemma \ref{error-prob-reduction}}
Since $L\in\oneqfastar/Qn$, by Lemma \ref{measure-reduction}, we take a rewritable advised 1qfa $M$ with no left endmarker $\cent$, an error bound $\varepsilon_0\in[0,1/2)$,   and a series $\Psi =\{\qubit{\phi_n}\}_{n\in\nat}$ of quantum advice states satisfying  $\prob_{M}[M(\track{x}{\phi_n})=L(x)] \geq 1-\varepsilon_0$ for every length $n\in\nat$ and any string $x\in\Sigma^n$. For a later reference, the notation  $\varepsilon_{x}$ is reserved for the value $1-  \prob_{M}[M(\track{x}{\phi_n})=L(x)]$.
Choose an arbitrary error bound $\varepsilon\in(0,1/2)$.
For simplicity, we ignore the case of $x=\lambda$ and assume that $|x|\geq1$. Lemma \ref{measure-reduction} helps $M$ conduct a projective measurement only once after scanning $\dollar$.

If $\varepsilon_0\leq \varepsilon$, then $M$ outputs $L(x)$ with probability at least $1-\varepsilon_0\geq 1-\varepsilon$, and thus
the lemma is obviously  true. Therefore, in what follows,
let us concentrate on the case where  $0<\varepsilon<\varepsilon_0$.
Depending on the value $\varepsilon$, we select a positive integer $k$, which indicates the number of times we repeat in parallel the execution of $M$ on each input $x$,  as  the minimal odd number satisfying that $1-\sum_{i=0}^{\floors{k/2}} \smallcomb{k}{\ceilings{k/2}+i} \varepsilon_0^{\floors{k/2}-i}(1-\varepsilon_0)^{\ceilings{k/2}+i} \leq\varepsilon$.

We then prepare the collection of all
$k$-tuples $(q_{i_1},q_{i_2},\ldots,q_{i_k})\in Q^k$ as a new set $Q'$ of inner states. We express those  $k$-tuples as  $\qubit{q_{i_1}}\qubit{q_{i_2}}\cdots\qubit{q_{i_k}}$ using $k$ different registers.
Next, we simulate $M$ on a new rewritable advised 1qfa $M'$ in the following manner.
On input $x$, $M'$ runs $M$ on each of the $k$ registers  simultaneously in parallel. In the end of $M$'s computation, if the $k$ registers altogether hold  a basis vector $\qubit{q_{i_1}}\qubit{q_{i_2}}\cdots\qubit{q_{i_k}}$ for certain indices $i_1,i_2,\ldots,i_k$, then $M'$ enters a new inner state $q^{(i_1,i_2,\ldots,i_k)}$. Let $Q'_{fin}$ denote the set of all such new inner states. Next, we will partition $Q'_{fin}$ into three subsets, $Q'_{acc}$, $Q'_{rej}$, and $Q'_{other}$.
The set $Q'_{acc}$ (resp., $Q'_{rej}$) is composed of all inner states $q^{(i_1,i_2,\ldots,i_k)}$ for which  $|\{i\in[k]\mid q_i\in Q_{acc}\}|\geq \ceilings{k/2}$   (resp., $|\{i\in[k]\mid q_i\in Q_{rej}\}|\geq \ceilings{k/2}$). Let $Q'_{other} = Q'_{fin} - Q'_{acc}\cup Q'_{rej}$.
For each string $x$,  the probability that $M'$ successfully produces $L(x)$ equals $\sum_{i=0}^{\floors{k/2}} \smallcomb{k}{\ceilings{k/2}+i} \varepsilon_x^{\floors{k/2}-i}(1-\varepsilon_x)^{\ceilings{k/2}+i}$, which exceeds $\sum_{i=0}^{\floors{k/2}} \smallcomb{k}{\ceilings{k/2}+i} \varepsilon_0^{\floors{k/2}-i}(1-\varepsilon_0)^{\ceilings{k/2}+i}$ since $\varepsilon_x\leq \varepsilon_0$. By the choice of $k$,  $M'$ recognizes   $L$ with success probability at least $1-\varepsilon$.
\end{proofof}

\section*{Appendix: Proof of Lemmas \ref{metric-space-prop} and \ref{norm-property}}

This appendix presents the proofs of Lemmas \ref{metric-space-prop} and \ref{norm-property} that have been omitted from Section \ref{sec:QFA/n} for the sake of readability.


We begin with the proof of Lemma \ref{metric-space-prop}. In the following two separate proofs of the lemma, we assume that $\HH$ is any Hilbert space and that  $\psi=(\ket{\phi},\gamma_1,\gamma_2)$, $\psi' = (\ket{\phi'},\gamma'_1,\gamma'_2)$, and $\psi'' = (\ket{\phi''},\gamma''_1,\gamma''_2)$ are three arbitrary vectors in the metric vector space  $\YY_{\HH}$.
Recall that $\|\psi\|^2 = \| (\ket{\phi},\gamma_1,\gamma_2) \|^2 = \|\ket{\phi}\|^2 +|\gamma_1|^2+|\gamma_2|^2$.

\ms
\n{\bf Proof of Lemma \ref{metric-space-prop}(\ref{triangle-inequality}):\hs{3}}
Firstly, we note that $\|\psi\|^2\|\psi'\|^2 = (\|\ket{\phi}\|^2+\gamma_1^2+\gamma_2^2) (\|\ket{\phi'}\|^2+(\gamma'_1)^2+(\gamma'_2)^2) \geq ( \|\ket{\phi}\|\|\ket{\phi'}\| + |\gamma_1\gamma'_1| + |\gamma_2\gamma'_2|)^2$. From this inequality, it is easy to deduce
\begin{eqnarray*}
(\| \psi\|  + \| \psi'\| )^2
&=&  \|\psi\|^2 + 2\|\psi\| \|\psi'\| + \|\psi'\|^2 \\
&\geq& \| \qubit{\phi}\| ^2 +\| \qubit{\phi'}\| ^2
 + \gamma_1^2+(\gamma'_1)^2 +\gamma_2^2+ (\gamma'_2)^2 + 2 \left( \|\ket{\phi}\|\|\ket{\phi'}\| + |\gamma_1\gamma'_1| + |\gamma_2\gamma'_2| \right).
\end{eqnarray*}
Since $\psi+\psi' = (\ket{\phi}+\ket{\phi'},\gamma_1 + \gamma'_1,\gamma_2 + \gamma'_2)$, the norm $\| \psi + \psi'\| ^2$ is estimated as
\begin{eqnarray*}
\| \psi + \psi'\| ^2
&=&
\| \qubit{\phi}+\qubit{\phi'}\| ^2 +(\gamma_1 + \gamma'_1)^2+ (\gamma_2 + \gamma'_2)^2 \\
&\leq& \| \qubit{\phi}\| ^2 +\| \qubit{\phi'}\| ^2 + |\braket{\phi}{\phi'}|+ |\braket{\phi'}{\phi}| + \gamma_1^2
 + (\gamma'_1)^2 +\gamma_2^2 + (\gamma'_2)^2 + 2\left( |\gamma_1\gamma'_1| + |\gamma_2\gamma'_2| \right) \\
&\leq& \| \qubit{\phi}\| ^2 +\| \qubit{\phi'}\| ^2 + \gamma_1^2
 + (\gamma'_1)^2 +\gamma_2^2 + (\gamma'_2)^2  + 2\left( \| \qubit{\phi}\| \| \qubit{\phi'}\| + |\gamma_1\gamma'_1| + |\gamma_2\gamma'_2| \right) \\
&\leq& (\| \psi\|  + \| \psi'\| )^2,
\end{eqnarray*}
where the second inequality follows from the fact that $|\braket{\phi}{\phi'}|\leq \| \qubit{\phi}\| \| \qubit{\phi'}\|$. Therefore, we obtain the inequality $(\| \psi\|  + \| \psi'\| )^2 \geq \| \psi + \psi'\| ^2$, which immediately yields the desired claim.
\qed

\ms
\n{\bf Proof of Lemma \ref{metric-space-prop}(\ref{norm-triangle-prop}):\hs{3}}
We start with the following inequality:
\begin{eqnarray*}
\lefteqn{(\| \psi -\psi' \| + \|\psi'-\psi''\| )^2} \hs{5} \\
&\geq& (\|\ket{\phi}-\ket{\phi'}\| + \|\ket{\phi'}-\ket{\phi''}\|)^2 + (|\gamma_1-\gamma'_1|+|\gamma'_1-\gamma''_1|)^2 + (|\gamma_2-\gamma'_2|+|\gamma'_2-\gamma''_2|)^2.
\end{eqnarray*}
The well-known triangular inequalities for vectors and real numbers further yield
\begin{equation*}
(\| \psi -\psi' \| + \|\psi'-\psi''\| )^2
\geq \|\ket{\phi}-\ket{\phi''}\|^2 + |\gamma_1-\gamma''_1|^2 + |\gamma_2-\gamma''_2|^2.
\end{equation*}
Thus, we obtain $(\| \psi -\psi' \| + \|\psi'-\psi''\| )^2
\geq \|\psi-\psi''\|^2$, which is logically equivalent to the desired claim.
\qed

\ms

Next, we will prove Lemma \ref{norm-property}.
Recall that $\Sigma$ is a basis alphabet and $M = (Q,\Sigma,\{U_{\sigma}\}_{\sigma\in\check{\Sigma}},q_0,Q_{acc},Q_{rej})$ is a 1qfa.
In what follows, we fix a string $x=x_1x_2\cdots x_n$ of length $n$ with  $x_i\in (\check{\Sigma})^*$ for every $i\in[n]$, where $\check{\Sigma} = \Sigma\cup\{\cent,\dollar\}$.
Let $\qubit{\phi}$ and $\qubit{\phi'}$ be any
two  quantum states in $E_{Q}$ and let $\psi = (\qubit{\phi},\gamma_1,\gamma_2)$ and $\psi' = (\qubit{\phi'},\gamma'_1,\gamma'_2)$ be any two vectors in $\YY_{E_Q}$.
Recall also that, for each symbol $\sigma\in\check{\Sigma}$, $T_{\sigma} = P_{non}U_{\sigma}$ and $T_x = T_{x_{n}}T_{x_{n-1}}\cdots T_{x_2}T_{x_1}$.

To show the target lemma,  we additionally define $\qubit{\phi_1} = \qubit{\phi}$ and $\qubit{\phi'_1} = \qubit{\phi'}$ and, for each index
$i\in[2,n+1]_{\integer}$, $\qubit{\phi_i} = T_{x_1x_2\cdots x_{i-1}}\qubit{\phi_1}$ and $\qubit{\phi'_i} = T_{x_1x_2\cdots x_{i-1}}\qubit{\phi'_1}$.
Moreover,  we set
$\alpha_i = \|  P_{acc}U_{x_i}\qubit{\phi_i}\| ^2$ and $\beta_i = \|  P_{rej}U_{x_i}\qubit{\phi_i}\| ^2$; similarly, we define $\alpha'_i$ and $\beta'_i$ using $\qubit{\phi'_i}$ in place of $\qubit{\phi_i}$.

Before providing the desired proof of the lemma, we will list seven  useful properties. In the next claim, let $\gamma_1,\gamma_2,\gamma'_1,\gamma'_2$ be any real numbers and let $\ket{\phi},\ket{\phi'}$ be any quantum states in $E_{Q}$.

\begin{claim}\label{useful-formulas}
\begin{enumerate}
  \setlength{\topsep}{0mm}%
  \setlength{\itemsep}{0mm}
  \setlength{\parskip}{0cm}%

\item\label{phi} $\| \qubit{\phi} \| ^2
= \|  T_x\qubit{\phi} \| ^2 +
\sum_{i=1}^{n} (\alpha_i + \beta_i)$.

\item\label{phi-phi-prime} $\| \qubit{\phi}-\qubit{\phi'}\|^2 = \|T_x(\qubit{\phi}-\qubit{\phi'})\|^2 + \sum_{i=1}^{n}\|P_{acc}U_{x_i}(\qubit{\phi_i}-\qubit{\phi'_i})\|^2 +
\sum_{i=1}^{n}\| P_{rej}U_{x_i}(\qubit{\phi_i}-\qubit{\phi'_i})\|^2$.

\item\label{P-acc-upper} $(\sqrt{\gamma_1^2+\sum_{i=1}^{n}\alpha_i} - \sqrt{(\gamma'_1)^2+\sum_{i=1}^{n}\alpha'_i})^2 \leq (|\gamma_1| - |\gamma'_1| )^2 + \sum_{i=1}^{n}\|P_{acc}U_{x_i}(\ket{\phi_i}-\ket{\phi'_i})\|^2$. A similar  inequality holds for $(\gamma'_1,\gamma'_2,\beta_i,\beta'_i,P_{rej})$.

\item\label{alpha-beta} $\|\qubit{\phi}-\qubit{\phi'}\|^2 - \|T_x(\qubit{\phi}-\qubit{\phi'})\|^2  \geq   (\sqrt{\sum_{i=1}^{n}\alpha_i} - \sqrt{\sum_{i=1}^{n}\alpha'_i} )^2 + (\sqrt{\sum_{i=1}^{n}\beta_i} - \sqrt{\sum_{i=1}^{n}\beta'_i})^2$.

\item\label{P-acc-alpha-lower} $\sum_{i=1}^{n}\|P_{acc}U_{x_i}(\ket{\phi}-\ket{\phi'})\|^2 \leq 2\sum_{i=1}^{n}(\alpha_i+\alpha'_i)$. A similar inequality holds for $(\beta_i,\beta'_i,P_{rej})$.

\item\label{P-acc-lower}  If $\gamma_1^2+(\gamma'_1)^2\leq 1$ holds, then $(\sqrt{\gamma_1^2+\sum_{i=1}^{n}\alpha_i} - \sqrt{(\gamma'_1)^2+\sum_{i=1}^{n}\alpha'_i} )^2 +  \sum_{i=1}^{n}(\alpha_i + \alpha'_i) + (2\sqrt{2}) \sqrt{ \sum_{i=1}^{n}(\alpha_i + \alpha'_i)} \geq (|\gamma_1|-|\gamma'_1|)^2 + \sum_{i=1}^{n}\|P_{acc}U_{x_i}(\ket{\phi}-\ket{\phi'})\|^2$. A similar inequality holds for $(\gamma_2,\gamma'_2,\beta_i,\beta'_i,P_{rej})$.

\item\label{diff-phi} $2|\braket{\phi}{\phi'} - \bra{\phi}T^{\dagger}_{x}T_{x}\ket{\phi'}| \leq
 (\| \qubit{\phi}\| ^2 - \| T_x\qubit{\phi}\| ^2)
 + (\| \qubit{\phi'}\| ^2 - \| T_x\qubit{\phi'}\| ^2)$.
\end{enumerate}
\end{claim}

\begin{proof}
(\ref{phi}) It holds that  $U_{x_i} = T_{x_i}+P_{acc}U_{x_i}+P_{rej}U_{x_i}$  for each index $i\in[n]$.
Since $U^{\dagger}_{x_i}U_{x_i} = I$, it obviously follows that
\begin{equation}\label{eqn:Pacc-Prej-T}
\| \qubit{\phi_i}\|^2 = \|U_{x_i}\qubit{\phi_i}\|^2 =
\| T_{x_i}\qubit{\phi_i}\|^2 + \|P_{acc}U_{x_i}\qubit{\phi_i}\|^2 + \| P_{rej}U_{x_i}\qubit{\phi_i}\|^2.
\end{equation}
Applying this equality repeatedly with $\ket{\phi_{i+1}} = T_{x_i}\ket{\phi_i}$, we then obtain
\begin{eqnarray*}
\| \qubit{\phi_1} \| ^2
&=& \|  \qubit{\phi_2} \| ^2 +
\|  P_{acc}U_{x_1} \qubit{\phi_1} \| ^2
 + \|  P_{rej}U_{x_1} \qubit{\phi_1} \| ^2 \\
&=& \|  \qubit{\phi_3} \| ^2 +
 \sum_{i=1}^{2}\|  P_{acc}U_{x_i} \qubit{\phi_i} \| ^2
 + \sum_{i=1}^{2}\|  P_{rej}U_{x_i} \qubit{\phi_i} \| ^2 \\
&=& \cdots\cdots\cdots \\
&=& \|  \qubit{\phi_{n+1}} \| ^2 +
\sum_{i=1}^{n}\|  P_{acc}U_{x_i} \qubit{\phi_i} \| ^2
 + \sum_{i=1}^{n}\|  P_{rej}U_{x_i} \qubit{\phi_i} \| ^2.
\end{eqnarray*}
The desired formula in the claim immediately follows since $\ket{\phi} = \ket{\phi_1}$, $T_x\qubit{\phi} = \qubit{\phi_{n+1}}$,   $\alpha_i  = \|  P_{acc}U_{x_i} \qubit{\phi_i} \| ^2$, and  $\beta_i  = \|  P_{rej}U_{x_i} \qubit{\phi_i} \| ^2$.

(\ref{phi-phi-prime}) This target equality can be obtained by an argument similar to the proof of Claim \ref{useful-formulas}(\ref{phi}) using, instead of Eq.(\ref{eqn:Pacc-Prej-T}), the equality
\begin{equation}\label{eqn:diff-Pacc-Prej-T}
\| \qubit{\phi_i}-\qubit{\phi'_i}\|^2 = \|T_{x_i}(\qubit{\phi_i}-\qubit{\phi'_i})\|^2 + \|P_{acc}U_{x_i}(\qubit{\phi_i}-\qubit{\phi'_i})\|^2 + \| P_{rej}U_{x_i}(\qubit{\phi_i}-\qubit{\phi'_i})\|^2.
\end{equation}

(\ref{P-acc-upper}) Note that $(\sqrt{a+b}-\sqrt{c+d})^2\leq (\sqrt{a}-\sqrt{c})^2+(\sqrt{b}-\sqrt{d})^2$ holds for any real numbers $a,b,c,d\geq0$,  because this formula is logically equivalent to $\sqrt{(a+b)(c+d)} \geq \sqrt{ac}+\sqrt{bd}$. It thus follows that
\begin{equation}\label{alpha-gamma-ballance}
\left( \sqrt{\gamma_1^2+\sum_{i}\alpha_i} - \sqrt{(\gamma'_1)^2+\sum_{i}\alpha'_i} \right)^2 \leq \left( |\gamma_1| - |\gamma'_1| \right)^2 + \left( \sqrt{\sum_{i}\alpha_i}-\sqrt{\sum_{i}\alpha'_i} \right)^2.
\end{equation}

\sloppy
For the time being, let $\{\ket{\xi_i},\ket{\xi'_i}\}_{i\in[n]}$ represent
any collection of $n$ pairs of vectors.
Notice that
$\sqrt{(\sum_{i}\|\ket{\xi_i}\|^2)( \sum_{i}\|\ket{\xi'_i}\|^2)}\geq 2\sum_{i}\|\ket{\xi_i}\| \|\ket{\xi'_i}\|$. Since $2\|\ket{\xi_i}\| \|\ket{\xi'_i}\| \geq \braket{\xi_i}{\xi'_i} + \braket{\xi'_i}{\xi_i}$, we then obtain $\sqrt{(\sum_{i}\|\ket{\xi_i}\|^2) (\sum_{i}\|\ket{\xi'_i}\|^2)} \geq \sum_{i}(\braket{\xi_i}{\xi'_i} + \braket{\xi'_i}{\xi_i})$. This inequality is used to prove that $\sum_{i}\|\ket{\xi_i}-\ket{\xi'_i}\|^2 \geq (\sqrt{\sum_{i}\|\ket{\xi_i}\|^2} - \sqrt{\sum_{i}\|\ket{\xi'_i}\|^2})^2$.
From this inequality, we immediately obtain
\begin{eqnarray*}
\lefteqn{\sum_{i=1}^{n}\|  P_{acc}U_{x_i}(\qubit{\phi_i}-\qubit{\phi'_i})\| ^2
\;\;=\;\; \sum_{i=1}^{n} \|P_{acc}U_{x_i}\ket{\phi_i}  - P_{acc}U_{x_i}\ket{\phi'_i}\|^2} \hs{10} \\
&\geq&
 \left( \sqrt{ \sum_{i} \|  P_{acc}U_{x_i}\qubit{\phi_i}\|^2 }   -
 \sqrt{ \sum_{i} \|  P_{acc}U_{x_i}\qubit{\phi'_i}\|^2 } \right)^2
\;\;=\;\;  \left( \sqrt{ \sum_{i} \alpha_i }   - \sqrt{ \sum_{i} \alpha'_i  } \right)^2.
\end{eqnarray*}
Combining this inequality with Eq.(\ref{alpha-gamma-ballance}) (by ``temporarily'' setting $\gamma_1=\gamma'_1=0$), we derive the desired claim.
In a similar manner, it follows that $(\sqrt{\gamma_2^2+\sum_{i}\beta_i} - \sqrt{(\gamma'_2)^2+\sum_{i}\beta'_i})^2 \leq (|\gamma_2|-|\gamma'_2|)^2 +  \sum_{i}\|P_{rej}U_{x_i}(\ket{\phi_i}-\ket{\phi'_i})\|^2$.

(\ref{alpha-beta})
From Claim \ref{useful-formulas}(\ref{phi-phi-prime}) combined with Claim \ref{useful-formulas}(\ref{P-acc-upper}), it instantly follows that
\begin{eqnarray*}
\|\qubit{\phi} - \qubit{\phi'} \|^2 &=& \|T_x(\qubit{\phi}-\qubit{\phi'})\|^2 + \sum_{i=1}^{n} \|P_{acc}U_{x_i}(\qubit{\phi}-\qubit{\phi'})\|^2 + \sum_{i=1}^{n} \|P_{rej}U_{x_i}(\qubit{\phi}-\qubit{\phi'})\|^2  \label{phi-eqn} \\
&\geq&  \|T_x(\qubit{\phi}-\qubit{\phi'})\|^2 + \left(\sqrt{\sum_{i}\alpha_i} - \sqrt{\sum_{i}\alpha'_i}\right)^2  + \left(\sqrt{\sum_{i}\beta_i} - \sqrt{\sum_{i}\beta'_i}\right)^2.
\end{eqnarray*}

(\ref{P-acc-alpha-lower})
Since $\|\ket{\xi}-\ket{\xi'}\|^2 \leq 2 (\|\ket{\xi}\|^2+\|\ket{\xi'}\|^2)$ holds for any vectors $\ket{\xi},\ket{\xi'}$, we obtain
\begin{equation*}
\sum_{i=1}^{n} \|P_{acc}U_{x_i}(\ket{\phi_i}-\ket{\phi'_i})\|^2 \leq
2 \sum_{i=1}^{n} \left( \|P_{acc}U_{x_i}\ket{\phi_i}\|^2 + \|P_{acc}U_{x_i}\ket{\phi'_i}\|^2  \right)
= 2 \sum_{i=1}^{n} (\alpha_i+\alpha'_i).
\end{equation*}
A similar argument leads to the inequality $\sum_{i}\|P_{rej}U_{x_i}(\ket{\phi_i}-\ket{\phi'_i})\|^2 \leq 2\sum_{i}(\beta_i+\beta'_i)$.

(\ref{P-acc-lower})
Let $a,b,c,d\in[0,1]$. Note that $(\sqrt{a^2+b}-\sqrt{c^2+d})^2 = (a-c)^2 + (b+d) +2(ac-\sqrt{(a^2+b)(c^2+d)})$. For the last term  $\sqrt{(a^2+b)(c^2+d)}$, it holds that $\sqrt{(a^2+b)(c^2+d)} \leq \sqrt{a^2c^2+(a^2+c^2)(b+d)+bd} \leq \sqrt{a^2c^2} + \sqrt{(a^2+c^2)(b+d) + bd} $. Since $b,d\leq 1$, $bd\leq b+d$ follows.
If $a^2+c^2\leq 1$ further holds, then
we obtain $\sqrt{(a^2+c^2)(b+d) + bd} \leq \sqrt{(a^2+c^2+1)(b+d)} \leq \sqrt{2(b+d)}$, which finally yields $\sqrt{(a^2+b)(c^2+d)} \leq ac + \sqrt{2(b+d)}$. Overall, we conclude that $(\sqrt{a^2+b}-\sqrt{c^2+d})^2 \geq (a-c)^2 +(b+d) - (2\sqrt{2})\sqrt{b+d}$.
Therefore, in our case, it follows that
\begin{eqnarray*}
\lefteqn{\left( \sqrt{\gamma_1^2+\sum_{i}\alpha_i} -  \sqrt{ (\gamma'_1)^2+\sum_{i}\alpha'_i} \right)^2} \hs{10} \\
&\geq& (|\gamma_1| - |\gamma'_1|)^2 + \left( \sum_{i=1}^{n}\alpha_i + \sum_{i=1}^{n}\alpha'_i\right) - 2\sqrt{2} \sqrt{ \sum_{i}\alpha_i + \sum_{i}\alpha'_i } \\
&\geq& (|\gamma_1| - |\gamma'_1|)^2 + \sum_{i=1}^{n}\|P_{acc}U_{x_i}(\ket{\phi_i}-\ket{\phi'_i})\|^2 - \sum_{i=1}^{n}(\alpha_i + \alpha'_i) - 2\sqrt{2} \sqrt{ \sum_{i}(\alpha_i + \alpha'_i) },
\end{eqnarray*}
where the last inequality  comes from Claim \ref{useful-formulas}(\ref{P-acc-alpha-lower}).

(\ref{diff-phi})
Since $U^{\dagger}_{x_i}U_{x_i} = I$ and  $U_{x_i} = T_{x_i}+P_{acc}U_{x_i}+P_{rej}U_{x_i}$, for each index $i\in[n]$, we conclude
\begin{equation*}\label{eqn:U-T-Pacc-Prej}
\braket{\phi_i}{\phi'_i} = \bra{\phi_i}U^{\dagger}_{x_i}U_{x_i}\ket{\phi'_i}
= \bra{\phi_i}T^{\dagger}_{x_i}T_{x_i}\ket{\phi'_i}
+ \bra{\phi_i}U^{\dagger}_{x_i}P_{acc}U_{x_i}\ket{\phi'_i}
+ \bra{\phi_i}U^{\dagger}_{x_i}P_{rej}U_{x_i}\ket{\phi'_i}.
\end{equation*}
Using this equality together with  Eq.(\ref{eqn:Pacc-Prej-T}) as well as the inequality $|\braket{\xi}{\xi'}|\leq \| \qubit{\xi} \|\| \qubit{\xi'} \|$, we obtain
\begin{eqnarray}
\nonumber \lefteqn{|\braket{\phi_i}{\phi'_i} - \bra{\phi_i}T^{\dagger}_{x_i}T_{x_i}\ket{\phi'_i}|
\;\;\leq\;\; |\bra{\phi_i}U^{\dagger}_{x_i}P_{acc}U_{x_i}\ket{\phi'_i} | + |\bra{\phi_i}U^{\dagger}_{x_i}P_{rej}U_{x_i}\ket{\phi'_i} |}\hs{15} \\
&\leq& \nonumber \|  P_{acc}U_{x_i}\qubit{\phi_i}\|   \|  P_{acc}U_{x_i}\qubit{\phi'_i}\|
 + \|  P_{rej}U_{x_i}\qubit{\phi_i}\|    \|  P_{rej}U_{x_i}\qubit{\phi'_i}\|  \\
&\leq& \nonumber \frac{1}{2}\left[ \left( \|  P_{acc}U_{x_i}\qubit{\phi_i}\| ^2 + \|  P_{rej}U_{x_i}\qubit{\phi_i}\| ^2 \right) + \left( \|  P_{acc}U_{x_i}\qubit{\phi'_i}\| ^2 + \|  P_{rej}U_{x_i}\qubit{\phi'_i}\| ^2 \right)  \right] \\
&=& \nonumber \frac{1}{2} \left[ (\| \qubit{\phi_i}\| ^2 - \|  T_{x_i}\qubit{\phi_i}\| ^2)  + (\| \qubit{\phi'_i}\| ^2 - \|  T_{x_i}\qubit{\phi'_i}\| ^2) \right] \\
&=& \label{phi-reduction} \frac{1}{2} \left[ (\| \qubit{\phi_i}\| ^2 - \| \qubit{\phi_{i+1}}\| ^2)  + (\| \qubit{\phi'_i}\| ^2 - \| \qubit{\phi'_{i+1}}\| ^2) \right].
\end{eqnarray}
Recalling $x=x_1x_2\cdots x_n$, we note that $\ket{\phi} = \ket{\phi_1}$, $\ket{\phi'} = \ket{\phi'_1}$, $T_x\ket{\phi} = T_{x_n}\ket{\phi_n}$, and $T_x\ket{\phi'} = T_{x_n}\ket{\phi'_n}$. Since the equality
\begin{equation*}
|\braket{\phi}{\phi'} - \bra{\phi}T^{\dagger}_{x}T_{x}\ket{\phi'}|
= \left|\sum_{i=1}^{n}(\braket{\phi_i}{\phi'_i} - \bra{\phi_i}T^{\dagger}_{x_i}T_{x_i}\ket{\phi'_i})\right|
\end{equation*}
holds, Eq.(\ref{phi-reduction}) helps us deduce
\begin{eqnarray*}
|\braket{\phi}{\phi'} - \bra{\phi}T^{\dagger}_{x}T_{x}\ket{\phi'}|
&\leq&
 \sum_{i=1}^{n}|\braket{\phi_i}{\phi'_i} - \bra{\phi_i}T^{\dagger}_{x_i}T_{x_i}\ket{\phi'_i}| \\
&=& \frac{1}{2}\sum_{i=1}^{n} \left[ (\| \qubit{\phi_i}\| ^2 - \| \qubit{\phi_{i+1}}\| ^2)  + (\| \qubit{\phi'_i}\| ^2 - \| \qubit{\phi'_{i+1}}\| ^2) \right] \\
&=& \frac{1}{2} \left[ (\| \qubit{\phi_1}\| ^2 - \| \qubit{\phi_{n+1}}\| ^2)  + (\| \qubit{\phi'_1}\| ^2 - \| \qubit{\phi'_{n+1}}\| ^2)\right].
\end{eqnarray*}
The desired claim immediately follows from the fact that $\qubit{\phi_{n+1}} = T_{x}\qubit{\phi}$ and  $\qubit{\phi'_{n+1}} = T_{x}\qubit{\phi'}$.
\end{proof}


Now, we are ready to prove Lemma \ref{norm-property}. In the following proof of this lemma, let $\psi = (\qubit{\phi},\gamma_1,\gamma_2)$ and $\psi' = (\qubit{\phi'},\gamma'_1,\gamma'_2)$ be any vectors in $\YY_{E_Q}$ with $\gamma_1,\gamma_2,\gamma'_1,\gamma'_2\in[0,1]$. The operator $\hat{T}_x$ is the functional composition $\hat{T}_{x_n}\hat{T}_{x_{n-1}}\cdots \hat{T}_{x_2}\hat{T}_{x_1}$, where $x=x_1x_2\cdots x_n$.

\ms
\n{\bf Proof of Lemma \ref{norm-property}(\ref{diff-estimate}):\hs{3}}
A simple calculation shows
\begin{eqnarray*}
\lefteqn{\| \qubit{\phi}-\qubit{\phi'}\| ^2 - \|  T_{x}(\qubit{\phi}-\qubit{\phi'})\| ^2
\;\;=\;\; (\| \qubit{\phi}-\qubit{\phi'}\| ^2 - \|  T_{x}\qubit{\phi}- T_{x}\qubit{\phi'}\| ^2}  \hs{5} \\
&=& (\| \qubit{\phi}\| ^2 - \|  T_{x}\qubit{\phi}\| ^2) + (\| \qubit{\phi'}\| ^2 - \|  T_{x}\qubit{\phi'}\| ^2)
+ (\bra{\phi}T^{\dagger}_{x}T_{x}\ket{\phi'} + \bra{\phi'}T^{\dagger}_{x}T_{x}\ket{\phi} - \braket{\phi}{\phi'} - \braket{\phi'}{\phi}) \\
&\leq& (\| \qubit{\phi}\| ^2 - \|  T_{x}\qubit{\phi}\| ^2) + (\| \qubit{\phi'}\| ^2 - \|  T_{x}\qubit{\phi'}\| ^2)
+ |\bra{\phi}T^{\dagger}_{x}T_{x}\ket{\phi'} - \braket{\phi}{\phi'} | + \left|\bra{\phi'}T^{\dagger}_{x}T_{x}\ket{\phi} - \braket{\phi'}{\phi}\right|.
\end{eqnarray*}
Combining the above inequality with Claim \ref{useful-formulas}(\ref{diff-phi}), we then  obtain the desired consequence.
\qed

\ms
\n{\bf Proof of Lemma \ref{norm-property}(\ref{upper-bound}):\hs{3}}
By applying $\hat{T}_{x_1}$ to $\psi$, we can obtain  $\hat{T}_{x_1}\psi = (T_{x_1}\qubit{\phi}, ( \gamma_1^2+ \alpha_1 )^{1/2},  ( \gamma_2^2+ \beta_1 )^{1/2} )$. Similarly, the application of $\hat{T}_{x_1x_2}$ to $\psi$ leads to $\hat{T}_{x_1x_2}\psi = (T_{x_1x_2}\qubit{\phi}, (  \gamma_1^2 +\sum_{i=1}^{2}\alpha_i )^{1/2}, ( \gamma_2^2+\sum_{i=1}^{2}\beta_i )^{1/2} )$. By continuing this process up to $n$, we then  obtain
\begin{equation*}
\hat{T}_{x}\psi = \left( T_{x}\qubit{\phi}, \sqrt{ \gamma_1^2+\sum_{i=1}^{n}\alpha_i }, \sqrt{ \gamma_2^2 +\sum_{i=1}^{n}\beta_i } \;\; \right).
\end{equation*}
Similar reasoning shows that $\hat{T}_{x}\psi' =(T_{x}\qubit{\phi'}, ((\gamma'_1)^2+\sum_{i=1}^{n}\alpha'_i)^{1/2}, ((\gamma'_2)^2 +\sum_{i=1}^{n}\beta'_i)^{1/2} )$.

Now, let us consider the norm of $\hat{T}_{x}\psi - \hat{T}_{x}\psi'$.
Note that $\hat{T}_{x}\psi - \hat{T}_{x}\psi'$ equals $(T_x(\ket{\phi}-\ket{\phi'}), (\gamma_1^2+\sum_{i}\alpha_i)^{1/2} -  ((\gamma'_1)^2+\sum_{i}\alpha'_i)^{1/2}, (\gamma_2+\sum_{i}\beta_i)^{1/2} -  ((\gamma'_2)^2+\sum_{i}\beta'_i)^{1/2})$.
The value $\| \hat{T}_{x}\psi - \hat{T}_{x}\psi'\| ^2$ therefore equals
\begin{equation}\label{T-minus-T}
\|  T_{x}(\qubit{\phi}-\qubit{\phi'})\| ^2
+ \left( \sqrt{ \gamma_1^2+\sum_{i}\alpha_i } -
 \sqrt{ (\gamma'_1)^2+\sum_{i}\alpha'_i } \right)^2
+ \left( \sqrt{ \gamma_2^2+\sum_{i}\beta_i } -
 \sqrt{ (\gamma'_2)^2+\sum_{i}\beta'_i }  \right)^2.
\end{equation}
Recall that $\gamma_1,\gamma_2,\gamma'_1,\gamma'_2\geq0$.
The equality Eq.(\ref{T-minus-T}) together with Claim \ref{useful-formulas}(\ref{P-acc-upper}) then implies
\begin{eqnarray*}
\| \hat{T}_{x}\psi - \hat{T}_{x}\psi'\| ^2
&\leq& \nonumber \|  T_{x}(\qubit{\phi}-\qubit{\phi'})\| ^2
+ (\gamma_1-\gamma'_1)^2 + (\gamma_2-\gamma'_2)^2 \\
&& \label{T-upper-estimate} \hs{10} + \sum_{i}\| P_{acc}U_{x_i}(\ket{\phi_i}-\ket{\phi'_i})\|^2 + \sum_{i}\| P_{rej}U_{x_i}(\ket{\phi_i}-\ket{\phi'_i})\|^2.
\end{eqnarray*}
From the above inequality with the help of Claim \ref{useful-formulas}(\ref{phi-phi-prime}), we conclude
\begin{eqnarray}
\|\psi -\psi'\|^2
&=&  \nonumber  \|\qubit{\phi}-\qubit{\phi'}\|^2 + \left( \gamma_1 -\gamma'_1\right)^2 + \left( \gamma_2 -\gamma'_2\right)^2 \\
&=&  \nonumber \|T_x(\qubit{\phi}-\qubit{\phi'})\|^2 + \left( \gamma_1 -\gamma'_1\right)^2 + \left( \gamma_2 -\gamma'_2 \right)^2    \\
&& \label{psi-minus-phi-prime} \hs{10} + \sum_{i}\| P_{acc}U_{x_i}(\ket{\phi_i}-\ket{\phi'_i})\|^2 + \sum_{i}\| P_{rej}U_{x_i}(\ket{\phi_i}-\ket{\phi'_i})\|^2 \\
&\geq& \nonumber \left\|\hat{T}_x\psi - \hat{T}_x\psi' \right\|^2.
\end{eqnarray}
\qed

\ms
\n{\bf Proof of Lemma \ref{norm-property}(\ref{lower-bound}):\hs{3}}
A simple application of  Claim \ref{useful-formulas}(\ref{P-acc-lower}) to Eq.(\ref{T-minus-T}) leads to the inequality:
\begin{eqnarray*}
\left\| \hat{T}_{x}\psi - \hat{T}_{x}\psi'\right\| ^2
&\geq& \|  T_{x}(\qubit{\phi}-\qubit{\phi'})\| ^2  + \left( \gamma_1 -\gamma'_1\right)^2 + \left( \gamma_2 -\gamma'_2\right)^2 \\
&& \hs{5} +\sum_{i=1}^{n} \| P_{acc}U_{x_i}(\ket{\phi_i}-\ket{\phi'_i})\|^2
 + \sum_{i=1}^{n} \| P_{rej}U_{x_i}(\ket{\phi_i}-\ket{\phi'_i})\|^2  \\
&& \hs{5} -  \left( \sum_{i=1}^{n}(\alpha_i+\alpha'_i) + \sum_{i=1}^{n}(\beta_i+\beta'_i) \right)
-2\sqrt{2} \left( \sqrt{\sum_{i}(\alpha_i+\alpha'_i)} + \sqrt{\sum_{i}(\beta_i+\beta'_i)} \right).
\end{eqnarray*}
Note that $\sqrt{a+b}+\sqrt{c+d}\leq \sqrt{2(a+b+c+d)}$ holds for any $a,b,c,d\geq0$.
Using this inequality together with Eq.(\ref{psi-minus-phi-prime}), the value $\| \hat{T}_{x}\psi - \hat{T}_{x}\psi'\| ^2$ is further lower-bounded as
\begin{equation*}
\left\| \hat{T}_{x}\psi - \hat{T}_{x}\psi'\right\| ^2
\geq \| \psi - \psi' \|^2
-   \sum_{i=1}^{n}(\alpha_i+\beta_i) - \sum_{i=1}^{n}(\alpha'_i+\beta'_i)
- 4 \sqrt{ \sum_{i}(\alpha_i+\alpha'_i) +  \sum_{i}(\beta_i+\beta'_i) }.
\end{equation*}
Finally, we apply Claim \ref{useful-formulas}(\ref{phi}) to the above inequality and conclude
\begin{eqnarray*}
\left\| \hat{T}_{x}\psi - \hat{T}_{x}\psi'\right\| ^2
&\geq& \| \psi - \psi' \|^2
- \left( \|\ket{\phi}\|^2 - \| T_{x}\ket{\phi} \|^2 \right)  - \left(  \|\ket{\phi'}\|^2 - \| T_{x}\ket{\phi'} \|^2 \right) \\
&& \hs{15} - 4  \sqrt{ \left( \|\ket{\phi}\|^2 - \| T_{x}\ket{\phi} \|^2 \right)   +  \left(  \|\ket{\phi'}\|^2 - \| T_{x}\ket{\phi'} \|^2 \right) }.
\end{eqnarray*}
\qed

\bs
\paragraph{\bf Acknowledgements}
The author is grateful to anonymous reviewers for useful comments and, moreover, providing him with additional references noted in Section \ref{sec:introduction}. He also thanks Marcos Villagra for a pleasant discussion on reversible finite automata.

\let\oldbibliography\thebibliography
\renewcommand{\thebibliography}[1]{%
  \oldbibliography{#1}%
  \setlength{\itemsep}{1pt}%
}
\bibliographystyle{plain}

\begin{thebibliography}{Gur91}

\bibitem{Aar05}
S. Aaronson. Limitations of quantum advice and one-way communication. {\em  Theory of Computing}, 1 (2005) 1--28.

\bibitem{AF98}
A. Ambainis and R. Freivalds. 1-way quantum finite automata: strengths, weaknesses, and generalizations. In {\em Proc. of the 39th Annual Symposium on Foundations of Computer Science} (FOCS'98), pp.332--342, 1998.

\bibitem{AKV01}
A. Ambainis, A. \c{K}ikusts, and M. Valdats. On the class of languages recognized by 1-way quantum finite automata. In {\em Proc. of the 18th Annual Symposium on Theoretical Aspects of Computer Science} (STACS 2001), Lecture Notes in Computer Science, Vol.2010, pp.75--86, Springer, 2001.

\bibitem{ANTV02}
A. Ambainis, A. Nayak, A. Ta-Shma, and U. Vazirani. Dense quantum coding and quantum finite automata. {\em J. ACM}, 49 (2002), 496--511.

\bibitem{Ang82}
D. Angluin. Inference of reversible languages. {\em J. ACM}, 29 (1982) 741--765.

\bibitem{BP02}
A. Brodsky and N. Pippenger. Characterizations of 1-way quantum finite automata. {\em SIAM J. Comput.}, 31 (2002) 1456--1478.

\bibitem{DH95}
C. Damm and M. Holzer. Automata that take advice. In {\em Proc. of
the 20th Symposium on Mathematical Foundations of Computer Science} (MFCS'95),
Lecture Notes in Computer Science, Vol.969, pp.149--152, Springer, 1995.

\bibitem{Fre10}
R. Freivalds. Amount of nonconstructivity in determinsitic finite automata. {\em Theor. Comput. Sci.}, 411 (2010) 3436--3443.

\bibitem{Fre12}
R. Freivalds. Multiple usage of random bits in finite automata. In {\em Proc. of the 9th Annual Conference on Theory and Applications of Models of Computation} (TAMC 2012), Lecture Notes in Computer Science, vol.7287, pp.537--547, 2012.

\bibitem{Gru00}
J. Gruska. {\em Quantum Computing}. McGraw Hill, 2000.

\bibitem{HMU01}
J. E. Hopcroft, R. Motwani, and J. D. Ullman. {\em Introduction to Automata Theory, Languages, and Computation}, (Second Edition). Addison Wesley, 2001.

\bibitem{KL82}
R. M. Karp and R. Lipton.
Turing machines that take advice. {\em L{'}Enseignement Math{\'e}matique},  2nd series, 28 (1982) 191--209.

\bibitem{KW97}
A. Kondacs and J. Watrous. On the power of quantum finite state automata.
In {\em Proc. of the 38th Annual Symposium on Foundations of Computer Science} (FOCS'97), pp.66--75, 1997.

\bibitem{KSY13}
U. K\"{u}\c{c}\"{u}k, A. C. C. Say, and A. Yakaryilmaz. Finite automata with advice tapes. In {\em Proc. of the 17th International Conference on Developments in Language Theory} (DLT 2013), Lecture Notes in Computer Science, vol.7907, pp.301--312.

\bibitem{MC00}
C. Moore and J. Crutchfield. Quantum automata and quantum languages. {\em Theor. Comput. Sci.}, 237 (2000) 275--306.

\bibitem{Mic91}
P. Michel.
An NP-complete language accepted in linear
time by a one-tape Turing machine.
{\em Theor. Comput. Sci.}, 85 (1991) 205--212.

\bibitem{NY04b}
H. Nishimura and T. Yamakami. Polynomial-time quantum computation with advice. {\em Inf. Process. Lett.}, 90 (2004) 195--204.

\bibitem{NY04a}
H. Nishimura and T. Yamakami. An application of quantum finite automata to interactive proof systems (extended abstract). In
{\em Proc. of the 9th International Conference on Implementation and
Application of Automata} (CIAA 2004),
Lecture Notes in Computer Science, Vol.3317, pp.225--236, Springer, 2005.

\bibitem{NY09}
H. Nishimura and T. Yamakami. An application of quantum finite automata to interactive proof systems. {\em J. Comput. System Sci.}, 75 (2009) 255--269.
This is a complete version of the first half part of \cite{NY04a}.

\bibitem{NY14}
H. Nishimura and T. Yamakami. Interactive proofs with quantum finite automata. To appear in {\em Theor. Comput. Sci.}
This extends the second half part of \cite{NY04a}. Available also at arXiv:1401.2929.

\bibitem{Pas00}
K. Paschen. Quantum finite automata using ancilla qubits. Technical Report, Universit\"{a}t Karlsruhe, 2000. Available at http://digbib.ubka.uni-karlsruhe.de/volltexte/1452000.

\bibitem{Pin92}
J. Pin. On reversible automata. In {\em Proc. of the 1st Latin American Symposium on Theoretical Informatics} (LATIN'92), Lecture Notes in Computer Sceince, vol.583, pp.401--415, 1992.

\bibitem{Raz09}
R. Raz. Quantum Information and the PCP Theorem. {\em Algorithmica}, 55 (2009)  462-489.

\bibitem{TYL10}
K. Tadaki, T. Yamakami, and J. C. H. Lin. Theory of one-tape linear-time Turing machines. {\em Theor. Comput. Sci.}, 411 (2010) 22--43. A preliminary version appeared in {\em Proc. of the 30th SOFSEM Conference on Current Trends in Theory and Practice of Computer Science}, Lecture Notes in Computer Science, Vol.2932, pp.335--348, Springer, 2004.

\bibitem{Thi68}
G. Thierrin. Permutation automata. {\em Math. Systems Thoery}, 2 (1968) 83--90.

\bibitem{YFS+12}
A. Yakaryilmaz, R. Freivalds, A. C. C. Say, and R. Agadzanyan. Quantum computation with write-only memory. {\em Nat. Comput.}, 11 (2012) 81--94.

\bibitem{Yam08}
T. Yamakami. Swapping lemmas for regular and context-free languages.
Available at arXiv:0808.4122, 2008.

\bibitem{Yam09}
T. Yamakami. Pseudorandom generators against advised context-free languages. Available at arXiv:0902.2774, 2009.

\bibitem{Yam10}
T. Yamakami. The roles of advice to one-tape linear-time Turing machines and finite automata. {\em Int. J. Found. Comput. Sci.}, 21 (2010) 941--962. An early version appeared in the {\em Proc. of the 20th International Symposium on Algorithms and Computation}, Lecture Notes in Computer Science, Vol.5878, pp.933--942, Springer, 2009.

\bibitem{Yam11}
T. Yamakami. Immunity and pseudorandomness of context-free languages. {\em Theor. Comput. Sci.}, 412 (2011) 6432--6450.

\bibitem{Yam14}
T. Yamakami. Constant-space quantum interactive proofs against multiple provers. {\em Inf. Process. Lett.}, 114 (2014) 611--619.

\bibitem{YKI05}
T. Yamasaki, H. Kobayashi, and H. Imai. Quantum versus deterministic
counter automata. {\em Theor. Comput. Sci.},  334 (2005)  275--297.

\end{thebibliography}

\end{document}